\documentclass[11pt,letterpaper]{article}

\usepackage{ldlr-sq}

\author{
\normalsize Matthew Brennan\thanks{MIT, \textit{brennanm@mit.edu}. Supported by MIT-IBM Watson AI Lab, NSF Career Award CCF-1940205, and ONR N00014-17-1-2147.} \and
\normalsize Guy Bresler\thanks{MIT, \textit{guy@mit.edu}. Supported by MIT-IBM Watson AI Lab, NSF Career Award CCF-1940205, and ONR N00014-17-1-2147.} \and
\normalsize Samuel B. Hopkins\thanks{UC Berkeley, \textit{hopkins@berkeley.edu}. Supported by a Miller Postdoctoral Fellowship.} \and
\normalsize Jerry Li\thanks{Microsoft Research, \textit{jerrl@microsoft.com}.} \and
\normalsize Tselil Schramm\thanks{Stanford University, \textit{tselil@stanford.edu}. Part of this work was done while virtually visiting the Microsoft Research Machine Learning and Optimization group.}}

\begin{document}

\title{Statistical Query Algorithms and Low-Degree Tests\\ Are Almost Equivalent}

\date{\today}
\maketitle
\begin{abstract}
Researchers currently use a number of approaches to predict and substantiate information-computation gaps in high-dimensional statistical estimation problems.
A prominent approach is to characterize the limits of restricted models of computation, which on the one hand yields strong computational lower bounds for powerful classes of algorithms and on the other hand helps guide the development of efficient algorithms.
In this paper, we study two of the most popular restricted computational models, the statistical query framework and low-degree polynomials, in the context of high-dimensional hypothesis testing.
Our main result is that under mild conditions on the testing problem, the two classes of algorithms are essentially equivalent in power.
As corollaries, we obtain new statistical query lower bounds for sparse PCA, tensor PCA and several variants of the planted clique problem.
\end{abstract}
\vfill
Accepted for presentation at the Conference on Learning Theory (COLT) 2021.
\thispagestyle{empty}
\clearpage

\tableofcontents
\addtocontents{toc}{\protect\thispagestyle{empty}}
\thispagestyle{empty}
\clearpage
\setcounter{page}{1}

%!TEX root = main.tex

\section{Introduction}

Information-computation tradeoffs are ubiquitous in high dimensional statistics.
As the amount and quality of the data increase, inference and estimation tasks often require fewer {\em computational resources}, creating an \emph{information-computation gap} between the signal-to-noise ratios at which the problem is information-theoretically solvable and at which computationally efficient algorithms are known.
This phenomenon is widespread, appearing in estimation of a sparse vector from linear observations, low-rank matrix estimation, sparse principal component analysis, subgraph recovery, random constraint satisfaction, dictionary learning, tensor completion, covariance estimation, phase retrieval, graph matching, and well beyond (c.f., \cite{donoho2006compressed, candes2006stable,feng1996spectrum,candes2007dantzig,lustig2007sparse,recht2010guaranteed,jain2013low,chai2010array, ranieri2013phase,jaganathan2013sparse,candes2013phaselift,arias2014community,arias2012detection,montanari2015finding,feige2002relations,johnstone2009consistency,berthet2013optimal,rubinstein2010dictionaries,spielman2012exact,friedman2008sparse}).
Tradeoffs between computational resources and statistical accuracy are also widely observed empirically in machine learning: both increasing model size and using more iterations of gradient descent to fit models to training data often improve generalization \cite{ji2018risk,soudry2018implicit,nakkiran2019deep,kaplan2020scaling}.
However, we lack a comprehensive theory that explains or predicts information-computation gaps.

In classical complexity theory, computational (in)tractability is explained by organizing problems into equivalence classes via efficient reductions.
While this approach has strong merits, it is challenging to carry out in statistical settings 
(as discussed at length in \cite{brennan2020reducibility}).
Despite recent advances (e.g. \cite{berthet2013complexity,
 ma2015computational,
 hajek2015computational, 
brennan18,zhang2018tensor,
brennan2019optimal,brennan2019universality,luo2020tensor,brennan2020reducibility}), it's too early to tell whether a complete theory of information-computation gaps based on reductions is possible.

Currently, the predominant form of rigorous evidence for information-computation gaps is lower bounds against restricted models of computation.
Here, the goal is to characterize the signal-to-noise ratio needed by specific algorithms for estimation tasks, sometimes taking this as a proxy for the signal-to-noise ratio required by polynomial time algorithms more generally.
So far, such lower bounds have typically been proved separately for each statistical estimation problem, for each distribution over data, and for each model of computation.
For instance, consider the \emph{planted clique} problem, where the goal is to find a clique of size $k$ placed at random in random graph on $n$ vertices.
The problem is solvable by exhaustive search for $k\gg \log n$, but all known polynomial-time algorithms require $k=\Omega(\sqrt{n})$;
the \emph{planted clique conjecture} postulates that the problem is computationally hard if $k= o(\sqrt{n})$.
The foundational work \cite{jerrum1992large} showed lower bounds for Markov-Chain Monte-Carlo methods.
\cite{feige2003probable} prove lower bounds against Lov{\'a}sz--Schrijver semidefinite programs, and lower bounds against stronger Sum-of-Squares semidefinite programs were developed later in \cite{barak2019nearly,deshpande2015improved,meka2015sum,hopkins2018integrality}.
\cite{FGRVX} rule out algorithms for a similar problem in the \emph{statistical query} model, while \cite{atserias2018clique,rossman2008constant,rossman2014monotone} study proof and circuit complexity. Most of these lower bounds rule out algorithms for any $k=o(\sqrt{n})$.

Taken together, these works constitute some evidence for the planted clique conjecture.
However, the proliferation of lower bounds suggests a need for unifying principles, especially because this story is repeated for numerous statistical estimation problems: lower bounds against a variety of restricted computational models are proven independently, all usually pointing to the same signal-to-noise ratios tolerated by efficient algorithms.
This appears to be a miracle: why, for so many distinct problems, should so many restricted computational models point to the same signal-to-noise thresholds for efficient algorithms? (E.g., $k \geq \Omega(\sqrt n)$ for planted clique.)
We ask:
\begin{quote}
\emph{Are some or all of these restricted models equivalent in power? Do lower bounds in some models imply lower bounds in others?}
\end{quote}
If a single class of algorithms were to turn out to be at least as powerful as any of the other popular computational models for an interesting class of statistics problems, then numerous lower bounds could be replaced with a \emph{single} bound.
One might hope to achieve this objective by giving reductions {\em between computational models}, establishing a hierarchy among them and quelling the proliferation of lower bounds.

In this paper, we make a small step towards this goal.
Under mild conditions, we establish the equivalence of two popular frameworks for lower bounds on restricted models of computation for high-dimensional hypothesis testing: \emph{statistical dimension} and \emph{low-degree polynomials}.
Statistical dimension is closely related to \emph{statistical query (SQ) algorithms}, and our results also show that algorithms based on low-degree polynomials are at least as powerful as SQ algorithms.

\subsection{Hypothesis Testing and Models of Computation}
\label{sec:intro-defs}

\paragraph{Hypothesis Testing.}
We consider simple-versus-simple hypothesis testing problems in which we have one {\em null} distribution $D_{\emptyset}$ over $\R^n$, and a family of {\em alternative} distributions $\calS = \{D_u\}_{u \in S}$ over the same space, with a \emph{prior distribution} $\mu$ on $S$.

Under the null hypothesis $H_0$ we are given samples $x_1,\ldots,x_m \in \R^n$ generated independently according to $D_{\emptyset}$, whereas under the alternative hypothesis $H_1$ the samples are instead generated according to $D_u$ for $u \sim \mu$ (we often write $u \sim \calS$).
The objective is to determine which hypothesis is correct.
One example is the \emph{sparse principal component analysis} problem (sparse PCA), where $D_\emptyset = \calN(0,\Id_{n})$, $\calS = \{D_u\}$ where for each $u \in \R^n$ with $\|u\| =1$ and $\rho n$ nonzero entries, $D_u = \calN(0, \Id_{n} + 0.1 uu^\top)$, and $\mu$ taken uniform over $\calS$---here, the testing problem amounts to detecting the presence of the sparse rank-one spike.\footnote{As we discuss below, this problem is unlike planted clique in that the \emph{number of samples} rather than the \emph{signal per sample} governs information-theoretic and computational complexity.}

Testing problems are of great interest in their own right; moreover, to give a lower bound for an estimation problem, it is often sufficient to show that a related hypothesis testing problem is hard (see, e.g., \cite{brennan2020reducibility} -- estimation and testing are related similarly to search and decision in worst-case complexity).

Since we study a model of computation (low degree polynomials) which most naturally outputs real rather than Boolean values, we will use the following notion of a successful test between $H_0,H_1$.

\begin{definition}[$\beta$-distinguisher]\label{def:dist}
We call a function $p \, : \, \R^{n \times m} \rightarrow \R$ of $m$ vectors $\mathbf{x} = x_1,\ldots,x_m \in \R^n$ an $m$-sample $\beta$-distinguisher for a testing problem $D_\emptyset$ vs. $\mathcal{S}$ if $\Abs{\E_{\mathbf{x} \sim D_\emptyset} p(\mathbf{x}) - \E_{u \sim S} \E_{\mathbf{x} \sim D_u} p(\mathbf{x})} \geq \beta \cdot \sqrt{  \Var_{\mathbf{x} \sim D_\emptyset} p(\mathbf{x})}$.
If $\beta > 1$, we call $p$ a \emph{good} distinguisher.\footnote{Here, $\beta > 1$ is chosen to guarantee bounded one-sided error under Chebyshev's inequality.}
\end{definition}

\noindent 
A hypothesis test with small probability of error automatically furnishes a good distinguisher. The converse is not necessarily true; though one might naturally try to apply thresholding to a distinguisher to obtain a hypothesis test, a good distinguisher may have large variance under the alternative hypothesis $H_1$, so there is only a one-sided error guarantee. Thus, from the perspective of \emph{lower bounds}, ruling out the existence of a $\beta$-distinguisher in a restricted computational model is at least as strong as ruling out the existence of a small-error hypothesis test (in that model).

\paragraph{Low Degree Polynomials.}
Given $m$ samples $\mathbf{x} = x_1,\ldots,x_m \in \R^n$, our first model of computation is allowed to output the value of any fixed polynomial $p(\mathbf{x})$ of bounded degree, usually constant or logarithmic in $m,n$.
Note that this model allows polynomials in all $m$ samples jointly, \emph{not} just empirical averages over $m$ samples of the form $\tfrac 1 m \sum_{i=1}^m p(x_i)$.

An extraordinary variety of high-dimensional hypothesis testing algorithms boil down to evaluating low-degree polynomials: for example, most spectral algorithms, the method of moments, algorithms based on small-subgraph statistics, and message passing algorithms (see \cite{BKW19,hopkinsThesis}).
And, although faster implementations are often possible, any degree-$k$ polynomial can be evaluated in time $(nm)^{O(k)}$ by evaluating all monomials.

A recent line of work characterizes the limitations of such algorithms by ruling out the existence of low-degree distinguishers: such lower bounds are now known in the computationally-hard regimes of planted clique \cite{barak2019nearly}, stochastic block model \cite{hopkins2017efficient,BBKMW20}, sparse principal component analysis \cite{ding2019subexponential}, tensor principal component analysis \cite{BKW19}, and more.
Remarkably, excluding problems with unusual algebraic structure \cite{holmgren2020counterexamples}, the (non)existence of a low-degree distinguisher closely tracks the (non)existence of \emph{any known poly-time hypothesis test}.

\paragraph{Statistical Queries and Statistical Dimension.} Our second model of computation is the {\em statistical query} (SQ) model $\VSTAT(m)$.
$\VSTAT(m)$ algorithms access a distribution $D$ over $\R^n$ via \emph{queries} $\phi:\R^n \to [0,1]$ to an oracle.
For each query $\phi$, the oracle returns $\E_{x \sim D} \phi(x) + \zeta$, for an adversarially chosen $\zeta\in \R$ with $|\zeta| \le \max(\frac{1}{m}, \sqrt{\frac{\E[\phi](1-\E[\phi])}{m}})$.
This approximates $\E_D \phi$ with the same accuracy as an $m$-sample empirical estimate under the guarantees of Bernstein's inequality.

The SQ model was first proposed as a framework for designing noise-tolerant algorithms \cite{kearns}, and is a popular restricted model of computation for studying information-computation tradeoffs (see e.g. \cite{FGRVX,feldman2018complexity,diakonikolas2017statistical}, as well as numerous supervised learning problems).
An algorithm which makes $q$ queries to $\VSTAT(m)$ is a proxy for an algorithm running in time $q$ on $m$ samples, albeit an imperfect one, since (1) the queries $\phi$ need not be polynomial-time computable, and (2) each query $\phi$ is permitted to be a function of only a single sample (whereas a general polynomial time algorithm may be allowed to, for instance, compare pairs of samples).

We will treat the SQ model via \emph{statistical dimension}, a complexity measure on hypothesis testing problems which implies lower bounds against SQ algorithms.
Most existing SQ lower bounds are proved by analyzing one of a few possible notions of statistical dimension.
We use a mild strengthening of the statistical dimension introduced by \cite{FGRVX}.\footnote{We remark on technical differences between our setup and that of \cite{FGRVX} in Appendices~\ref{sec:sda-counterex} and \ref{app:mv1}.}

\begin{definition}[Statistical Dimension]
\label{def:sda}
  Let $D_\emptyset$ vs. $\calS$ be a testing problem with prior $\mu$.
  For $D_u\in \calS$, define the relative density $\ol D_u(x) = \frac{D_u(x)}{D_\emptyset(x)}$, and the inner product $\iprod{f, g} = \E_{x \sim D_\emptyset} f(x)g(x)$.
  The {\em statistical dimension} $\SDA(\calS,\mu, m)$ measures tails of $\iprod{\ol D_u, \ol D_v} - 1$ with $u,v$ drawn independently from $\mu$.
  \[
  \SDA(\calS,\mu, m) = \max \left \{ q \in \N \, : \, \E_{u,v \sim \mu} \left [ \left|\iprod{\ol D_u ,\ol D_v} - 1\right| \, | \, A \right ] \leq \tfrac 1 m \text{ for all events $A$ s.t. } \Pr_{u,v \sim \mu}(A) \geq \tfrac 1 {q^2} \right \}\mper
  \]
Often we will write $\SDA(m)$ or $\SDA(\calS,m)$ when $\calS$ and/or $\mu$ are clear from context.
\end{definition}

\noindent We offer some intuition about the definition, which may be opaque at first.
The quantity $\langle \ol D_u, \ol D_v \rangle - 1$ is equivalent to $\E_{x \sim D_u} \frac{\Pr_{D_v}[x]}{\Pr_{D_{\emptyset}}[x]} - 1$; that is, the centered average of the likelihood ratio of $D_v$ to $D_{\emptyset}$ over samples from $D_u$. When this quantity is at least $\delta$, $D_u$ and $D_v$ may have common events that allow one to distinguish them both from $D_{\emptyset}$ with probability $\delta'$. The statistical dimension quantifies the measure of pairs of distributions (according to $\mu$) with no such common events.

In \cite{FGRVX}, it is shown that the statistical dimension is a lower bound on the query complexity of hypothesis testing with a $\VSTAT$ oracle:\footnote{We extend their result to our notion of SDA  via a near-identical argument in Appendix~\ref{app:mv1}.}
\begin{theorem}[Theorem 2.7 of \cite{FGRVX}]\label{thm:sda}
Let $D_{\emptyset}$ be a null distribution and $\calS$ be a set of alternate distributions over $\R^n$.
Then any (randomized) statistical query algorithm which solves the hypothesis testing problem of $D_{\emptyset}$ vs. $\calS$ with probability at least $(1-\delta)$ requires at least $(1-\delta)\SDA(\calS, m)$ queries to $\VSTAT(m/3)$ (corresponding to $m/3$ samples).
\end{theorem}

\subsection{Our Results}

\label{sec:comparing}
Our main result is a surprisingly tight equivalence, under mild conditions, between statistical dimension and the minimum degree of any good distinguisher.

Summarizing the discussion of running times and sample complexities above, we might hope to equate $m$-sample distinguishers of degree $k$ (which can be evaluated in time $(nm)^{O(k)}$) with $2^{O(k)}$-query $\VSTAT(m)$ algorithms.
To understand the conditions under which this is possible, we first observe that planted clique already furnishes a counterexample -- a case where a single-query SQ algorithm exists but there is no corresponding low-degree distinguisher.
Concretely, to detect a $k$-clique planted in a graph $G$ from $G(n,1/2)$, for any $k \gg \log n$ it suffices to make the single query $\phi(G) = \Ind(G \text{ contains a $k$-clique})$ to $\VSTAT(4)$.
By contrast, it is known that no degree $o(\log^2 n)$ polynomial successfully distinguishes for any $k < n^{1/2 -\eps}$ \cite{barak2019nearly}.

The issue here is that there is a high-degree function of a single sample which solves planted clique -- that function can be used as a statistical query.
As a condition for equivalence between statistical dimension and low-degree distinguishers, therefore, we must insist that such high-degree one-sample distinguishers do not exist.
Our main theorem applies under the following \emph{niceness} condition, which asks for just slightly more: no high degree function of a very small number of samples is a nontrivial distinguisher.

While niceness rules out problems like planted clique (which is what we want), we will see that it allows ``many-sample'' problems such as sparse PCA -- precisely the type of problems for which the SQ model can capture interesting information-computation gaps.
After our main theorem statement (Remark~\ref{rem:multi-sample}) we describe a principled approach to transform one-shot problems like planted clique into many-sample problems, so that they can also be studied with our techniques.

\begin{definition}[$(\delta,k)$-nice]
Fix a null distribution $D_\emptyset$ on $\R^N$.
Call a function $p \, : \, \R^{N \times k} \rightarrow \R$ of $k$ vectors $x_1,\ldots,x_k \in \R^N$ \emph{$k$-purely high degree} if it is orthogonal to all functions $f(x_1,\ldots,x_k)$ which have degree at most $k$ in one of $x_1,\ldots,x_k$ -- that is, $\E_{x_1,\ldots,x_k \sim D_\emptyset} p(x_1,\ldots,x_k)f(x_1,\ldots,x_k) = 0$ for all such $f$.
The testing problem $D_\emptyset, \{D_u\}_{u \in \calS}$ is $(\delta,k)$-nice if no $k$-purely high-degree function of $k$ samples is a $\delta$-distinguisher.
\end{definition}

We emphasize that $(\delta,k)$-niceness concerns hardness of a testing problem when given \emph{very few} samples -- we typically think of $k = O(1)$ or $k = \polylog N$.
We will show that almost any reasonable multi-sample testing problem which is not too easy to solve with $k$ samples becomes nice after the addition of a small amount of noise.
The following is stated for a \emph{coordinate-wise resampling} noise process -- it follows from standard arguments about noise operators and high-degree functions.
In Section~\ref{sec:noise-robust} we give versions allowing a broad class of noise processes (additive Gaussian noise, random restriction, etc.).
\begin{fact}[See Theorem~\ref{thm:noisy}]
\label{thm:noise-intro}
  Let $\calS = \{D_u\}, D_\emptyset$ be a testing problem on $\R^N$ and suppose that $D_\emptyset = D^{\otimes N}$ is a product distribution.
  Let $k \in \N$ and suppose that $\calS, D_\emptyset$ does not have a $k$-sample $C$-distinguisher.
  Let $\calS' = \{D_u'\}$, where to sample $x' \sim D_u'$ we first sample $x \sim D_u$ and then each coordinate $x_i$ is independently replaced with a fresh sample from $D$ with probability $\rho \in [0,1]$.
  Then $D_\emptyset$ versus $\calS'$ is $(C (1-\rho)^{k^2} ,k)$-nice.
\end{fact}

Many natural high-dimensional hypothesis testing problems are robust to noise (including the main examples we have mentioned so far), and remain qualitatively unchanged by the addition of some form of noise captured by our theorems.
The typical effect is a small decrease in the signal-to-noise ratio in each sample.
In typical applications, $C = O(1)$, and when working with $m$ samples we will want roughly $(m^{-k/2},k)$-niceness, which we can achieve by taking $k \approx \log m$ and $\rho$ a small constant, so that $\calS$ and $\calS'$ are very similar.
In this case, our main theorem will lead to $(\log m)^2$-degree distinguishers, whereas brute-force algorithms would correspond to degree $Nm \gg (\log m)^2$ -- with more refined definitions later on, in many cases (e.g. Planted Clique) we can avoid the logarithmic loss and replace $(\log m)^2$ with $\log m$.

\paragraph{Main Theorem.}
\label{sec:main-meta-intro}

We turn to our main theorem.
On first reading we suggest the interpretation that $m' = m$ and $k$ is constant or logarithmic in $m$.

\begin{theorem}[Main Theorem, see Theorem~\ref{thm:low-deg-sq} and Theorem~\ref{thm:converse}]\label{thm:main-intro}
  Let $D_\emptyset$ vs. $\mathcal{S}$ be an $(m^{-k/2}/4,k)$-nice testing problem on $\R^N$ for some even $k > 0$.
  \begin{compactenum}
  \item If there is some $0 \le m' \leq m$ such that $\SDA(\calS, m') \leq \Paren{\frac {2m} {m'}}^{k/2}$ (in particular, if there is an SQ algorithm making $o(2^{k/2})$ queries to $\VSTAT(m/3)$), then there is a good $4mk$-sample distinguisher $p$ which has degree $d \le k^2$,\footnote{As mentioned above, Theorems~\ref{thm:low-deg-sq} and~\ref{thm:converse} are stated in terms of a more refined notion of degree (defined in Section~\ref{sec:prelims}) which allows us in many cases to improve the bound to $d \le O(k)$, which is the best we can hope for.} and
  \item if there is a degree $k$ function $p$ which is a good $m$-sample distinguisher, then there exists $m' \leq m$ such that $\SDA(\calS, m') \leq \Paren{\frac{2m}{m'}}^{O(k)}$ (e.g. $\SDA(\calS, m) \leq 2^{O(k)}$).
  \end{compactenum}
\end{theorem}

Using Fact~\ref{thm:noise-intro}, we already see that Theorem~\ref{thm:main-intro} applies to any noisy testing problem.
Even without adding noise, our next theorem shows that the guarantees of Theorem~\ref{thm:main-intro} apply to some problems with additional structure -- for instance, if $D_\emptyset$ and the $D_u$'s are all product distributions.
(This is the case even though such problems may not be nice; we are still able to apply a variant of the proof of Theorem~\ref{thm:main-intro}.)
This leads to slightly tighter results, especially for problems where the difference between degree $\log m$ and $\poly(\log m)$ distinguishers is important.

\begin{theorem}[Gaussian or Independent Coordinates, see Theorems~\ref{thm:gauss} \& \ref{thm:prod}]
\label{thm:indep-gaussian-intro}
  Let $\calS = \{D_u\}, D_\emptyset$ be a testing problem on $\R^N$ with one of the following structures:
  \begin{compactitem}
  \item $D_\emptyset = \mathcal{N}(0,\Id_N)$ is the standard Gaussian distribution and each $D_u = \mathcal{N}(u, \Id_N)$ for some vector $u \in \R^N$
  \item $D_\emptyset$ and all $D_u$ are product measures on $\{\pm 1\}^N$
  \end{compactitem}
  Let $m,k \in \N$ with $k \ll m$ and suppose that $\calS,D_\emptyset$ has no $k$-sample $2^k$-distinguisher.
  Then the conclusion of Theorem~\ref{thm:main-intro} holds for $\calS$ (with the upper bound on $d$ in part 1 replaced by $d \le O(k)$).
\end{theorem}
Even with the additional requirements, Theorem~\ref{thm:indep-gaussian-intro} captures numerous interesting problems -- spiked matrix and tensor models, variants of random constraint satisfaction and linear equations, community detection, and beyond.

\begin{remark}[Simulation Arguments Are Lossy]
A natural approach to prove a theorem like Theorem~\ref{thm:main-intro} would be to na\"ively simulate SQ algorithms by low-degree distinguishers and vice versa.
However, direct simulation arguments that we are aware of (for instance, taking each monomial in a low-degree distinguisher to be an SQ query) at best relate $\SDA(\calS, m)$ to low-degree distinguishers on $\poly(m)$ samples (or vice versa).
By contrast, Theorem~\ref{thm:main-intro} translates between $\SDA(\calS, m)$ and low-degree distinguishers on approximately $m$ samples -- this is crucial for most applications, where information-computation gaps occur on the scale of $m$ versus $\poly(m)$ samples.

We remark as well that the statistical dimension is a lower bound on the SQ complexity, but does not always offer a tight characterization. There are problems for which polynomial-query $\VSTAT$ SQ algorithms require polynomially more samples than suggested by the statistical dimension, for example, in random constraint satisfaction problems \cite{feldman2018complexity}.
Hence, sometimes a low-degree distinguishers may exist for $m$ samples even if no polynomial-query $\VSTAT(m)$ algorithms exist, and as a consequence simulation arguments will not tightly characterize the existence of low-degree distinguishers.

Our proof of Theorem~\ref{thm:main-intro} directly relates statistical dimension to the minimum degree of a distinguisher, without a simulation argument.
We also give (Appendix~\ref{sec:conversion}) a different proof of a slightly weaker version of part 1 of Theorem~\ref{thm:main-intro},\footnote{The quantitative bounds we obtain are identical to Theorem~\ref{thm:main-intro}; the theorem is weaker because the existence of a $\VSTAT$ algorithm is a stronger assumption than an upper bound on the statistical dimension.}
 which is based on a simulation-style argument (though it has a non-constructive component) of an algorithm making calls to $\VSTAT$ via a low-degree distinguisher without $\poly(m)$ losses.
\end{remark}

\begin{remark}[One-Shot Versus Multi-Sample Problems] \label{rem:multi-sample}
Theorem~\ref{thm:main-intro} only applies to nice testing problems.
In particular, niceness rules out many ``one-shot'' problems which are information-theoretically easy to solve with a single sample, such as the usual formulation of planted clique, where the SQ model does not make sense -- the model originates in PAC learning, where having many independent samples is fundamental.
By contrast, low-degree tests can still be formulated for one-shot problems.

To give evidence of hardness for a one-shot problem in the SQ framework, one must first formulate a multi-sample version.
For instance, the SQ lower bounds of \cite{FGRVX} for planted clique treat a ``bipartite'' version where each sample is the adjacency list of a node in a bipartite graph.
These multi-sample formulations are often \emph{ad hoc}, which is problematic, as \emph{the choice of multi-sample version can significantly affect the resulting statistical query complexity}!

Based on Theorem~\ref{thm:main-intro}, we propose a canonical approach to translate one-shot problems into nice many-sample problems: decrease the per-sample signal-to-noise ratio (e.g., clique size versus graph density in planted clique) until the resulting problem is information-theoretically unsolvable given $O(1)$ independent samples, while simultaneously increasing the number of samples appropriately.
For example, in a Gaussian model, one sample from $\calN(u,\Id)$ is equivalent to $m$ samples from $\calN(\frac{1}{\sqrt{m}} u, \Id)$.
In numerous cases -- additive Gaussian models and planted clique, for example -- this yields problems which are \emph{polynomial-time equivalent} to the underlying one-shot problem (see Section~\ref{s:cloning}).
For an illustration, see the Tensor PCA problem discussed in and above Corollary \ref{cor:tpca-intro}.
\end{remark}

\subsubsection{Overview of Techniques}
\label{sec:tech-over}

\paragraph{Proof Sketch of Theorem~\ref{thm:main-intro}.}
We outline the proof of case (1) of our main theorem; case (2) follows a similar argument in reverse.
We argue contrapositively, starting with the hypothesis that there is no good degree $k^2$ $m$-sample distinguisher.
For this sketch, we ignore the case $m' < m$ and consider the goal of proving a lower bound on the statistical dimension $\SDA(\calS, m)$.
Unpacking the definition of $\SDA$, this amounts to the tail bound $\E_{u,v \sim \calS} [|\iprod{\ol D_u, \ol D_v} - 1|  \mid A] \lesssim 1/m$ for any event $A$ of probability roughly $2^{-k}$.
This tail bound will be implied by an upper bound on the $k$-th moment -- our goal will be to show $\E_{u,v \sim \calS} (\iprod{\ol D_u, \ol D_v} -1)^k \lesssim m^{-k}$.

Simple manipulations (which rely on the independence of the samples) show that the maximum value of $\alpha$ such that there is a $k$-sample $\alpha$-distinguisher is given by the related quantity $ \alpha = \sqrt{\E_{u,v\sim \calS} \langle \ol D_u, \ol D_v \rangle^k - 1}$.
To see why, recall that a $k$-sample $\beta$-distinguisher is a function of $k$ samples, $p(x_1,\ldots,x_k)$ that satisfies $\beta \cdot (\Var_{D_\emptyset^{\otimes k}} p)^{1/2}\le |\E_{u \sim \calS}\E_{D_u^{\otimes k}} p - \E_{D_\emptyset^{\otimes k}} p| = \left|\langle p, \E_u \ol D_u^{\otimes k} - 1 \rangle_{D_\emptyset^{\otimes k}}\right|$.\footnote{Here we have used the notation that for a distribution $D$, $\langle f,g \rangle_D = \E_{x \sim D} f(x) g(x)$ and $D^{\otimes k}$ is the joint distribution of $k$ random samples from $D$, and for a function $f(x)$, $f^{\otimes k}(x_1,\ldots,x_k) = \prod_{i = 1}^k f(x_i)$.}
By rescaling we may without loss of generality consider $p$ with $\Var_{D_\emptyset^{\otimes k}} p =\langle p,p\rangle_{D_\emptyset^{\otimes k}}= 1$.
So now by Cauchy-Schwarz and by the independence of the samples,
\[
\beta
= \left|\Big\langle p, \E_{u \sim \calS} \ol D_u^{\otimes k} - 1 \Big\rangle_{D_\emptyset^{\otimes k}}\right|
\le \sqrt{\E_{u,v \sim \calS}\left\langle \ol D_u^{\otimes k} - 1, \ol D_v^{\otimes k} - 1 \right\rangle_{D_\emptyset^{\otimes k}}}
= \sqrt{\E_{u,v \sim \calS} \langle \ol D_u, \ol D_v \rangle_{D_\emptyset}^k - 1},
\]
where in the final step we have used that $\ol D_u^{\otimes k}$ is a density and so $\langle \ol D_u^{\otimes k}, 1 \rangle = 1$, as well as the independence of the samples.
By choosing the $p$ for which the Cauchy-Schwarz is tight, we have our conclusion.

Thus, pretending for the sake of this overview that the $k$-th moment $\E_{u,v} (\langle \ol D_u, \ol D_v \rangle -1)^k \approx \E_{u,v} \langle \ol D_u, \ol D_v \rangle^k -1$,
to show that $\E_{u,v \sim \cal S} ( \langle \ol D_u, \ol D_v \rangle - 1 )^k \lesssim m^{-k}$, it suffices for us to rule out $k$-sample $m^{-k/2}$-distinguishers.
Since by assumption $D_\emptyset$ versus $\calS$ is $(m^{-k/2},k)$ nice, such a distinguisher could not be $k$-purely high degree.
Via a careful application of H\"older's inequality (Lemma~\ref{lem:holder}), we are able to show that it suffices to consider only functions of \emph{purely high degree} or \emph{purely low degree}.
The main challenge is now to rule out a \emph{low-degree} $k$-sample $m^{-k/2}$ distinguisher -- that is, we need to show that every function $p(x_1,\ldots,x_k)$ with degree at most $k$ in each sample $x_i$ has
\begin{equation}\label{eq:intro-1}
\Abs{ \E_{u \sim \calS} \E_{D_u^{\otimes k}} p - \E_{D_\emptyset^{\otimes k}} p } \lesssim m^{-k/2} \sqrt{\Var_{D_\emptyset^{\otimes k}} p }\mper
\end{equation}

Since we are analyzing $k$-sample distinguishers, it is not \emph{a priori} clear how such a $1/\poly(m)$ bound on the distinguishing power can appear, especially given that $m \gg k$.
Our key insight is that this strong quantitative bound follows from the assumption that there is no good degree-$k^2$ $m$-sample distinguisher:

\begin{lemma}[Key Lemma, Informal -- see Claim~\ref{lem:many-to-one}, Lemma~\ref{lem:boosting}]
If there is no good $m$-sample degree-$k^2$ distinguisher for the testing problem $D_\emptyset$ versus $\calS$, then no function $p(x_1,\ldots,x_k)$ with degree at most $k$ in each sample is an $m^{-k/2}$-distinguisher.
\end{lemma}

Once the (very careful) setup is in place, this lemma follows from elementary Fourier analysis, exploiting independence of samples.
Nonetheless, we find it striking that a relatively mild assumption on the distinguishing power of low degree polynomials of $m$ samples can be boosted into a strong quantitative bound on the distinguishing power of low degree polynomials of $k \ll m$ samples.
This lemma leads to (\ref{eq:intro-1}), finishing the proof.

\paragraph{Niceness of Noise-Robust Problems.}
To show that noise-robust testing problems satisfy the niceness criterion (Fact~\ref{thm:noise-intro} and its generalizations in Section~\ref{sec:noise-robust}), we again use Fourier Analysis; for some types of noise our arguments are entirely standard, exploiting the attenuation of high-degree functions under i.i.d. noise. 
We also allow for noise processes which make sense for problems with combinatorial structure which would be adversely affected by i.i.d. coordinate-wise noise (e.g. hypergraph planted clique) -- showing that these also lead to nice testing problems uses similar ideas but requires more care.

\paragraph{Avoiding Niceness for Product and Gaussian Distributions.}
Finally, we overview the proof of Theorem~\ref{thm:indep-gaussian-intro}.
We need to avoid the use of the niceness assumption that we described in the overview above of the proof of Theorem~\ref{thm:main-intro}.
That is, we need a different way to rule out high-degree $k$-sample $m^{-k/2}$-distinguishers.
Roughly speaking, we show that under either the product or Gaussian assumptions, a high-degree $k$-sample $\alpha$-distinguisher cannot exist unless a low-degree one does -- then we follow the argument above to rule out low-degree $k$-sample $m^{-k/2}$ distinguishers.
This argument turns on the fact that, for Gaussian and product distributions, high-degree moments are simple functions of low-degree moments.
(See Lemmas~\ref{lem:gauss} and \ref{lem:prod} for the details.)

\subsubsection{Applications: New Information-Computation Lower Bounds ``For Free''}
\label{sec:applications-intro}

We use our equivalence theorems to obtain new information-computation lower bounds for a number of testing problems. 
We obtain new lower bounds against SQ algorithms for tensor PCA (Corollary~\ref{cor:tpca-main}), (Hypergraph) Planted Clique and Planted Dense Subgraph (\ref{cor:mshpc}), and sparse PCA (\ref{cor:wishart}), and we obtain new lower bounds against low-degree distinguishers for Gaussian mixture models (\ref{cor:gmm-ldlr}) and Gaussian Graphical Models (\ref{cor:ggm}).
Our bounds are obtained essentially ``for free'' by starting with known SDA or degree lower bounds, then applying Theorem~\ref{thm:main-intro} and its derivatives. (One exception is the Gaussian Graphical Models bound, for which we prove an SQ lower bound from scratch. Interestingly, for this problem, it seems easier to prove SDA lower bounds than degree lower bounds.)

In the case of planted clique, in addition to capturing the ``bipartite'' model of \cite{FGRVX}, we also prove lower bounds for a new multi-sample version, in which we receive $m$ independent copies of the adjacency matrix of $G(n,p^{1/m})$ or $G(n,p^{1/m})$ with the same planted $k$-clique.
We show in Lemma~\ref{lem:p-clone} that our version is {\em information-theoretically} and {\em computationally equivalent} to the standard version of planted clique (albeit with slightly higher-than-usual edge density $p > 1/2$), a property not shared by the bipartite model.
This is an example of our approach to transforming one-sample problems into many-sample ones by weakening the per-sample signal-to-noise ratio.

For the sake of illustration, we state our result for Tensor PCA here, and defer formal statements of our lower bounds for the other problems to Section~\ref{sec:examples}.
Tensor PCA is a well-studied higher-order generalization of the principal components analysis problem (see e.g. \cite{richard2014statistical,hopkins2015tensor,lesieur2017statistical,wein2019kikuchi,arous2020algorithmic}). 
It is typically stated as a ``one-shot'' problem: distinguish a $3$-tensor $G$ with i.i.d. entries from $\calN(0,1)$ from a planted tensor of the form $G + \lambda u^{\tensor 3}$, where $G$ is as before, $\lambda > 0$, and $u$ is a unit vector.
In Lemma~\ref{lem:g-clone} we show that this problem is in fact \emph{equivalent} (both statistically and computationally) to the following $m$-sample problem: distinguish between i.i.d. $G_1,\ldots,G_m$ and $G_1 + \frac{\lambda}{\sqrt{m}} u^{\tensor 3},\ldots,G_m + \frac{\lambda}{\sqrt{m}} u^{\tensor 3}$.

By combining known bounds against low-degree distinguishers \cite{hopkins2017power,BKW19} with Theorem~\ref{thm:main-intro}, we obtain a new SQ lower bound against the multi-sample version of Tensor PCA:
\begin{corollary}[SQ lower bound for Tensor PCA (special case of Corollary~\ref{cor:tpca-main})]\label{cor:tpca-intro}
\label{cor:tpca-intro}
  Let $D_\emptyset = \calN(0,\Id_{n^3})$ and for unit $u \in \R^n$ let $D_u = \calN(u^{\otimes 3}, \Id_{n^3})$.
  Let $\calS$ be the uniform distribution on $\{D_u\}_{u \in \{\pm 1/\sqrt{n}\}^n}$.
  Any SQ algorithm solving the testing problem $\calS$ versus $D_\emptyset$ requires at least $n^{\omega(1)}$ queries to $\VSTAT( n^{3/2} / (\log n)^{O(1)})$.
\end{corollary}
Up to logarithmic factors, this SQ lower bound matches the best known polynomial-time algorithms, which require at least $m \geq \Omega(n^{3/2})$ samples (or, for the one-shot problem, $\lambda \geq \Omega(n^{3/4})$) \cite{hopkins2015tensor}. 
We discuss the information-computation tradeoff in greater detail in Section~\ref{sec:tpca}.
We note that similar bounds for tensor PCA were obtained concurrently and independently in \cite{DH20}.

\subsection{Prior Work}\label{sec:prior}
Researchers have long been aware of the information-computation gap phenomenon, with early work showing such gaps in artificially constructed learning problems \cite{decatur2000computational, servedio1999computational,shalev2012using} and more recent work focusing on algorithms that trade off between statistical and computational efficiency \cite{shalev2008svm,balakrishnan2011statistical,shalev2012using,chandrasekaran2013computational,chen2016statistical}.
Our goal here is to establish an equivalence between large classes of algorithms for a wide range of problems in high-dimensional statistics -- low-degree distinguishers and SQ algorithms.
Several prior works have a similar theme: in related contexts, \cite{hopkins2017power} shows that Sum-of-Squares semidefinite programs are no more powerful than a restricted class of spectral algorithms\footnote{This class of spectral algorithms, to our knowledge, is {\em not} captured by low-degree distinguishers.} for hypothesis testing, and
\cite{feldman2017statistical} shows that a restricted class of convex programs is captured by SQ algorithms.

Several related lines of work establish \emph{algorithm-independent} or \emph{structural} properties of high dimensional statistics problems which imply hardness results against restricted models of computation -- statistical dimension being one example.
Other examples come from statistical physics, where \emph{overlap gaps} and, more generally, \emph{solution-space geometry} are related to performance of algorithms such as Markov-Chain Monte Carlo and message passing, with early work focusing primarily on random constraint satisfaction \cite{jia2004spin,achlioptas2008algorithmic,ibrahimi2012set}, and more recent work studying other optimization and hypothesis testing problems \cite{gamarnik2014limits,gamarnik2019landscape,gamarnik2020low,arous2020algorithmic,arous2020free,gamarnik2019overlap}.

More broadly, information-computation tradeoffs have been studied in many restricted computational models: e.g. message-passing algorithms (see \cite{mezard2009information,zdeborova_statistical_2016} for overviews; we highlight recent work \cite{wein2019kikuchi} focusing on running time versus information \emph{tradeoffs}), Markov-Chain Monte Carlo (e.g. \cite{jerrum1992large, arous2020algorithmic}), and Sum-of-Squares semidefinite programs (see e.g. \cite{grigoriev2001linear,RRS17,kothari2017sum} or \cite{raghavendra2018high} for a survey).
In our view, charting the formal connections among all these lenses on information-computation tradeoffs -- the statistical physics approach, SQ models, low-degree tests, message-passing algorithms, Markov-Chain Monte Carlo methods, Sum-of-Squares, etc. -- is an excellent direction for future investigation.

\paragraph{Statistical Query Model.}
The SQ model was proposed by Kearns as a framework for designing noise-tolerant algorithms for PAC learning \cite{kearns}.
Blum et al. shortly thereafter introduced statistical query dimension \cite{blum1994weakly} as a framework for proving lower bounds on SQ algorithms for supervised learning.
The SQ framework has since been generalized to hypothesis testing and estimation \cite{FGRVX,feldman2018complexity}.

An advantage of SQ lower bounds is their implications for other algorithms: since many algorithms can be implemented with SQ oracle access, SQ lower bounds immediately imply lower bounds against a number of other algorithms, including some convex programs, gradient descent, and more (see e.g. \cite{feldman2017statistical}).

SQ lower bounds abound in the study of high-dimensional learning -- recent examples are in robust statistics \cite{diakonikolas2017statistical,diakonikolas2019efficient}, polytopes \cite{klivans2007unconditional}, neural nets \cite{goel2020superpolynomial}, and more.
In this work, we derive new SDA lower bounds for sparse PCA and for tensor PCA -- SQ lower bounds for tensor PCA also appear in the concurrent work of \cite{DH20}, who also obtain bounds for estimation.

Statistical dimension may not be a complete characterization of the query complexity in the VSTAT model, in that there are problems for which the statistical dimension is $q$ but we do not know any $q$-query VSTAT algorithms.
A complete characterization is given in \cite{feldman2012complete}.
In light of this, our results equate the power of low-degree distinguishers with a computational model that is \emph{at least as powerful} as VSTAT.
There are a number of other statistical query models for hypothesis testing problems defined in the literature, for example the MVSTAT oracle of \cite{feldman2018complexity}. 
An interesting open problem is whether a more direct equivalence (via simulation argument) can be achieved in an alternative SQ model.

\paragraph{Low-Degree Tests.}
Using low-degree polynomials to prove computational lower bounds is a classical idea in theoretical computer science; see e.g. \cite{beigel1993polynomial} on the polynomial method in circuit complexity.
Their recent study as a restricted model of computation for high-dimensional estimation and hypothesis testing problems emerged implicitly in the literature on Sum-of-Squares lower bounds \cite{barak2019nearly}, then more explicitly in \cite{hopkins2017efficient,hopkins2017power}.
See \cite{BKW19} for a survey.

Recent works prove lower bounds against low-degree tests for the Sherrington-Kirkpatrick spin glass model \cite{bandeira2019computational}, tensor PCA \cite{hopkins2017power}, sparse PCA \cite{hopkins2017power}, planted dense subgraphs \cite{schramm2020computational}, and more.
The lower bound approach has also inspired algorithms, for instance for (mixed-membership) community detection \cite{hopkins2017efficient}, graph matching in correlated Erd\"{o}s-R\'{e}nyi graphs \cite{barak2019nearly}, and sparse PCA \cite{ding2019subexponential}.

\paragraph{Organization.}
Section~\ref{sec:prelims} contains preliminaries; the proofs of parts 1 and 2 of Theorem~\ref{thm:main-intro} follow in Sections~\ref{s:LDHardImpliesSDALarge} and~\ref{sec:converse}.
In Section~\ref{sec:noise-robust} we obtain corollaries for noise robust problems (generalizations of Fact~\ref{thm:noise-intro}) and in Section~\ref{sec:indep} we derive even stronger corollaries for product measures (Theorem~\ref{thm:indep-gaussian-intro}).
Section~\ref{s:cloning} contains a discussion of the cloning methodology for transforming a one-shot problem to an appropriate multi-sample problem for the SQ framework.
Section~\ref{sec:examples} applies our main results to obtain new lower bounds for a number of testing problems.

Appendices~\ref{sec:sda-counterex} and \ref{app:mv1} give some further details on statistical dimension.
Appendix~\ref{sec:conversion} gives an argument showing how VSTAT algorithms can be simulated directly by low-degree distinguishers.
Some calculations are postponed to Appendices~\ref{app:clone} and~\ref{app:ex-cal}.

\section{Preliminaries}\label{sec:prelims}

We study hypothesis testing problems $D_{\emptyset}$ vs. $\calS = \{D_u\}_{u \in \calS}$ with a prior $\mu$ over $\calS$. 
We frequently write $u \sim \calS$ or $u \sim S$ to indicate that $D_u$ is sampled from $\calS$ according to the marginal $\mu$.
We use $\ol D_u$ to refer to the {\em likelihood ratio} or {\em relative density} $\frac{D_u}{D_{\emptyset}}$, where the background measure $D_{\emptyset}$ will be clear from context.
We always assume that the likelihood ratio is finite and that $\E_{x \sim D_\emptyset}( D_u(x) / D_\emptyset(x))^2 < \infty$, for every $D_u$.
This holds if $D_\emptyset, D_u$ have finite support and the support of $D_u$ is contained in that of $D_\emptyset$; it can also be enforced for continuous distributions by mild truncation of tails.

For $\R$-valued functions $f,g$, let the inner product $\langle f,g\rangle_{D_{\emptyset}} = \E_{x \sim D_{\emptyset}} f(x) g(x)$ and the corresponding norm $\|f\|_{D_\emptyset} = \langle f, f \rangle_{D_\emptyset}^{1/2}$.
We drop the subscript $D_{\emptyset}$ when $D_\emptyset$ is clear from context.
Note that always, $\langle \ol D_u, 1 \rangle = 1$.
For a distribution $D$ and an integer $k$, let $D^{\otimes k}$ denote the joint distribution of $k$ independent samples from $D$.
We will often use $\iprod{f^{\otimes k}, g^{\otimes k}}_{D_\emptyset^{\otimes k}} = \iprod{f,g}_{D_\emptyset}^k$, which is a consequence of independence.

For $D_{\emptyset}$ over $\R^n$, $d$ a non-negative integer, and any function $f:\R^n \to \R$, we let $f(x)^{\le d}$ denote the orthogonal (w.r.t. $D_\emptyset$) projection of $f$ to the span of functions of degree at most $d$ in $x$. 
We similarly define $f^{< d}$, $f^{= d}$, $f^{\ge d}$, and $f^{> d}$.
\paragraph{Ruling Out Distinguishers in Subspaces via Small Norms.}

We will repeatedly use the folklore fact that the \emph{optimal} $m$-sample low-degree test for a problem $\calS, D_\emptyset$ has a canonical form: it is the projection of the $m$-sample likelihood ratio $\E_{u \sim \calS} \ol{D}_u^{\otimes m}$ to the span of functions of low degree.
In fact, a more general statement is true (which we have essentially proved in Section~\ref{sec:tech-over}):

\begin{fact}
Let $D_{\emptyset}$ vs. $\calS$ be a testing problem on $\R^n$.
Let $\calC$ be a linear subspace of functions $p \, : \, (\R^n)^{\otimes m} \rightarrow \R$, and let $\Pi_{\calC}$ be the orthogonal projection to the subspace $\calC$.
Then
\[
\argmax_{\substack{p \in \calC\\\E_{D_\emptyset^{\otimes m}} p^2 \le 1}} \left|\E_{u \sim S}\E_{D_u^{\otimes m}} p - \E_{D_\emptyset^{\otimes m}} p \right| = \frac{\Pi_{\calC}\left(\E_{u \sim S} \ol D_u^{\otimes m} -1\right)}{\left \|  \Pi_{\calC}\left(\E_{u \sim S}\ol D_u^{\otimes m}-1\right) \right \|_{D_\emptyset^{\otimes m}}}.
\]
\end{fact}

Letting $p = \tfrac{\Pi_{\calC}\left(\E_{u \sim S} \ol D_u^{\otimes m}-1\right)}{\left \| \Pi_{\calC}\left(\E_{u \sim S} \ol D_u^{\otimes m}-1\right) \right \|}$ be the optimizer of the above program, observe also that
\[
\E_{u \sim S}\E_{D_u^{\otimes m}} p - \E_{D_\emptyset^{\otimes m}} p = \left \| \Pi_{\calC}\left(\E_{u \sim S} \ol D_u^{\otimes m}-1\right)\right \|_{D_\emptyset^{\otimes m}}\mper
\]
Consequently,
\begin{fact}\label{fact:ldlr-norm}
If $\left \| \Pi_{\calC}(\E_{u \sim S} \ol D_u^{\otimes m}- 1) \right \| \le \eps$, then $D_\emptyset$ vs. $\calS$ has no $m$-sample $\eps$-distinguisher in $\calC$.
\end{fact}

\paragraph{Samplewise Degree.}
Rather than directly ruling out distinguishers of low degree, it will be convenient for us to introduce a notion of degree which agrees with the product structure (across samples) of $D_{\emptyset}^{\otimes m}$.
\begin{definition}[Samplewise degree]
For integers $m,n \ge 1$, we say that a function $f:(\R^n)^{\otimes m} \to \R$ has {\em samplewise degree (d,k)} if $f(x_1,\ldots,x_m)$ can be written as a linear combination of functions which have degree at most $d$ in each $x_i$, and nonzero degree in at most $k$ of the $x_i$'s.
\end{definition}

Note that a function of samplewise degree $(d,k)$ has degree at most $d\cdot k$, and a function of degree $d$ has samplewise degree at most $(d,d)$.

In order to rule out low-degree distinguishers, we will rule out low-samplewise degree distinguishers using Fact~\ref{fact:ldlr-norm}.
We denote the orthogonal projection of $f:(\R^n)^{\otimes m} \to \R$ to the span of samplewise degree $(d,k)$ functions by $f^{\le d,k}$.
We define the following quantity:
\begin{definition}[Low degree likelihood ratio]
For a hypothesis testing problem $D_\emptyset$ vs. $\calS = \{D_u\}$, the {\em $m$-sample $(d,k)$-low degree likelihood ratio function} is the projection of the $m$-sample likelihood ratio $\E_{u \sim S} \Paren{\ol D_u^{\otimes m}}$ to the span of non-constant functions of sample-wise degree at most $(d,k)$:
\[\Paren{  \E_{u \sim S} \ol D_u^{\otimes m} - 1}^{\leq d,k} = \E_{u \sim S} \Paren{ \ol D_u^{\otimes m} }^{\leq d,k}- 1.\]
We refer to this function as the $(d,k)$-$\LDLR_m$.
Abusing terminology, we also use $(d,k)$-$\LDLR_m$ to refer to the norm of the low degree likelihood ratio, $\|\E_{u \sim S} ( \ol D_u^{\otimes m})^{\leq d,k} - 1\|$.
\end{definition}

\section{Bounds on Degree Imply Bounds on Statistical Dimension}

\label{s:LDHardImpliesSDALarge}

In this section, we prove part 1 of Theorem~\ref{thm:main-intro}, showing that an upper bound on the low-degree likelihood ratio's norm (LDLR) implies lower bounds on the statistical dimension.

\begin{theorem}[LDLR to SDA Lower Bounds] \label{thm:low-deg-sq}
Let $d,k \in \mathbb{N}$ with $k$ even and $\calS = \{D_v\}_{v \in S}$ be a collection of probability distributions with prior $\mu$ over $\calS$.
Suppose that $\calS$ satisfies:
\begin{compactenum}
\item The $k$-sample high-degree part of the likelihood ratio is bounded by $\| \E_{u \sim \calS}( \ol D_u^{> d})^{\otimes k} \| \le \delta$. 
\item For some $m \in \mathbb{N}$, the $(d,k)$-$\LDLR_m$ is bounded by $\|\E_{u\sim\calS} (\ol D_u^{\otimes m})^{\le d,k}-1\|\le \eps$.
\end{compactenum}
Then for any $q \ge 1$, it follows that
\[
\SDA\left(\calS, \frac{m}{q^{2/k}(k\eps^{2/k} + \delta^{2/k} m)} \right) \ge q.
\]
\end{theorem}

Notice that for a $(m^{-k/2}/4,k)$-nice testing problem, Condition 1 of Theorem~\ref{thm:low-deg-sq} holds with $d =k$ and $\delta = m^{-k/2}/4$ (by definition).
So for $(m^{-k/2}/4,k)$-nice problems with no good $4mk$-sample degree $k^2$ distinguisher (and therefore no good samplewise degree $(k,k)$ distinguisher), setting $q = (2 m/m')^{k/2}$ in Theorem \ref{thm:low-deg-sq} implies that $\SDA(\calS, \Theta(m'/k)) \ge (2m/m')^{k/2}$, which establishes the contrapositive of part 1 of Theorem~\ref{thm:main-intro}.	
In subsequent sections, we will demonstrate that the niceness condition holds for many natural hypothesis testing problems (or in some cases, holds if the $(d,k)$-$\LDLR_m$ is small).
Combining these conditions with Theorem \ref{thm:low-deg-sq} will yield Theorems \ref{thm:noisy}, \ref{thm:gauss} and \ref{thm:prod}.

\begin{proof}[Proof of Theorem~\ref{thm:low-deg-sq}, for overview see Section~\ref{sec:comparing}]
Let $X$ be the random variable $X = \left|\langle \ol D_u, \ol D_v \rangle - 1\right|$ for $u,v \sim \calS$ sampled independently according to the prior $\mu$.
By definition, $\SDA(\calS,\frac{1}{t}) \ge q$ if $\E[X\mid A] \le t$ for all events $A$ over the choice of $u,v$ of probability at least $\frac{1}{q^2}$.
So our goal is to show that $\E[X \mid A] \le q^{2/k} (\tfrac{k}{m} \eps^{2/k} + \delta^{2/k})$.
We relate $\E[X \mid A]$ to moments of $X$ via H\"older's inequality:

\begin{fact}\label{fact:hol}
If $x$ is a real-valued random variable and $A$ is any event then $ \E[|x| \mid A] \le \left(\frac{\E[|x|^k]}{\Pr[A]}\right)^{1/k}$.
\end{fact}
We prove the fact below for completeness.
Since we have assumed that $k$ is even,
\[
\E X^k 
= \E_{u,v \sim S} \left(\langle \ol D_u, \ol D_v \rangle_{D_\emptyset} -1 \right)^k
= \E_{u,v \sim S} \left(\langle \ol D_u-1, \ol D_v-1 \rangle_{D_\emptyset} \right)^k
= \left\|\E_{u \sim S}(\ol D_u-1)^{\otimes k}\right\|_{D_\emptyset^{\otimes k}}^2,
\]
where we have first used that $\langle \ol D_u,1 \rangle =1$ for all $u \in \calS$, and then the independence of the samples.
Applying Fact~\ref{fact:hol},
\begin{equation}\label{eq:moment-to-cor}
\max_{\substack{A \text{ s.t. }\\ \Pr_{u,v \sim \calS}[A] \ge \frac{1}{q^2}}} \E_{u,v \sim \calS} \left [ \Abs{\iprod{\ol D_u, \ol D_v} - 1} \, | \, A \right ] \le \left(q \cdot \left\|\E_{u \sim S} (\ol D_u - 1)^{\otimes k}\right\|\right)^{2/k}\mper
\end{equation}
Now, applying H\"older's inequality (see Lemma~\ref{lem:holder} below), we can split the degree $\le d$ and degree $> d$ parts of $\ol D_u -1$ in our bound on the right-hand side,
\begin{equation}\label{eq:holder2}
\left\| \E_{u \sim S} (\ol D_u - 1)^{\otimes k} \right\|^{2/k} \le 
\left\| \E_{u \sim S} (\ol D_u^{\le d} - 1)^{\otimes k} \right\|^{2/k} +
\left\| \E_{u \sim S} (\ol D_u^{>d})^{\otimes k} \right\|^{2/k}.
\end{equation}
The second right-hand-side term is bounded by $\delta^{2/k}$ from Condition 1.
So, it remains to bound the first term.
This is our crucial ``boosting'' step.
We employ the following structural claim, which uses the independence of the samples to relate the correlation of the $(d,k)$ projections of $m$-sample likelihood ratios to the correlation of the $(d,k)$ projections of $k$-sample likelihood ratios, with $k \ll m$:
\begin{claim}\label{lem:many-to-one}
Let $D_u,D_v$ be distributions with relative densities $\overline{D}_u,\overline{D}_v$.
Then their $(d,k)$-projections are related as follows:
\[
\langle (\overline{D}_u^{\otimes m})^{\le d,k}, (\overline{D}_v^{\otimes m})^{\le d,k} \rangle -1
= \sum_{t = 1}^{k} \binom{m}{t}  \cdot \left(\langle \overline D_u^{\le d}, \overline D_v^{\le d} \rangle - 1\right)^t.
\]
\end{claim}
We give the (simple) proof of this claim below. Now, by linearity of expectation, the squared $(d,k)$-$\LDLR_m$ is equal to
\[
\left\| \E_{u \sim S} (\overline D_u^{\otimes m})^{\le d,k} -1\right\|^2
= \E_{u,v \sim S} \langle (\overline D_u^{\otimes m})^{\le d,k},(\overline D_v^{\otimes m})^{\le d,k}\rangle  - 1
= \E_{u,v \sim S} \sum_{t = 1}^k\binom{m}{t}\left(\langle \overline D_u^{\le d}, \overline D_v^{\le d}\rangle -1\right)^t \, ,
\]
where in the final equality we applied Claim~\ref{lem:many-to-one}.
So Condition 2 ($\|\E_u (\ol D_u^{\otimes m})^{\le d,k}\| \le \eps$) combined with the above implies that 
\[
\eps^2 \ge \left\| \E_{u \sim S} (\overline D_u^{\otimes m})^{\le d,k} -1\right\|^2 - \left\| \E_{u \sim S} (\overline D_u^{\otimes m})^{\le d,k-1} -1\right\|^2 = \binom{m}{k} \cdot \E_{u,v \sim S} \left(\langle \overline D_u^{\le d}, \overline D_v^{\le d}\rangle -1\right)^k \ge 0 \, .
\]
Dividing through by $\binom{m}{k}$ we have $\E_{u,v} (\langle \ol D_u^{\le d},\ol D_v^{\le d} \rangle -1)^k =\|\E_u (\ol D_u^{\le d}-1)^{\otimes k}\|^2 \le \frac{\eps^2}{\binom{m}{k}} \le \eps^2 \left(\frac{k}{m}\right)^k$.
Combining this with Equations~(\ref{eq:moment-to-cor}) and~(\ref{eq:holder2}) finishes the proof.
\end{proof}

We now prove the outstanding claims, in order of mathematical interest.
\begin{proof}[Proof of Claim~\ref{lem:many-to-one}]
  We write $\ol D_u = 1 + (\ol D_u^{\leq d} -1) + \ol D_u^{> d}$.
  Expanding the tensor power,
  \begin{align*}
    (\ol D_u^{\otimes m})^{\leq d,k} = \sum_{A \subseteq [m], B \subseteq [m] \setminus A} \Paren{1^{\otimes A} \otimes (\ol D_u^{\leq d} - 1)^{\otimes B} \otimes (\ol D_u^{> d})^{\otimes [m] \setminus (A \cup B)}}^{\leq d,k}\mper
  \end{align*}
  Now, $\ol D_u^{> d}$ is orthogonal to all functions of degree at most $d$.
  So the projection
  \[
  \Paren{1^{\otimes A} \otimes (\ol D_u^{\leq d} - 1)^{\otimes B} \otimes (\ol D_u^{> d})^{\otimes [m] \setminus (A \cup B)}}^{\leq d,k} = 0
  \]
  unless $A \cup B = [m]$, and hence
  \begin{align*}
    (\ol D_u^{\otimes m})^{\leq d,k} = \sum_{A \subseteq [m]} \Paren{1^{\otimes A} \otimes (\ol D_u^{\leq d} - 1)^{\otimes [m] \setminus A}}^{\leq d,k} \mper
  \end{align*}
  Furthermore, if $|[m] \setminus A| > k$, then $1^{\otimes A} \otimes (\ol D_u^{\leq d} - 1)^{\otimes [m] \setminus A}$ is orthogonal to every function depending on at most $k$ samples.
  So again applying the projection to degree-$(d,k)$,
  \begin{align*}
    (\ol D_u^{\otimes m})^{\leq d,k} = \sum_{B \subseteq [m], |B| \leq k} 1^{\otimes [m] \setminus B} \otimes (\ol D_u^{\leq d} - 1)^{\otimes B}\mper
  \end{align*}

  Observe also that if $B, B' \subseteq [m]$ and $B \neq B'$, then 
  \[
  \iprod{1^{\otimes [m] \setminus B} \otimes (\ol D_u^{\leq d} - 1)^{\otimes B}, 1^{\otimes [m] \setminus B'} \otimes (\ol D_v^{\leq d} - 1)^{\otimes B'}} = 0\mper
  \]
  So we have
  \[
  \langle (\overline{D}_u^{\otimes m})^{\le d,k}, (\overline{D}_v^{\otimes m})^{\le d,k} \rangle -1 = \sum_{B \subseteq [m], B \neq \emptyset} \iprod{\ol D_u^{\leq d} - 1, \ol D_v^{\leq d} - 1}^{|B|} \, ,
  \]
  which, by the independence of samples, proves the claim.
\end{proof}

\begin{lemma}\label{lem:holder}
Let $D_{\emptyset}$ be a null distribution and $\calS = \{D_u\}_{u \in S}$ be a set of alternate distributions with $D_u$'s density relative to $D_\emptyset$ density given by $\ol D_u$ for each $u \in S$.
Let $k,d\ge 1$ be integers with $k$ even.
Then the centered $k$-sample likelihood ratio  may be bounded in terms of the $k$-sample-homogeneous low-degree part and the $k$-sample-homogeneous high degree part:
\[
\left\| \E_{u \sim S} (\ol D_u - 1)^{\otimes k} \right\|^{2/k} \le 
\left\| \E_{u \sim S} (\ol D_u^{\le d} - 1)^{\otimes k} \right\|^{2/k} +
\left\| \E_{u \sim S} (\ol D_u^{>d})^{\otimes k} \right\|^{2/k}.
\]
\end{lemma}
\begin{proof}
By the triangle inequality, H\"{o}lder's inequality and the fact that $k$ is even, we have that
\allowdisplaybreaks
\begin{align*}
\E_{u,v}\left[ \left(\langle \ol D_u, \ol D_v \rangle-1\right)^k\right]
&= \E_{u,v}\left[ \left(\langle \ol D_u^{\le d}, \ol D_v^{\le d} \rangle-1 + \langle \ol D_u^{> d}, \ol D_v^{> d} \rangle\right)^k\right]\\
&\le \E_{u,v}\left[ \left( \left| \langle \ol D_u^{\le d}, \ol D_v^{\le d} \rangle-1 \right| + \left| \langle \ol D_u^{> d}, \ol D_v^{> d} \rangle \right| \right)^k\right]\\
&\le \sum_{\ell = 0}^k\binom{k}{\ell} \E_{u,v}\left[ \left(\langle \ol D_u^{\le d}, \ol D_v^{\le d} \rangle-1\right)^k \right]^{\ell/k} \E_{u,v}\left[\left( \langle \ol D_u^{> d}, \ol D_v^{> d} \rangle\right)^{k}\right]^{(k-\ell)/k}\\
&= \left(\E_{u,v}\left[ \left(\langle \ol D_u^{\le d}, \ol D_v^{\le d} \rangle-1\right)^k \right]^{1/k} + \E_{u,v}\left[\left( \langle \ol D_u^{> d}, \ol D_v^{> d} \rangle\right)^{k}\right]^{1/k}\right)^k,
\end{align*}
and the conclusion now follows because $\langle \ol D_u,1\rangle =1$ for all $u \in S$, which implies $\E_{u,v} (\langle \ol D_u, \ol D_v \rangle -1)^k = \|\E_{u}(\ol D_u - 1)^{\otimes k}\|^2$ and $\E_{u,v}(\langle \ol D_u^{\le d}, \ol D_v^{\le d} \rangle -1)^k = \|\E_u (\ol D_u^{\le d}-1)^{\otimes k}\|^2$.
\end{proof}

\begin{proof}[Proof of Fact~\ref{fact:hol}]
Observe that
\[
\E[|x|\mid A] = \frac{\E[|x| \cdot \Ind[A]]}{\Pr[A]} \le \frac{\E[|x|^k]^{1/k}\E[\Ind[A]]^{1-1/k}}{\Pr[A]} = \left(\frac{\E[|x|^k]}{\Pr[A]}\right)^{1/k}.
\]
where we have applied H\"{o}lder's inequality.
\end{proof}

We encapsulate the conclusion of the boosting argument above in the following standalone lemma, which will be useful later:
\begin{lemma}[Samplewise-LDLR boosting]\label{lem:boosting}
If the $(d,k)$-$\LDLR_m$ for the hypothesis testing problem of $D_{\emptyset}$ vs $\{D_v\}_{v \in S}$ is bounded,
then the moments of the low-degree single-sample LR are also bounded, by
\[
\|\E_{u\sim S} \ol (D_u^{\le d} -1)^{\otimes k} \|^2 = \E_{u,v \sim S} \left(\langle \overline D_u^{\le d},\overline D_v^{\le d} \rangle  - 1\right)^k \le \frac{1}{\binom{m}{k}}\left\| \E_{u \sim S} (\overline D_u^{\otimes m})^{\le d,k}-1\right\|^2.
\]
\end{lemma}

The proof is identical to the end of the proof of Theorem~\ref{thm:low-deg-sq}.

\section{Bounds on Statistical Dimension Imply Bounds on Degree}\label{sec:converse}
In this section, we show that lower bounds on the statistical dimension imply that the low-degree likelihood ratio norm is small (hence ruling out good low-degree distinguishers).
We will prove the following theorem:

\begin{theorem}\label{thm:converse}
  Let $\calS$ be a hypothesis testing problem on $\R^N$ with respect to null hypothesis $D_\emptyset$.
  Let $m,k \in \N$ with $k$ even.
  Suppose that for all $0 \leq m' \leq m$,
  $\SDA(\calS, m') \geq 100^k \cdot (m/m')^k$.
  (In particular, $\SDA(\calS, m) \geq 100^k$.)
  Then for all $d$, $\|\E_{u \sim \calS} (\overline{D}_u^{\otimes m})^{\leq d,\Omega(k)} -1 \|^2 \leq 1$.
\end{theorem}

The key lemma to prove Theorem~\ref{thm:converse} is the following, which translates the bound $\SDA(\calS, m') \geq 100^k \cdot (m/m')^k$ to a bound on the moments of $\iprod{D_u, D_v} - 1$.

\begin{lemma}\label{lem:converse-main}
  In the setting of Theorem~\ref{thm:converse}, for any $t \leq k/8$, $\E_{u,v \sim S}(\iprod{D_u, D_v} - 1)^t \leq 4 \cdot (1/100m)^t$.
\end{lemma}

Now we prove Theorem~\ref{thm:converse}.

\begin{proof}[Proof of Theorem~\ref{thm:converse}]
  We use Claim~\ref{lem:many-to-one} and Lemma~\ref{lem:converse-main} to obtain
  \[
  \E_{u,v \sim \calS} \iprod{(\overline{D}_u^{\otimes m})^{\leq \infty,k/8}, (\overline{D}_v^{\otimes m})^{\leq \infty,k/8}} \leq \sum_{t=1}^{k/8} \binom{m}{t} \E_{u,v \sim \calS} \Paren{ \iprod{\overline{D}_u, \overline{D}_v} -1}^t \leq \sum_{t=1}^{k/8} \binom{m}{t} \cdot 4 \cdot \Paren{\frac 1 {100m}}^t\mper
  \]
  Using $\binom{m}{t} \leq (me/t)^t$, we find that this is at most $4 \sum_{t=1}^{k/8} \Paren{\frac e {100t}}^t \leq 4 (e^{e/100}-1) \leq 1$.
  But for all $d \in \N$ we have
  \[
  \|\E_{u \sim \calS} (\overline{D_u}^{\otimes m})^{\leq d,k/8} -1 \|^2 \leq \E_{u,v \sim \calS} \iprod{(\overline{D}_u^{\otimes m})^{\leq \infty,k/8}, (\overline{D}_v^{\otimes m})^{\leq \infty,k/8}}
  \]
  which completes the proof.
\end{proof}

We turn to the proof of Lemma~\ref{lem:converse-main}.
We need the following basic fact to relate the moments and tails of $\iprod{D_u,D_v} - 1$.
(The proof is straightforward calculus; see e.g. Appendix A.2 of \cite{HL19}.)

\begin{fact}
\label{fact:mom-mean}
  Let $X$ be an $\R$-valued random variable.
  For every $p > q > 0$, $\E |X|^q \leq (2 \sup_A \Pr[A] \cdot (\E [ X \, | \, A])^p )^{q/p} \cdot \tfrac p {p-q}$.
  (The supremum is taken over all events $A$.)
\end{fact}

\begin{proof}[Proof of Lemma~\ref{lem:converse-main}]
  Let $X = |\iprod{D_u,D_v} - 1|$ be the $\R$-valued random variable given by two random draws $u,v \sim \calS$.
  Our assumption $\SDA(\calS, m') \geq 100^k \cdot (m/m')^k$ for all $m' \leq m$ implies that for every event $A$ of probability $\alpha \geq 100^{-2k} \cdot (m'/m)^{2k}$, we have $\E [X \, | \, A] \leq 1/m'$.
  Rearranging, for all events $A$ of probability $\alpha$, we have $\E[ X \, | \, A] \leq \tfrac 1 {100m \alpha^{2/k}}$.
  So for any $t\leq k/2$,
  \[
  \sup_{A} \Pr(A) \cdot (\E [X \, | \, A])^{t} \leq \sup_{\alpha \geq 0} \alpha^{1 - 2t/k} \cdot \Paren{\frac 1 {100m}}^{t} \leq \Paren{\frac 1 {100m}}^{t}\mper
  \]
  So applying Fact~\ref{fact:mom-mean} for any $t \leq k/8$,
  \[
  \E X^t \leq 4 \cdot (1/100m)^t\mper \qedhere
  \]
\end{proof}

\section{Specialization to Noise-Robust Problems}\label{sec:noise-robust}

In this section, we observe that Theorem \ref{thm:low-deg-sq} immediately applies to noise-robust problems, as noise-robustness implies a bound on the high-degree part of the LR.

\subsection{Noise Operators}
We define a class of Markov operators which generalize the Gaussian and discrete noise operators.
  Recall that a Markov operator $T$ is a linear operator such that if $f$ is a probability density, then so is $Tf$.

\begin{definition}[$(d,\epsilon)$-Markov operator]
  Let $D_\emptyset$ be a probability measure on $\R^N$ (or a discrete distribution on $\Omega^N$ for some finite set $\Omega$), inducing an inner product on functions $f,g \, : \, \R^N \rightarrow \R$ (or $f,g \, : \, \Omega^N \rightarrow \R$) by $\iprod{f,g} = \E_{x \sim D_\emptyset} f(x) g(x)$.
  Let $\ell_2 = \{ f \, : \, \R^N \rightarrow \R \text{ s.t. } \E_{x \sim D_\emptyset} f(x)^2 \leq \infty \}$.
  Let $d \in \N$, and let $\ell_2^{\geq d}$ be the orthogonal complement of $\mathrm{span} \{ f \in \ell_2 \, : \, f \text{ has degree $(d-1)$} \}$ with respect to $\iprod{\cdot, \cdot}$. 

  Any hypothesis testing problem $(D_\emptyset,\calS)$ and Markov operator $T:\ell_2 \to \ell_2$ induce another hypothesis testing problem $(D_\emptyset, T \calS)$ by applying $T$ to each of the distributions $D_u \in \calS$.
  We call a Markov operator $T$ a {\em $(d,\epsilon)$-operator} if
  \[
  \ell_2^{\geq d} \subseteq \mathrm{span} \{ f \in \ell_2 \,: \, f \text{ is an eigenfunction of $T$ with eigenvalue $\lambda$ such that $|\lambda| \leq \epsilon$ } \}\mper
  \]
\end{definition}
  Our main examples are the Ornstein-Uhlenbeck operator $U_\rho$ (a.k.a. the Gaussian noise operator) and the discrete noise operator $T_\rho$, both of which are $(d,\rho^d)$ operators.
  In both cases, the testing problems $(D_\emptyset, T \calS)$ will be noisy versions of original problems $(D_\emptyset, \calS)$.
  However, we will use a different family of noise operators to treat certain statistical problems where there is planted structure which is not robust to independent entrywise noise, such as planted clique.

\subsection{Results for Noise-Robust Problems}

\begin{theorem}\label{thm:noisy}
Let $d,k \in \mathbb{N}$ with $k$ even and $\calS = \{D_v\}_{v \in S}$ be a collection of probability distributions, let $\ol D_u$ be the relative density of $D_u$ with respect to $D_{\emptyset}$.
Let $T$ be a $(d+1,\rho^{d+1})$ Markov operator.
Suppose that the $k$-sample likelihood ratio is bounded by $ \|\E_{u} \ol D_u^{\otimes k}\|^2\le C^k$, and the noised $(d,k)$-$\LDLR_m$ is bounded by $\|\E_{u} (T \ol D_u^{\otimes m})^{\le d,k}-1\|\le \eps$.
Then it follows that for any $q \ge 1$,
\[
\SDA\left(\calS, \frac{m}{q^{2/k}(k\eps^{2/k} + \rho^{2(d+1)} Cm)} \right) \ge q\,.
\]
\end{theorem}

\begin{proof}
Since $T$ is a $(d+1,\rho^{d+1})$ Markov Operator by assumption, the $k$-sample high-degree part of the LR is bounded by
$$\left\| \E_u (T \ol D_u^{> d})^{\otimes k} \right\|^2
\le \rho^{2(d+1)k} \cdot \left\| \E_u (\ol D_u^{> d})^{\otimes k} \right\|^2
\le \rho^{2(d+1)k} \cdot \left\| \E_u (\ol D_u)^{\otimes k} \right\|^2 \le \rho^{2(d + 1)k} \cdot C^k\,.$$
Applying Theorem \ref{thm:low-deg-sq} now completes the proof of this theorem.
\end{proof}

\subsection{Robustness to Random Restrictions}

Some problems of interest are not noise-robust under nontrivial $(\rho,d)$-operators. 
For example, consider the (bipartite) planted clique problem---the clique structure is not preserved if the coordinates are resampled independently.\footnote{In the bipartite version, we further require that the resampling procedure be dependent {\em across samples}.}
To accommodate such problems, we generalize Theorem \ref{thm:noisy} to a different class of noise operators: random restrictions.
A random restriction fixes a random subset of coordinates, then applies noise to the remaining coordinates across all of the samples.

\begin{definition}[Random Restriction]
Let $T$ be a Markov operator on $\R^N$.
Given a subset $R \subset [N]$, let $T^{\ol R}$ be the Markov operator on $\R^N$ that applies $T$ to all entries except those in $R$.
Given a set of probability distributions $\calS$ and a prior $\mu$ over $\calS$, the {\em $(T,s)$-random restriction of $\calS$} is the set of distributions
\[
\calS' = \left\{T^{\ol R} D ~\mid~ D \in \calS,\, R \subseteq [N]\right\}
\]
 equipped with the prior $\mu'$ where a sample $T^{\ol R} D \sim \mu'$ is generated sampling $D \sim \mu$ and sampling $R$ by including every coordinate in $[N]$ independently with probability $\frac{s}{N}$.
Denote the distribution on subsets as $\calR_N(s)$.
\end{definition}
We will often abuse notation and let $T^{\ol R}$ stand in for $(T^{\otimes n})^{\ol R}$ when $T$ is a noise operator on $\R$.

For simplicity we restrict our attention to distributions $D_v$ over the boolean hypercube $\{ \pm 1\}^n$, and to null distributions $D_{\emptyset}$ which are product measures for which all biases are the same, $D_{\emptyset} = D_0^{\otimes N}$.\footnote{ We expect that a near-identical proof will extend to the case when $D_{\emptyset}$ is a product measure with arbitrary coordinate biases.}
We now have the following lemma:

\begin{lemma} \label{lem:randrest}
Let $D_{\emptyset}$ be a product measure over $\{\pm 1\}^N$.
Let $d, k \in \mathbb{N}$, let $T$ be a $(1,\rho)$-operator over $\{\pm 1\}$ (with respect to the measure induced by $D_{\emptyset}$ on a single coordinate).
Then for $\calS = \{ D_v \}_{v \in S}$ a family of distributions over $\{ \pm 1\}^N$ with prior $\mu$, we have that the $(T,s)$-random restriction $\calS',\mu'$ of $\calS$ has degree $(> d, =k)$ bounded by
$$\left\| \E_{R \sim \calR_N(s)} \E_{u \sim \mu} \left( T^{\ol R} \ol D_u^{> d} \right)^{\otimes k} \right\|^2
\le \max\left\{ 4^{d + 1} \rho^{2(d+1)k}, \left( \frac{2s}{n} \right)^{2(d + 1)} \right\} \cdot \left\| \E_{u \sim \mu} (\ol D_u)^{\otimes k} \right\|^2\,.$$
\end{lemma}

\begin{proof}
We will abuse notation and let $T^{\ol R}$ simultaneously denote the noise operator on $(\mathbb{R}^N)^{\otimes k}$ that applies $T^{\ol R}$ independently to each copy of $\mathbb{R}^N$.
Let $\ol D = \E_{u \sim \mu} \left( \ol D_u \right)^{\otimes k}$ and let $\widehat{D}(\alpha_1, \alpha_2, \dots, \alpha_k)$ denote the Fourier character of $\ol D$ at the subsets $\alpha_1, \alpha_2, \dots, \alpha_k \subseteq [N]$.
By the definition of $T^{\ol R}$, we have that
$$\widehat{T_\rho^R D}(\alpha_1, \alpha_2, \dots, \alpha_k) = \rho^{\sum_{i = 1}^k |\alpha_i \cap R^c |} \cdot \widehat{D}(\alpha_1, \alpha_2, \dots, \alpha_k)$$
for any $\alpha_1, \alpha_2, \dots, \alpha_k \subseteq [N]$. Let $T'$ denote the operator $\E_{R \sim \calR_N(s)} T^{\ol R}$ and observe that
$$\widehat{T' D}(\alpha_1, \alpha_2, \dots, \alpha_k) = \E_{R \sim \calR_N(s)}\left[ \rho^{\sum_{i = 1}^k |\alpha_i \cap R^c |} \right] \cdot \widehat{D}(\alpha_1, \alpha_2, \dots, \alpha_k)\,.$$
Now by H\"{o}lder's inequality, we have that
\begin{align*}
\E_{R \sim \calR_n(s/n)}\left[ \rho^{\sum_{i = 1}^k |\alpha_i \cap R^c |} \right] &\le \prod_{i = 1}^k \E_{R \sim \calR_N(s)} \left[ \rho^{k |\alpha_i \cap R^c |} \right]^{1/k} \\
&= \prod_{i = 1}^k \E_{R \sim \calR_N(s)} \left[ \prod_{j \in \alpha_i} \rho^{k \cdot \Ind (j \not \in R)} \right]^{1/k} 
= \left( \frac{s}{N} + \left(1 - \frac{s}{N} \right) \rho^k \right)^{\sum_{i = 1}^k |\alpha_i|/k},
\end{align*}
where the final equality follows from the fact that the events $\Ind(j \not \in R)$ are independent and occur with probability $1 - \frac{s}{N}$ under $R \sim \calR_N(s)$.
Now by Parseval's inequality, we have that
\allowdisplaybreaks
\begin{align*}
\left\| \E_{R \sim \calR_N(s)} \E_{u \sim \mu} \left( T_\rho^R \ol D_u^{> d} \right)^{\otimes k} \right\|^2 &= \sum_{|\alpha_1|, |\alpha_2|, \dots, |\alpha_k| > d} \widehat{T' D}(\alpha_1, \alpha_2, \dots, \alpha_k)^2 \\
&\le \sum_{|\alpha_1|, |\alpha_2|, \dots, |\alpha_k| > d} \left( \frac{s}{N} + \left(1 - \frac{s}{N} \right) \rho^k \right)^{2\sum_{i=1}^k|\alpha_i|/k} \cdot \widehat{D}(\alpha_1, \alpha_2, \dots, \alpha_k)^2 \\
&\le \left( \frac{s}{N} + \left(1 - \frac{s}{N} \right) \rho^k \right)^{2(d + 1)} \sum_{|\alpha_1|, |\alpha_2|, \dots, |\alpha_k| > d} \widehat{D}(\alpha_1, \alpha_2, \dots, \alpha_k)^2 \\
&\le \left( \frac{s}{N} + \rho^k \right)^{2(d + 1)} \cdot \left\| \E_{u \sim \mu} (\ol D_u)^{\otimes k} \right\|^2\,. \numberthis \label{eqn:parsevals-randrest}
\end{align*}
The lemma then follows from the fact that $s/N + \rho^k \le \max\{ 2\rho^k, \frac{2s}{N} \}$.
\end{proof}

Applying Theorem \ref{thm:low-deg-sq} yields the following Corollary: 

\begin{corollary} \label{cor:rand-rest}
Let $D_{\emptyset}$ be a product measure over $\{\pm 1\}^N$.
Let $d, k \in \mathbb{N}$ with $k$ even, let $T$ be a $(1,\rho)$-operator over $\{\pm 1\}$ (with respect to the measure induced by $D_{\emptyset}$ on a single coordinate).
Let $\calS = \{ D_v \}_{v \in S}$ a family of distributions over $\{ \pm 1\}^N$ with prior $\mu$ over $\calS$, and let $\ol D_u$ be the relative density of $D_u$ with respect to $D_{\emptyset}$.
Suppose that the $k$-sample likelihood ratio is bounded by $ \|\E_{u} \ol D_u^{\otimes k}\|^2\le C^k$, and suppose that the $(T,s$)-randomly restricted alternate hypothesis class $\calS,\mu'$ has $(d,k)$-$\LDLR_m$ bounded,
$$\left\| \E_{R \sim \calR_N(s)} \E_{u\sim \mu} (T^{\ol R} \ol D_u^{\otimes m})^{\le d,k}-1 \right\|\le \eps,$$
Then it follows that for any $q \ge 1$,
\[
\SDA\left( \calS',\mu', \, \frac{m}{q^{2/k}} \left( k\eps^{2/k} + \max\left\{ 4^{(d+1)/k} \rho^{2(d+1)}, \left( \frac{2s}{n} \right)^{2(d + 1)/k} \right\} Cm \right)^{-1} \right) \ge q.
\]
\end{corollary}

\begin{remark}[Comparison to Theorem~\ref{thm:noisy}]
As long as $k = \Omega(d)$, $4^{(d + 1)/k} = O(1)$ and thus this theorem can be viewed as a natural extension of Theorem \ref{thm:noisy}, recovering (essentially) the same result when $s = 0$.\footnote{
We also remark that the $(2s/N)^{2(d + 1)}$ factor in Lemma \ref{lem:randrest} cannot in general be improved.
In particular, when $\rho = 0$, the diagonal Fourier coefficients of the form $\widehat{T'D}(\alpha, \alpha, \dots, \alpha)$ are exactly equal to $(s/N)^{|\alpha|} \cdot \widehat{D}(\alpha, \alpha, \dots, \alpha)$.
However, other Fourier coefficients are scaled down more heavily under $T'$ and it is possible to improve the bound in Lemma \ref{lem:randrest} under further assumptions about the Fourier coefficients of $\ol D$.}
\end{remark}
In Section~\ref{sec:pc-ex}, we show that Corollary~\ref{cor:rand-rest} implies an equivalence between distinguishers and statistical queries for a number of models such as planted clique, in which the planted structure is not robust to independent noise.

\subsubsection{Random Subtensor Restrictions}

In the above, we treated random restrictions in which coordinates in $[N]$ are fixed independently.
In tensor- and matrix-problems, where $\{\pm 1\}^N$ is identified with $({\pm 1}^n)^{\otimes p}$ for an integer $p$, the natural notion of random restriction restricts to a random principal minor $(\{\pm 1\}^R)^{\otimes p}$.
Below, we will generalize Corollary~\ref{cor:rand-rest} to this type of random restriction.

Let $\calR_n(s)$ be as in the section above, and for $R \in \calR_n(s)$ let $R^{\otimes p}$ denote the set of all coordinates in $(\{\pm 1\}^n)^{\otimes p}$ where all $p$ modes lie in $R$.

\begin{lemma} \label{lem:submatrix-rand-rest}
Let $p,s,n,k,d \in \N$ and $\rho \in (0,1)$ with $2s \le n$, $2^{p/k} \rho \le 1$.
Let $D_{\emptyset}$ be a product measure over $\{\pm 1\}^N$ where $N = n^p$, and let $T$ be a $(1,\rho)$-operator over $\{\pm 1\}$ (with respect to the measure induced by $D_{\emptyset}$ on a single coordinate). 
Then for $\calS = \{ D_v \}_{v \in S}$ a family of distributions over $\left(\{ \pm 1\}^n\right)^{\otimes p}$ with prior $\mu$, we have that the $(T,s)$-random restriction $\calS',\mu'$ of $\calS$ has degree $(> d, =k)$ bounded by
$$\left\| \E_{R \sim \calR_n(s)} \E_{u \sim \mu} \left( T^{\ol{R^{\otimes p}}} \ol D_u^{> d} \right)^{\otimes k} \right\|^2
\le \max\left\{ 4^{d + 1} \rho^{(d+1)k/p}, \left( \frac{2s}{n} \right)^{2\left(\frac{1}{2}(d + 1)\right)^{1/p}} \right\} \cdot \left\| \E_{u \sim S'} (\ol D_u)^{\otimes k} \right\|^2\,.$$
\end{lemma}

\begin{proof}
As in Lemma~\ref{lem:randrest}, let $\ol D = \E_{u \sim \mu} (\ol D_u)^{\otimes k}$ with Fourier coefficients $\widehat{D}(\alpha_1, \alpha_2, \dots, \alpha_k)$ for any sequence of subsets $\alpha_1, \alpha_2, \dots, \alpha_k \subseteq [n]^p$.
Similarly, let $T' = \E_{R \sim \calR_n(s)} T^{\ol{R^{\otimes p}}}$. 
Applying H\"{o}lder's inequality just as in the proof of Lemma \ref{lem:randrest}, we have that
\begin{align*}
\widehat{T' D}(\alpha_1, \alpha_2, \dots, \alpha_k) &= \E_{R \sim \calR_n(s)}\left[ \rho^{\sum_{\ell = 1}^k |\alpha_\ell \cap (R^{\otimes p})^c |} \right] \cdot \widehat{D}(\alpha_1, \alpha_2, \dots, \alpha_k) \\
&\le \left( \prod_{\ell = 1}^k \E_{R \sim \calR_n(s)} \left[ \prod_{(i_1, i_2, \dots, i_p) \in \alpha_\ell} \rho^{k \cdot \Ind (\exists a \in [p],\, i_a \not \in R)} \right]^{1/k} \right) \cdot \widehat{D}(\alpha_1, \alpha_2, \dots, \alpha_k) \numberthis \label{eqn:randrest-initial}
\end{align*}
We now will prove the following claim which will complete the proof of the lemma.

\begin{claim} \label{claim:subtensor}
For any $\alpha \subseteq [n]^p$, so long as $2^{p/k} \rho \le 1$ and $2s \le n$,
\begin{equation}
\E_{R \sim \calR_n(s)} \left[ \prod_{(i_1, i_2, \dots, i_p) \in \alpha} \rho^{k \cdot \Ind (\exists a \in [p],\,i_a \not \in R)} \right] \le  \max\left\{ 2^{\frac{1}{2}|\alpha|}\rho^{\frac{k}{2p}|\alpha|}, \left( \frac{2s}{n} \right)^{(\frac{1}{2}|\alpha|)^{1/p}} \right\}\,. \label{eqn:randrest-claim}
\end{equation}
\end{claim}
\begin{proof}
Let $V(\alpha)
= \{i \in [n] \mid \exists (i_1,\ldots,i_p) \in \alpha,\, a \in [p] \, \text{s.t. } i = i_a \}$ be the set of indices of $[n]$ that appear in $\alpha$.
For each $i \in V(\alpha)$, let $d_i \ge 1$ be the total number of times $i$ appears as an index in $\alpha$.
Since $|\rho| \le 1$ and $\Ind(\exists a \in [p], \, i_a \not\in R) \le \frac{1}{p} \sum_{a \in [p]} \Ind(i_a \not \in R)$, we have that
\begin{align*}
\E\left[\prod_{(i_1,\ldots,i_p) \in \alpha} \rho^{k \Ind(\exists a \in [p],\, i_a \not\in R)} \right]
&\le
\E\left[\prod_{(i_1,\ldots,i_p) \in \alpha} \rho^{\frac{k}{p} \sum_{a \in [p]} \Ind(i_a \not\in R)} \right]\\
&= \E\left[\prod_{i \in V(\alpha)} \rho^{\frac{k}{p} d_i \Ind(i \not\in R)} \right]\\
&= \prod_{i \in V(\alpha)} \E\left[\rho^{\frac{k}{p} d_i \Ind(i \not \in R)}\right]\\
&= \prod_{i \in V(\alpha)} \left(\frac{s}{n} + \left(1-\frac{s}{n}\right) \rho^{\frac{k}{p} d_i}\right)\\
&\le 2^{|V(\alpha)|} \cdot \max_{U \subseteq V(\alpha)} \left(\frac{s}{n}\right)^{|V(\alpha)\setminus U|}\cdot \rho^{\frac{k}{p}\sum_{i \in U} d_i},\\
&\le \max_{U \subseteq V(\alpha)} \left(\frac{2s}{n}\right)^{|V(\alpha)\setminus U|}\cdot (2^{p/k}\rho)^{\frac{k}{p}\sum_{i \in U} d_i},
\end{align*}
where to obtain the third line we have used the independence of the events $\Ind(i \not\in R)$, in the penultimate line we have bounded the product expansion by its maximum term, and in the final line we have used that $d_i \ge 1$ for all $i \in U$.
If $\sum_{i \in U} d_i \ge \frac{1}{2}|\alpha|$, then since $2s \le n$ and $2^{p/k}\rho \le 1$ we have $\left(\frac{2s}{n}\right)^{|V(\alpha)\setminus U|} (2^{p/k}\rho)^{\frac{k}{p}\sum_{i \in U} d_i} \le (2^{p/k}\rho)^{\frac{k}{2p}|\alpha|}$, and we have our conclusion.
Otherwise suppose $\sum_{i \in U} d_i < \frac{1}{2}|\alpha|$ and consider the set tuples $\alpha'$ which do not contain elements from $U$. 
We have that $|\alpha'| \ge \frac{1}{2}|\alpha|$, because the elements of $U$ participate in at most $\sum_{i \in U} d_i$ tuples.
Further, $|\alpha'| \le (|V(\alpha) \setminus U|)^p$, since this is the number of distinct tuples of at most $p$ elements that can be formed from the elements of $V(\alpha) \setminus U$.
Thus $|V(\alpha)\setminus U| \ge (\frac{1}{2}|\alpha|)^{1/p}$, and the bound now follows because $\left(\frac{s}{n}\right)^{|V(\alpha)\setminus U|} (2^{p/k}\rho)^{\frac{k}{p}\sum_{i \in U} d_i} \le \left(\frac{2s}{n}\right)^{(\frac{1}{2}|\alpha|)^{1/p}}$.
\end{proof}

Combining Equations (\ref{eqn:randrest-initial}) and (\ref{eqn:randrest-claim}) with a similar application of Parseval's inequality as in Equation (\ref{eqn:parsevals-randrest}) from Lemma \ref{lem:randrest} now completes the proof of the lemma.
\end{proof}

Combining this lemma with Theorem \ref{thm:low-deg-sq} now yields that LDLR bounds for problems that can be realized as random submatrix or subtensor restrictions imply SQ lower bounds, as in Corollary \ref{cor:rand-rest} in the previous section.
We remark that the bounds in Lemma \ref{lem:submatrix-rand-rest} are nearly tight.\footnote{
When $\rho = 0$, the diagonal Fourier coefficients corresponding to submatrices are given by $\widehat{T'D}(R^{\otimes p}, \dots, R^{\otimes p}) = (s/n)^{|R|} \cdot \widehat{D}(R^{\otimes p}, \dots, R^{\otimes p})$.
This implies that the $\left(\frac{1}{2}(d + 1)\right)^{1/p}$ factor in the exponent of $(2s/n)^{2\left(\frac{1}{2}(d + 1)\right)^{1/p}}$ in Lemma \ref{lem:submatrix-rand-rest} is necessary.}

\begin{remark}\label{remark:pc-submatrix-rand-rest}
A final setting of interest (e.g. for multi-sample planted clique) is when $N = \binom{n}{p}$ and the indices of samples are identified with subsets in $\binom{[n]}{p}$.
 The natural notion of a random restriction is then to subsets of the form $\binom{R}{p} \in \binom{[n]}{p}$ where $R \sim \calR_n(s)$.
Lemma \ref{lem:submatrix-rand-rest} can be seen to handle this case as well: repeating the argument identically, but considering only tuples $(i_1,\ldots,i_p)$ with $i_1 < \cdots < i_p$, yields the following theorem. 
\end{remark}

\begin{theorem}\label{thm:random-rest-tensor}
Let $p,s,n,k,d \in \N$ and $\rho \in (0,1)$ with $2s \le n$, $2^{p/k} \rho \le 1$.
Let $D_{\emptyset}$ be a product measure over $\{\pm 1\}^N$ where $N = \binom{n}{p}$, and let $T$ be a $(1,\rho)$-operator over $\{\pm 1\}$ (with respect to the measure induced by $D_{\emptyset}$ on a single coordinate). 
Then for $\calS = \{ D_v \}_{v \in S}$ a family of distributions over $\{ \pm 1\}^{\binom{[n]}{p}}$ with prior $\mu$, we have that the $(T,s)$-random restriction $\calS',\mu'$ of $\calS$ has degree $(> d, =k)$ bounded by
$$\left\| \E_{R \sim \calR_n(s)} \E_{u \sim \mu} \left( T^{\ol{\binom{R}{p}}} \ol D_u^{> d} \right)^{\otimes k} \right\|^2
\le \max\left\{ 4^{d + 1} \rho^{(d+1)k/p}, \left( \frac{2s}{n} \right)^{2\left(\frac{1}{2}(d + 1)\right)^{1/p}} \right\} \cdot \left\| \E_{u \sim S'} (\ol D_u)^{\otimes k} \right\|^2\,.$$
\end{theorem}

\section{Specialization to Distributions with Independent Coordinates}\label{sec:indep}

In this section, we prove Theorems~\ref{thm:gauss} and \ref{thm:prod}. 
In each case, we bound the high-degree part of the LR in terms of the LDLR and then apply Theorem \ref{thm:low-deg-sq} to deduce the result.

\subsection{Identity-Covariance Gaussians}

\begin{theorem}\label{thm:gauss}
Let $k$ be an even integer.
For the null distribution $D_{\emptyset} = \calN(0,\Id_n)$ and alternate distributions $\calS = \{D_v\}_{v \in S}$ with $D_v = \calN(v, \Id_n)$, let $\ol D_u$ be the relative density of $D_u$ with respect to $D_{\emptyset}$.
Suppose that the $2k$-sample likelihood ratio is bounded by $ \|\E_{u} \ol D_u^{\otimes 2 k}\|^2\le C^k$, and the $(1,4k)$-$\LDLR_m$ is bounded by $\|\E_{u} (\ol D_u^{\otimes m})^{\le 1,4k} - 1\|\le \eps$.
Then for any $q \ge 1$,
\[
\SDA\left(\calS, \frac{m}{q^{2/k}\eps^{1/k} k }\left(\frac{1}{\eps^{1/k} + \left(\frac{4e^2k(1+C)}{m}\right)}\right)\right) \ge q\,.
\]
\end{theorem}

We first will prove a lemma bounding the high-degree part of the LR in terms of its low-degree part.

\begin{lemma}\label{lem:gauss}
Let $\calS = \{D_u\}_{u \in S}$ be a set of identity-covariance Gaussian distributions, where $D_u = \calN(u, \Id_n)$ and $D_{\emptyset} =  \calN(0,\Id_n)$.
For each $u \in S$, let $\ol D_u$ be the relative density of $D_u$ with respect to $D_\emptyset$.
For any integers $d,k\ge 1$ with $k$ even,
$$\left\| \E_u (\ol D_u^{> d})^{\otimes k} \right\|^{2/k} \le \frac{1}{(d+1)!}\E_{u,v}\left[ \left(\langle \ol D_u^{\le 1}, \ol D_v^{\le 1} \rangle -1\right)^{2k(d+1)}\right]^{1/2k} \left(1 + \|\E_u \ol D_u^{\otimes 2k}\|^2\right)^{1/2k}.$$
\end{lemma}

\begin{proof}
We will exploit some properties of identity-covariance Gaussians. Let $\exp^{> d}(x) = \sum_{t = d+1}^\infty \frac{x^d}{d!}$ be truncation error of the degree-$d$ Taylor approximation of $\exp(x)$ about $0$.
In this setting, for each $u,v \in S$, it is shown in \cite{BKW19} (Theorem 2.6) that 
\begin{equation}
\langle \ol D_u^{> d}, \ol D_v^{> d} \rangle_{D_\emptyset} = \exp^{> d}(\langle u, v \rangle).\label{eq:gauss-form}
\end{equation}
By Taylor's theorem, we have that $\exp^{>d}(x)$ is bounded by
\[
\left|\exp^{> d}(x)\right| \le \left|\frac{x^d}{(d+1)!}\cdot \exp(\xi(x))\right|,
\]
For some function $\xi(x)$ with $\sgn(\xi(x)) = \sgn(x)$ and $|\xi(x)| \le |x|$. Thus, using that $k$ is even,
\allowdisplaybreaks
\begin{align*}
\left\| \E_u (\ol D_u^{> d})^{\otimes k} \right\|^2 &= \E_{u,v}\left[\left( \langle \ol D_u^{> d}, \ol D_v^{> d} \rangle\right)^{k}\right] \\
&= \E_{u,v}\left[\left|\exp^{>d}(\langle u,v \rangle)\right|^{k}\right]\\
&\le\E_{u,v}\left[\left|\frac{\langle u, v \rangle^{d+1}}{(d+1)!}\exp(\xi(x))\right|^{k}\right] \\
&\le\left(\frac{1}{(d+1)!}\right)^k\sqrt{\E_{u,v}\left[\langle u, v \rangle^{2dk+2k}\right]\E_{u,v}\left[\exp(\xi(x))^{2k}\right]} \\
&\le\left(\frac{1}{(d+1)!}\right)^k\sqrt{\E_{u,v}\left[\langle u, v \rangle^{2dk+2k}\right]\E_{u,v}\left[1+ \exp(x)^{2k}\right]} \\
&=\left(\frac{1}{(d+1)!}\right)^k\sqrt{\E_{u,v}\left[\left(\langle \ol D_u^{\le 1}, \ol D_v^{\le 1} \rangle - 1\right)^{2dk+2k}\right](1 + \E[\langle\ol D_u,\ol D_v \rangle^{2k}])}\,.
\end{align*}
The fourth line follows from Cauchy-Schwarz, and the fifth line uses that $\sgn(\xi(x)) = \sgn(x)$ and therefore $1 + \exp(x) \ge |\max(1,\exp(x))| \ge |\exp(\xi(x))|$.
The final line then follows from (\ref{eq:gauss-form}).
Substituting this back in for the above, we have our desired conclusion.
\end{proof}

\begin{proof}[Proof of Theorem~\ref{thm:gauss}]
We will show that a more general result holds given $\|\E_{u} (\ol D_u^{\otimes m})^{\le d,2k(d + 1)} - 1\|\le \eps$, and then set $d = 1$. By Lemma~\ref{lem:boosting}, we have that
$$\left\|\E_u (\ol D_u^{\le 1} -1)^{\otimes 2k(d+1)}\right\|^2 \le \left\|\E_u (\ol D_u^{\le d} -1)^{\otimes 2k(d+1)}\right\|^2 \le \frac{\eps^2}{\binom{m}{2k(d + 1)}}\,.$$
Therefore Lemma \ref{lem:gauss} implies that
\begin{align*}
\left\| \E_u (\ol D_u^{> d})^{\otimes k} \right\|^{2/k} &\le \frac{1}{(d+1)!} \cdot \frac{\eps^{1/k}}{\binom{m}{2k(d + 1)}^{1/2k}} \left(1 + C^k \right)^{1/2k} \\
&\le \frac{1 + C}{(d + 1)!} \cdot \frac{\eps^{1/k} (2k(d+1))^{d+1}}{m^{d+1}} \\
&\le (1 + C) \cdot \frac{\eps^{1/k} (2ke)^{d+1}}{m^{d+1}}
\end{align*}
using Stirling's approximation to the factorials and the fact that $\binom{a}{b} \ge (a/b)^b$. Since $(d,k)$-$\LDLR_m \le (d,2k(d+1))$-$\LDLR_m$, we also have that $\|\E_{u} (\ol D_u^{\otimes m})^{\le d, k} - 1 \| \le \eps$. Now applying Theorem \ref{thm:low-deg-sq} to the $(d,k)$-$\LDLR_m$ and then setting $d = 1$ completes the proof of the theorem.
\end{proof}

\subsection{Product Measures Over the Boolean Hypercube}
\begin{theorem}\label{thm:prod}
Let $k$ be an even integer.
Let $\calS = \{D_u\}_{u \in S}$ be a set of product distributions over the $n$-dimensional hypercube.
Let $D_{\emptyset}$ be any product measure over $\{\pm 1\}^n$ with no fixed coordinates, and let $\ol D_u$ be the relative density of $D_u$.
Suppose that the $2k$-sample likelihood ratio is bounded by $ \|\E_{u} \ol D_u^{\otimes 2 k}\|^2\le C^k$, and the $(1,4k)$-$\LDLR_m$ is bounded by $\|\E_{u} (\ol D_u^{\otimes m})^{\le 1,4k}\|\le \eps$.
Then for any $q \ge 1$,
\[
\SDA\left(\calS, \frac{m}{q^{2/k}\eps^{1/k} k }\left(\frac{1}{\eps^{1/k} + \frac{16k C^{1/2}}{m}}\right)\right) \ge q\,.
\]
\end{theorem}

We again will prove a lemma bounding the high-degree part of the LR in terms of its low-degree part.

\begin{lemma}\label{lem:prod}
Let $\calS = \{D_u\}_{u \in S}$ be a set of product distributions over the $n$-dimensional hypercube.
Let $D_{\emptyset}$ be any product measure over $\{\pm 1\}^n$ with no fixed coordinates, and let $\ol D_u$ be the relative density of $D_u$.
For any integers $d,k\ge 1$ with $k$ even,
$$\left\| \E_u (\ol D_u^{> d})^{\otimes k} \right\|^2 \le \E_{u,v\sim S}\left[ \left(\langle \ol D_u^{\le 1}, \ol D_v^{\le 1} \rangle -1\right)^{2k(d+1)}\right]^{1/2}  \left\|\E_{u\sim S} \ol D_u^{\otimes 2k}\right\|.$$
\end{lemma}

\begin{proof}
As in Lemma \ref{lem:gauss}, $\| \E_u (\ol D_u^{> d})^{\otimes k} \|^2 = \E_{u,v} \langle \ol D_u^{> d}, \ol D_v^{> d} \rangle^k$. 
We let $\chi_i(x)$ be the unique function such that $\E_{x \sim D_{\emptyset}} \chi_{i}(x) = 0$, $\E_{x \sim D_{\emptyset}} \chi_i(x)^2 = 1$, and $\chi_i(x) \ge 0$ when $x_i = 1$.
For convenience, we associate each $u  \in S$ with a vector $u \in \R^n$ as follows:
if $D_u$ is the (unique) product measure $P_u$ over $\{\pm 1\}^n$ with $\E_{x \sim D_u} [\chi_i(x)] = u_i$. 
Let $e_k:\R^n \to \R$ be the $k$th elementary symmetric polynomial:
\[
e_k(x) = \sum_{\substack{S \subset [n]\\|S| = k}} \prod_{i = 1}^k x_{i}.
\]
For any $t \in [n]$, using standard Fourier analysis over the Boolean hypercube one can see that
\[
\langle \ol D_u^{=t}, \ol D_{v}^{=t} \rangle 
= \sum_{\substack{S \subseteq [n]\\|S| = t}} \E_{D_u}\left[\prod_{i \in S} \chi_i(x) \right]\E_{D_v}\left[\prod_{i \in S} \chi_i(x) \right] 
= \sum_{\substack{S \subseteq [n]\\|S| = t}}\prod_{i \in S} u_i v_i 
=  e_t(u \circ v),
\]
where $u \circ v \in \R^n$ is the Hadamard (or ``entrywise'') product of $u$ and $v$.
So we may re-express
\begin{align}
\langle \ol D_u^{>d}, \ol D_v^{>d} \rangle = \sum_{t = d+1}^n e_t(u \circ v).\label{eq:prod-hd}
\end{align}

We will exploit the following claims regarding polynomials in $u \circ v$ and the elementary symmetric polynomials:
\begin{claim} \label{claim:monomial}
Let $A$ be any multiset of elements from $[n]$, and for a vector $x \in \R^n$ denote by $x^A = \prod_{i \in A} x_i$.
Then, for any set $S \subset \R^n$,
\[
\E_{u,v\sim S} (u \circ v)^A = \E_{u,v \sim S} \prod_{i \in A} u_i v_i = \left(\E_{u \sim S} u^A \right)^2\ge 0.
\]
\end{claim}
The proof of Claim~\ref{claim:monomial} is evident from the expression above. 
One consequence is the following:
\begin{claim}\label{claim:prod}
Let $p:\R^{n+1} \to \R$ be any polynomial which is a sum of monomials with non-negative coefficients, let $S\subset \R^n$ and for each $u\in S$ let there be a $\lambda_u\in \R$.
Then for any integers $a,b \ge 1$,
\[
\E_{u,v}\left[ e_{a + b}(u \circ v) \cdot p(u \circ v)\right] \le \E_{u,v} \left[e_a(u \circ v)\cdot e_b(u \circ v)\cdot p(u \circ v)\right].
\]
\end{claim}
\begin{proof}

For any $x \in \R^n$, we can expand the product
\[
e_{a}(x) e_b(x)
= \sum_{\substack{A \subset [n]\\|A| = a}} x^A \sum_{\substack{B \subset [n]\\|B| = b}} x^B
= \sum_{i = 0}^{\min(a,b)} \sum_{\substack{I \subset [n]\\|I|=i}} x^{2I} \sum_{\substack{S,T \subset [n]\setminus I\\|S| = a-i, |T| = b -i\\|S \cap T| = 0}}x^{S \cup T},
\]
where we have arranged the second sum according to the intersection size $i$ that a monomial from $e_a$ and a monomial from $e_b$ may have.
Extracting the $i = 0$ summand, we have that
\[
\sum_{\substack{S,T \subset [n]\\|S| = a, |T| = b, |S \cap T| = 0}} x^{S \cup T} = \binom{a + b}{a} e_{a+b}(x),
\]
since each set $S \cup T$ is counted in this sum $\binom{a+b}{a}$ times.
Write $p(x') = \sum_{C} \hat p_C \cdot (x')^C$ where the sum is over monomials.
Therefore we have that
\[
e_a(x)\cdot e_b(x) \cdot p(x') = \binom{a+b}{b} e_{a+b}(x) \cdot p(x') +  q(x) p(x'), 
\]
where $q(x)$ (the summation over over $i > 0$) is a sum of monomials with non-negative coefficients. 
The claim now follows from taking expectations on both sides and applying Claim~\ref{claim:monomial}.
\end{proof}
Given these facts and (\ref{eq:prod-hd}), we can deduce the following upper bound:
\begin{align*}
\E_{u,v} \left[\langle \ol D_u^{>d}, \ol D_v^{>d}\rangle^k\right]
&= \E_{u,v} \left[\left(\sum_{t = d+1}^n e_t(u \circ v)\right)^k\right]\\
&= \E_{u,v} \left[\sum_{t = d+1}^n e_t(u \circ v)\cdot \left( \sum_{t = d+1}^n e_t(u \circ v)\right)^{k-1}\right]\\
&\le \E_{u,v} \left[\sum_{t = d+1}^n e_{d+1}(u \circ v)\cdot e_{t-(d+1)}(u \circ v)\cdot \left( \sum_{t = d+1}^n e_t(u \circ v)\right)^{k-1}\right]\\
&= \E_{u,v} \left[\left( e_{d+1}(u \circ v)\cdot \sum_{s = 0}^{n-d-1} e_{s}(u \circ v)\right) \left( \sum_{t = d+1}^n e_t(u \circ v)\right)^{k-1}\right],\\
\intertext{Where to obtain the inequality we have applied Claim~\ref{claim:prod} with $p =\left(\sum_{t = d+1}^n e_t(u \circ v)\right)^{k-1}$, $a=d+1$, and $b=t-d-1$.
Repeating this for the $k-1$ remaining powers, we have}
&\le \E_{u,v} \left[\left(e_{d+1}(u \circ v)\cdot \sum_{s = 0}^{n-d-1} e_s(u \circ v)\right)^k\right]\\
&\le \E_{u,v} \left[\left(e_{d+1}(u \circ v)\cdot \sum_{s = 0}^{n} e_s(u \circ v)\right)^k\right]\\
&= \E_{u,v}\left[ \left(e_{d+1}(u \circ v) \cdot \langle \ol D_u, \ol D_v\rangle \right)^k\right],
\end{align*}
where in the second-to-last line we have used Claim~\ref{claim:monomial} to add the terms for $s = n-d,\ldots,n$ as they contribute positively to the expectation.
Applying Cauchy-Schwarz to the conclusion of the above display, 
\[
\E_{u,v} \left[\langle \ol D_u^{>d}, \ol D_v^{>d}\rangle^k\right]
\le \sqrt{\E_{u,v} \left[e_{d+1}(u \circ v)^{2k}\right] \E_{u,v}\left[\langle D_u, D_v \rangle^{2k}\right]}
\le \sqrt{\E_{u,v} \left[\left(\langle \ol D_u^{\le 1}, \ol D_v^{\le 1} \rangle -1 \right)^{2k(d+1)}\right]} \| \E_u \ol D_u^{\otimes 2k}\|,
\]
where we have used that $\E_{u,v}\left(\langle \ol D_u^{\le 1}, \ol D_v^{\le 1} \rangle -1 \right)^{2k(d+1)} \ge \E_{u,v}\left(e_{d+1}(u \circ v)\right)^{2k}$, again by applying Claim~\ref{claim:monomial} in a similar manner to the proof of Claim~\ref{claim:prod}.
This completes the proof.
\end{proof}

\begin{proof}[Proof of Theorem~\ref{thm:prod}]
As in the proof of Theorem~\ref{thm:prod}, we will show that a more general result holds given $\|\E_{u} (\ol D_u^{\otimes m})^{\le d,2k(d + 1)} - 1\|\le \eps$, and then set $d = 1$. By Lemma~\ref{lem:boosting}, we have that
$$\left\|\E_u (\ol D_u^{\le 1} -1)^{\otimes 2k(d+1)}\right\|^2 \le \left\|\E_u (\ol D_u^{\le d} -1)^{\otimes 2k(d+1)}\right\|^2 \le \frac{\eps^2}{\binom{m}{2k(d + 1)}}\,.$$
The same application of Lemma~\ref{lem:boosting} as in the proof of Theorem~\ref{thm:prod} and Lemma \ref{lem:gauss} imply that
$$\left\| \E_u (\ol D_u^{> d})^{\otimes k} \right\|^{2/k} \le\frac{C^{1/2} \eps^{1/k}}{\binom{m}{2k(d + 1)}^{1/2k}} \le \frac{C^{1/2}\eps^{1/k} (2k(d+1))^{d+1}}{m^{d+1}}$$
using the fact that $\binom{a}{b} \ge (a/b)^b$. As in the proof of Theorem~\ref{thm:prod}, we have that $\|\E_{u} (\ol D_u^{\otimes m})^{\le d, k} - 1 \| \le \eps$. Applying Theorem \ref{thm:low-deg-sq} to the $(d,k)$-$\LDLR_m$ and then setting $d = 1$ completes the proof of the theorem.
\end{proof}

\section{Diluting the Power of Statistical Queries via Cloning: Leveling the Playing Field} \label{s:cloning}

As discussed in Remark \ref{rem:multi-sample}, many average-case problems of interest such as planted clique and tensor PCA do not have a natural notion of samples.
In contrast, the SQ framework requires problem formulations involving multiple samples. 
In this section we describe how to convert certain single sample problems into multiple-sample problems, and then address the question of how to choose the number of samples so that the SQ complexity of the resulting problem captures the computational complexity of the original problem (as predicted by e.g. low-degree tests).

\medskip
\noindent {\bf Multi-sample formulations of single-sample problems.}
The idea is to apply an SQ bound to a ``diluted'' or ``cloned'' version of the single-sample problem, wherein each ``dilute'' sample carries little information compared to a single sample.
When multiple cloned samples can be combined into one original sample in polynomial time, a lower bound against the cloned problem implies a lower bound against the original problem (within the framework of polynomial time algorithms).

We first state a general and somewhat obvious sufficient condition for the existence of an average-case reduction from a multi-sample problem to a single-sample problem. A computational lower bound for the multi-sample problem is then transferred to the single-sample problem via the reduction.  

\begin{fact}
Let $D_\varnothing$ and $\calS = \{D_u\}_{u \in S}$ be distributions on $\R^N$ and let $\mu$ be a prior over $S$. 
Let $\{P_\theta\}_{\theta\in \Omega}$ be an exponential family of distributions on $\R^N$ with sufficient statistic $T$ that can be computed in time polynomial in the size of its input. Suppose that for each distribution $D\in \{D_\varnothing\}\cup \calS$, there is a $\theta = \theta(D)$ such that if $Y_1,\dots, Y_m\stackrel{i.i.d.}\sim P_\theta$ then $T(Y_1,\dots, Y_m)\sim D$. Then if there is no polynomial time algorithm testing between $H_0: (Y_1,\dots, Y_m)\sim P_{\theta(D_\varnothing)}^{\otimes m}$ versus $H_1:(Y_1,\dots, Y_m)\sim P_{\theta(D_u)}^{\otimes m}$ where $u \sim \mu$, with Type I$+$II error $1 - \eps$, then the same is true for the original testing problem.
\end{fact}

If one can efficiently generate $m$ samples $Y_1,\dots, Y_m$ as described in the fact just above given the single sample $X$, then the mapping is \emph{invertible}, which implies that no signal is lost and the single and multi-sample versions of the problem are \emph{computationally and statistically equivalent}.
Note that by the definition of sufficient statistic it is possible to generate samples with given sufficient statistic, but it is not always possible to do so efficiently (assuming the widely believed computational complexity conjecture $\mathrm{RP}\neq \mathrm{NP}$) \cite{BGS14b,Montanari2014}.

We now describe two examples where simple randomized algorithms show that it is possible to generate samples efficiently given a sufficient statistic.
In the first, the data consists of unit variance Gaussians, for which the mean is the sufficient statistic.

\begin{lemma}[Gaussian Cloning] \label{lem:g-clone}
There is a randomized algorithm taking as input a real number $x$ and outputting $m$ independent random variables $Y_1,\dots, Y_m$ such that for any $\mu\in \R$
if $x\sim \calN(\mu,1)$, then $Y_i\sim \calN(\mu/\sqrt{m}, 1)$. 
\end{lemma}
We will give the proof in Appendix~\ref{app:clone}.
In the second example, we show that the planted clique problem has an equivalent multi-sample version. Given a subset $U \subseteq [n]$, let $\calG(n, U, \gamma)$ denote the distribution of $\calG(n, \gamma)$ conditioned on the vertices in $U$ forming a clique (again see Appendix~\ref{app:clone} for a proof). This reduction is a mild variant of Bernoulli Cloning in \cite{brennan18}, which corresponds to the regime where $m = O(1)$.

\begin{lemma}[Planted Clique Cloning]\label{lem:p-clone}
There is an algorithm that when given $m$ independent samples from $\calG(n,U,\gamma)$ for any $U \subseteq [n]$, efficiently produces a single instance distributed according to $\calG(n,U,\gamma^m)$. 
	Conversely, there is an efficient algorithm taking a graph as input and producing $m$ random graphs, such that given an instance of planted clique $\calG(n,U,\gamma)$ with unknown clique position $U$, produces $m$ independent samples from $\calG(n,U,\gamma^{1/m})$. 
\end{lemma}

The same equivalence holds in the hypergraph formulation of planted clique.
The Gaussian cloning algorithm runs in $\text{poly}(m)$ time given access to an oracle for sampling standard normal random variables.
When applied entry-wise, this cloning procedure can be used to show average-case equivalences between single and multi-sample variants of problems with Gaussian noise such as tensor PCA and the spiked Wigner model. 
Furthermore, increasing the number of samples from $1$ to $m$ dilutes the level of signal in the problem exactly by a factor of $1/\sqrt{m}$.
The planted clique cloning algorithm runs in $\text{poly}(m, n)$ randomized time.
This again shows a precise tradeoff between the level of signal and number of samples $m$ -- as the ambient edge density varies as $\gamma$ to $\gamma^{1/m}$ with the number of samples $m$.

\paragraph{Choosing the number of samples.}

The number of queries used by statistical query algorithms is a proxy for runtime. However,
the statistical query framework allows queries that cannot be computed in polynomial time, and for this reason can lead to predictions that do not correspond to polynomial time algorithms. 
For example, a naive application of the statistical query framework in \cite{FGRVX} to the planted clique problem treats an instance as a single sample from the planted clique distribution has a single-query $\VSTAT(\frac{1}{3})$ algorithm, using the $\{0,1\}$ query: {\em does the graph $G$ have a clique of size at least $k$?}

For this reason, prior SQ lower bounds for planted clique \cite{FGRVX} consider instead the planted \emph{biclique} problem in a bipartite graph, and furthermore, assumed that i.i.d. data is generated by observing a random column from the adjacency matrix. While this is an interesting problem to study, it is not known to be equivalent to planted clique, the original problem of interest. 
More troubling is that this approach of generating samples fails badly for hypergraph planted clique. If one views a sample as a random slice of the adjacency tensor, then statistical query algorithms can perform an exhaustive search over what amounts to an instance of planted clique and this succeeds if at least one sample contains a planted clique, which occurs with positive probability once one has $n/k$ samples. 

The methodology described earlier in this section of converting a single-sample problem to many-sample problem is applicable to a broad class of problems and thus gives a unified way of addressing a variety of problems within the SQ framework.
 If we are free to study multi-sample versions of problems, it remains to specify the correct number of samples in order to obtain meaningful predictions within the SQ framework. 
As noted in the introduction, a prescription is suggested by Theorem~\ref{thm:main-intro}: we should dilute the signal so that each the problem is information-theoretically unsolvable from $O(1)$ samples.  Concretely, we convert to a hypothesis testing problem with $m$ samples, $D_{\emptyset}^{\otimes m}$ vs. $D_{u}^{\otimes m}$ where $\|\E_{u} \ol D_u \| = O(1)$.

%!TEX root=main.tex

\section{Example Applications}\label{sec:examples}

\subsection{Tensor PCA}\label{sec:tpca}

\begin{problem}[Tensor Principal Components Analysis (PCA)]
For $n,r$ positive integers, $\lambda \in \R$, and $S  = \{\pm \frac{1}{\sqrt n}\}^n$, the {\em $n$-dimensional $r$-tensor PCA with signal strength $\lambda$} problem is the following many-vs-one hypothesis testing problem:
\begin{itemize}
\item Null: a tensor in $(\R^n)^{\otimes r}$ with independent standard Gaussian entries, $D_{\emptyset} = \calN(0,\Id_{n^{r}})$.
\item Alternate: uniform mixture of $D_u = \calN(\lambda \cdot u^{\otimes r}, \Id_{n^r})$ over $u \in S$.
\end{itemize}
\end{problem}
Variations on the tensor PCA problem are possible; for example one may insist that the tensors be symmetric, or that $S$ be a different subset of $\calS^{n-1}$.
 
\begin{claim}
\torestate{
\label{claim:tpca-k}
For any integers $k,n$, and $r \ge 2$ satisfying $k\lambda^2 < \frac{n}{2}$, the $k$-sample likelihood ratio for the $n$-dimensional $r$-tensor PCA problem with signal strength $\lambda$ is bounded by
\[
\left\|\E_{u \sim S} \ol D_u^{\otimes k}\right\|^2 \le \sqrt{\frac{2\pi}{1-\frac{2k\lambda^2}{n}}}. 
\]
}
\end{claim}
We prove this claim in Appendix~\ref{sec:tpca-app}.

\begin{claim}
\torestate{
\label{claim:tpca-ldlr}
For any integers $n,r,k,m$ and real number $\lambda$ which satisfy $2 e m \lambda^2 k^{(r-2)/2} \le n^{r/2}$, the $(1,k)$-$\LDLR_m$ for the $m$-sample, dimension-$n$ tensor PCA problem with signal strength $\lambda$ is bounded by
\[
\left\|\E_u (\ol D_u^{\otimes m})^{\le 1,k} \right\|^2 \le 2\frac{e^{r+1} m \lambda^2 k^{(r-2)/2}}{n^{r/2}}
\]
}
\end{claim}
The proof is a straightforward calculation which appears in \cite{hopkins2017power,BKW19}---these works consider the single-sample version, but it is not difficult to see that their bounds imply ours. 
For completeness we give a full proof in Appendix~\ref{sec:tpca-app}.
Together these claims are sufficient to deduce the following Corollary of Theorem~\ref{thm:gauss}.
\begin{corollary}
\label{cor:tpca-main}
For integers $k,n,m,r$ and real numbers $\lambda,\delta$ with $\delta \in (0,1)$ satisfying 
\[
\left|\lambda\right| \le \min \left(\sqrt{\left(\frac{n}{(4k)^{(r-2)/r}}\right)^{r/2} \frac{1}{2 e m}}, \sqrt{(1-\delta)\frac{n}{4k}}\right), \text{ and } 4e^2k\left(1 + \left(\frac{2\pi}{\delta}\right)^{1/k}\right) \le \frac{m}{2},
\]
then for the $n$-dimensional $r$-tensor PCA problem with signal strength $\lambda$, for all $q \ge 1$, $\SDA(\frac{m}{8 q^{2/k}k}) \ge q$.
\end{corollary}
\begin{proof}
By Claims~\ref{claim:tpca-k} and \ref{claim:tpca-ldlr} and our assumptions, we have that
\[
\left\|\E_u (\ol D_u)^{\otimes 2k}\right\|^2 \le \sqrt{\frac{2\pi}{1-\frac{4k\lambda^2}{n}}} \le \sqrt{\frac{2\pi}{\delta}}, 
\qquad \left\|\E_u (\ol D_u^{\otimes m})^{\le 1,4k} -1 \right\|^2 \le 2\frac{em\lambda^2 (4k)^{(r-2)/2}}{n^{r/2}} \le 1.
\]
We instantiate Theorem~\ref{thm:gauss} with $C = \left(\frac{2\pi}{\delta}\right)^{1/k}$ and $\eps = 1$, and using our assumption on $\delta$ we have our conclusion.
\end{proof}

\paragraph{Comparison with prior work and predictions.}
In the literature, it is most common to consider the single-sample version of tensor PCA; for translations' sake, notice that $m$ samples from $\calN(\lambda u^{\otimes r}, \Id_{n^r})$ are equivalent to a single sample from $\calN(\sqrt{m}\lambda u^{\otimes r}, \Id_{n^r})$, since the sum of the samples is a sufficient statistic.
So we compare the $m$-sample problem to the single-sample hypothesis testing problem with signal strength $\sqrt{m}\lambda$.
Similarly, we compare the $\VSTAT(M)$ to the single-sample hypothesis testing problem with signal strength $\sqrt{M}\lambda$.

Applying this transformation, the best $n^k$-time algorithms for the $n$-dimensional $r$-tensor PCA problem requires signal strength $\sqrt{m}\lambda \ge \tilde \Omega\left(\sqrt{k}\left(\frac{n}{k}\right)^{r/4}\right)$ \cite{BGL17, RRS17, wein2019kikuchi}.
To see that this is consistent with the obtained $\VSTAT(M)$ bound with $M = \frac{m}{8 e k q^{2/k}}$, note that by Theorem~\ref{thm:testing} our bound implies that any $q = 2^k$-query algorithm requires the  ``adjusted signal strength'' to satisfy either $\lambda^2 k = \Omega(\sqrt{n})$ (which we will discuss below) or 
\[
\sqrt{M}|\lambda| \ge \left(\frac{n}{(4k)^{(r-2)/r}}\right)^{r/4}\sqrt{\frac{1}{16 e k q^{2/k}}} = \Omega\left(\frac{1}{2}\left(\frac{n}{k}\right)^{r/4}\right).
\] 
In the $k  \gg \log n$ regime, this is equivalent to the performance of the best-known algorithms up to a factor of $\tilde O(\sqrt{k})$.

We remark as well that the condition $\lambda^2 k < O(\sqrt{n})$ is necessary to rule out statistical query algorithms which use brute force on individual samples. 
If $\lambda^2 > 100n$, then there is a single-query SQ algorithm for the many-vs-one hypothesis testing problem: for a given sample $T\in (\R^n)^{\otimes r}$, simply query whether there exists some vector $x \in \{\pm \frac{1}{\sqrt{n}}\}^n$ which achieves $|\langle x^{\otimes r},T \rangle| \ge \frac{1}{2}\lambda$. 
When $|\lambda| \ge 10\sqrt{n}$,\footnote{No effort has been made to optimize the constants, which may be improved using, e.g., chaining arguments} it is easy to see that for $T \sim D_{\emptyset}$ this query will return false with high probability; this follows from the fact that $\langle x^{\otimes r}, T \rangle \sim \calN(0,\Id_{n^r})$.
On the other hand, for any $T \sim D_u$, this query will return true with high probability for similar reasons.

\subsection{Planted Clique and Planted Dense Subgraph}
\label{sec:pc-ex}

In this section, we consider several formulations of planted clique (PC) and planted dense subgraph (PDS). 
We begin by using our results to reproduce SQ lower bounds for ``bipartite'' formulations previously considered in the SQ literature \cite{FGRVX}, and then give new SQ lower bounds for non-bipartite multi-sample formulations.

\subsubsection{Bipartite Models}

The classical planted clique problem is a single-sample problem, which makes it incompatible with the SQ framework.
In an effort to address the complexity of the PC problem, the authors of \cite{FGRVX} give an SQ lower bounds for the following related problem: ``bipartite planted clique'' where each column of the resulting adjacency matrix is treated as an i.i.d. sample from a mixture distribution.

\begin{problem}[Bipartite Planted Dense Subgraph/Planted Clique]
Given $K, N \in \mathbb{N}$ and $0 < q < p \le 1$, {\em bipartite planted dense subgraph} with edge densities $p$ and $q$ is the following simple-vs-simple hypothesis testing problem:
\begin{itemize}
\item Null: independent Bernoulli random variables $D_{\emptyset} = \Ber(q)^{\otimes N}$.
\item Alternate: the mixture of $D_u = \frac{K}{N} \cdot D_u' + \left( 1 - \frac{K}{N} \right) \cdot \Ber(q)^{\otimes N}$ over random subsets $u \subseteq [N]$, sampled by including each element of $[N]$ in $u$ independently with probability $K/N$. Here, $D_u'$ is the distribution of $x \in \{0, 1\}^N$ with independent entries and $\Pr[x_i = 1] = p$ if $i \in u$ and $\Pr[x_i = 1] = q$ otherwise.
\end{itemize}
The {\em bipartite planted clique} problem is the bipartite PDS problem with $p = 1$.
\end{problem}

\paragraph{LDLR and $k$-sample LR bounds.} The following claims carry out standard computations to identify the relevant quantities needed to apply our main theorems. These calculations are deferred to Appendix \ref{sec:pc-app}. Let $\mu$ denote the distribution over $u$ described in the alternate hypothesis above.

\begin{claim}
\torestate{
\label{claim:bpds-ldlr}
For any $K, N, k, d, m \in \mathbb{N}$, define $\gamma = \tfrac{(p - q)^2}{q(1 - q)}$. 
Then the $(d, k)$-$\textnormal{LDLR}_m$ for bipartite PDS is bounded $\| \E_{u \sim \mu} (\ol D_u^{\otimes m})^{\le d, k} - 1 \| = O_N(1)$ if
$$\frac{K^2}{N} \cdot \max\left\{ \frac{m}{N}, (1 + \gamma)^k \right\} \le 1 - \Omega_N(1).$$
}
\end{claim}

\begin{claim}
\torestate{
\label{claim:bpds-k}
For any $K, N, k \in \mathbb{N}$, the $k$-sample LR is bounded by $\|\E_{u \sim \mu} \ol D_u^{\otimes k}\| = O_N(1)$ if
$$\frac{K^2}{N} \cdot \max\left\{ \frac{k}{N}, (1 + \gamma)^k \right\} \le 1 - \Omega_N(1)$$
where $\gamma = \tfrac{(p - q)^2}{q(1 - q)}$.
}
\end{claim}

\paragraph{Implications of our results.} Given these computations, we now can deduce the following implication of Corollary~\ref{cor:rand-rest}.

\begin{corollary} \label{cor:bpds}
Suppose that $K = \Theta(N^{1/2 - \delta})$ for some small constant $\delta > 0$ and $0 < q < p \le 1$ are constants. Then for bipartite PC and PDS with $N$ vertices, edge densities $0 < q < p \le 1$ and planted dense subgraph size $K$, it holds that $\SDA(N) = N^{\omega(1)}$.
\end{corollary}

\begin{proof}
Let $T$ be the noise operator that resamples independently from $\Ber(q)$, so $T$ is a $(1,0)$-operator. Note that bipartite PDS with $K = \Theta(N^{1/2 - \delta})$ can be realized as a random restriction with noise operator $T$ of bipartite PDS with $K = \Theta(N^{1/2 - \delta/2})$, restriction probability $s/N = N^{-\delta/2}$ and noise parameter $\rho = 0$. Suppose that $d, k = \Theta((\log N)^{c_1})$ where $c_1 \in (0, 1)$ and $d/k \sim c_2$ where $c_2$ is a sufficiently large constant. If again $m = \Theta(N^{1 + \delta})$, then the parameters for both the restricted and unrestricted bipartite PDS instances satisfy condition (1) in Claims \ref{claim:bpds-ldlr} and \ref{claim:bpds-k}. Now consider applying Corollary \ref{cor:rand-rest} with dimension lower bound $q' \sim 2^{k (\log N)^{c_3}}$ for some constant $c_3 \in (1 - c_1, 1)$. If $c_2$ is sufficiently large, then $(2s/N)^{2(d + 1)/k} m = o(1)$ and we have that $\SDA(N) \ge q' = N^{\omega(1)}$.
\end{proof}

\begin{remark}
Our generic noise-robustness result (Theorem~\ref{thm:noisy}) also recovers this lower bound in the case of bipartite PDS when $p < 1$. 
We choose $T$ to be the $(1,\rho)$-noise operator that resamples entries independently from $\Ber(q)$ with probability $1 - \rho = \tfrac{p - q}{1 - q}$. 
Then the distributions $D_u$ can be realized by applying $T$ entrywise to an instance of bipartite PC with edge density $q$. Note that the parameters $d \sim c_1 \log N$ for a sufficiently large constant $c_1$, $k \sim c_2 \log N$ for a sufficiently small constant $c_2$, $K = \Theta(N^{1/2 - \delta})$ and $m = \Theta(N^{1 + \delta})$ satisfy condition (1) in Claims \ref{claim:bpds-ldlr} and \ref{claim:bpds-k} for both the bipartite PDS instance in question and the bipartite PC instance before applying $T$. Now apply Theorem \ref{thm:noisy} with dimension lower bound $q' \sim 2^{k (\log N)^{c_3}}$ for some constant $c_3 \in (0, 1)$. If $c_1$ is sufficiently large, then $\rho^{2(d+1)} m = o(1)$ and it again follows that $\SDA(N) \ge q' = N^{\omega(1)}$. We also remark that, unlike in our previous applications of our main results where we set $q' = 2^k$, we must take $q' = 2^{\omega(k)}$ in this application of our noise-robustness theorem to show superpolynomial SQ lower bounds.
\end{remark}

\paragraph{Comparison to prior work and predictions.} Corollary \ref{cor:bpds} recovers the $K = \Theta(N^{1/2 - \delta})$ barrier from \cite{FGRVX} at which the SDA for bipartite PC/PDS with constant edge densities ceases to be $\textnormal{poly}(N)$. Despite being the consequence of a much more general theorem on random restrictions, our results for bipartite PC/PDS also nearly recover the precise SDA lower bounds from \cite{FGRVX}. In \cite{FGRVX}, for planted clique with edge density $1/2$, it is shown that $\SDA(\tfrac{N^2}{2^{\ell + 1} K^2}) \ge N^{2 \ell \delta}/3$ for all $\ell \le K$. Fine-tuning our parameter choices in Corollary \ref{cor:bpds} yields that $\SDA(\tfrac{N^{2 - \epsilon}}{2^{\ell + 1} K^2}) \ge N^{\Omega(\ell)}$ for any constant $\epsilon > 0$, which matches the bound from \cite{FGRVX} up to arbitrarily small polynomial factors in the sample complexity.

\subsubsection{Multi-Sample Hypergraph Planted Clique}

We now consider a variant of planted clique where the observations consist of multiple samples from the planted clique distribution. 
As discussed in Section \ref{s:cloning}, there is a natural tradeoff between the number of samples $m$ and edge density $q$ for which this variant has an average-case equivalence with ordinary PC. 
In this section, we will treat a generalization of this variant to $s$-uniform hypergraphs (including the case $s=2$ corresponding to simple graphs).

Let $\calG_s(N, q)$ denote the Erd\H{o}s-R\'{e}nyi distribution over $s$-uniform hypergraphs, where each $s$-subset of $[N]$ is included as a hyperedge independently with probability $q$. Given a subset $u \subseteq [N]$, let $\calG_s(N, u, q)$ denote the hypergraph where hyperedges among the vertices within $u$ are always included and all other hyperedges are included independently with probability $q$. Throughout this section, we will treat $s$ as a fixed positive integer constant.

\begin{problem}[Multi-Sample Hypergraph PC]
Given $s, K, N \in \mathbb{N}$ with $N \gg K \gg s \ge 2$ and $q \in (0, 1)$, the {\em  multi-sample $s$-uniform hypergraph planted clique problem} with edge density $q$ is the following hypothesis testing problem:
\begin{itemize}
\item Null: the Erd\H{o}s-R\'{e}nyi hypergraph $D_{\emptyset} = \calG_s(N, q)$.
\item Alternate: uniform mixture of $D_u = \calG_s(N, u, q)$ over $K$-subsets $u \subseteq [N]$.
\end{itemize}
\end{problem}

\paragraph{The complexity of multi-sample hypergraph PC as $m$ and $q$ vary.} To the best of our knowledge, multi-sample hypergraph PC has not been considered in this generality before. However, because of the average-case equivalence from Section \ref{s:cloning}, its complexity can be extrapolated exactly from that of ordinary hypergraph planted clique, i.e. when $m = 1$. For $m = 1$, its complexity conjecturally behaves as follows (as a function of $q$):
\begin{enumerate}
\item If $q$ is near constant with $N^{-o(1)} \le q \le 1 - N^{-o(1)}$, then the threshold at which polynomial-time algorithms begin to solve the distinguishing problem is $K^2 = N^{1 \pm  o(1)}$, which is consistent with the threshold in the classical setting of $q = \frac{1}{2}$. 
\item If $q$ is polynomially small with $q = \Theta(N^{-\alpha})$ for some $\alpha > 0$, then the clique number of $\calG_s(N, q)$ is constant and the problem begins to be easy when $K = \Theta(1)$.
\item If $q$ is very close to $1$ with $q = 1 - \Theta(N^{-\alpha})$ for some $\alpha \in (0, 1)$, then polynomial-time algorithms begin to solve the distinguishing problem at the shifted threshold $K^2 = \tilde{\Theta}(N^{1 + \alpha/s})$.
\end{enumerate}
The best known algorithm in the last regime simply counts the total number of edges. In the graph case when $s = 2$, it was shown in \cite{brennan18} that the PC conjecture with $q = 1/2$ implies a lower bound up to the barrier $K^2 = \tilde{\Theta}(N^{1 + \alpha/2})$ when $q = 1 - \Theta(N^{-\alpha})$. We remark that, in this regime, recovering the vertices in the planted clique is conjectured to be a harder problem that only becomes easy at larger values of $K$. Our focus in this section will be on the transition in the first parameter regime, when $N^{-o(1)} \le q \le 1 - N^{-o(1)}$.

As discussed in Section \ref{s:cloning}, there is a natural average-case equivalence between the single and multi-sample problems. Specifically, hypergraph PC with $m$ samples and edge density $q$ is equivalent to hypergraph PC with $m = 1$ sample and edge density $q^m$. Thus the parameter regime of interest corresponds to the $q$ with $\tfrac{1}{mN^{o(1)}} \le 1 - q \ll \tfrac{\log N}{m}$. We remark that at $1 - q = \Theta(\tfrac{\log N}{m})$, the distinguishing problem undergoes a (conjecturally sharp) transition to algorithmically easy. Specifically, taking the bit-wise AND of the edge indicators across the different samples corresponds to a single-sample instance of hypergraph PC with edge density $q^m = N^{-\Theta(1)}$, which can be solved in polynomial time whenever $K$ is a sufficiently large constant. 

As also discussed in Section \ref{s:cloning}, another concern when choosing $m$ is the existence of inefficient algorithms that can be implement with a small number of $\VSTAT(m)$. Let $h(G) \in \{0, 1\}$ be the indicator that $G$ has a clique of size $K$. While $h$ is NP-hard to compute, the single query of $h$ to a $\VSTAT(\Theta(1))$ oracle will solve the distinguishing problem unless $1 - q$ is sufficiently small. The expected number of cliques of size $K$ in $\calG_s(N, q)$ is
$$\binom{N}{K} q^{\binom{K}{s}} \le \exp\left( K \log N - \frac{1 - q}{q} \cdot \binom{K}{s} \right) = o(1)$$
as long as $1 - q \ge CK^{1 - s} \log N$ for a sufficiently large constant $C$. Thus unless $1 - q = O(K^{1 - s} \log N)$, Markov's inequality implies that $\calG_s(N, q)$ has no clique of size $K$ with probability $1 - o(1)$ and the SQ query of $h$ solves the distinguishing problem where no polynomial time algorithms are known to succeed. Thus to make the performance of SQ and polynomial-time algorithms comparable, it seems necessary to restrict to $q$ with $1 - q = O(K^{1 - s} \log N)$. As will be shown in Claim \ref{claim:mshpc-k}, this threshold is also roughly when the $k$-sample LR begins to have a constant-sized norm. To summarize this discussion, the natural choices of $m$ and $q$ are:
\begin{itemize}
\item sufficiently large $q$ with $q = 1 - O(K^{1 - s} \log N)$; and
\item $m$ such that $q$ lies in the range $\tfrac{1}{mN^{o(1)}} \le 1 - q \ll \tfrac{\log N}{m}$.
\end{itemize}
Note that this requires we take $m = \tilde{\Omega}(K^{s - 1})$ samples.

\begin{remark}
A different natural alternative formulation of hypergraph PC views the adjacency lists of individual vertices as independent samples, as in bipartite PC.
However, since each adjacency list is itself an $(s-1)$-uniform hypergraph, in this model a single-query SQ algorithm succeeds whenever $s > 2$: ask if the adjacency list contains a clique of size at least $K$.
For this reason, the bipartite model is not appropriate for the SQ framework.
\end{remark}

\paragraph{Choice of prior $\mu$.} We now discuss why the choice of prior $\mu$ over the the clique vertex set $u$ differs in the definitions of multi-sample hypergraph PC and bipartite PDS. The prior $\mu$ in which each vertex is included in the clique independently with probability $K/N$ was used in defining bipartite PDS because it is more convenient to work with when computing the LDLR, $k$-sample LR and applying our main results.

However, a subtle technical issues arises in multi-sample PC that precludes using this prior. The underlying problem is that $D_\emptyset$ and the mixture of $D_u$ induced by this prior do not necessarily converge in $\chi^2$ divergence even when they converge in total variation. This is because $\chi^2$ divergence is large if certain tail events have very mismatched probabilities while total variation is not. Specifically, the probability the mixture of $D_u$ contains a clique of size $t \gg K$ is at least $\Pr[\text{Bin}(N, K/N) \ge t]$, which is much larger than the probability that $D_\emptyset$ contains a clique of size $t$. This issue causes the average correlations defining SDA and the key quantity $\| \E_{u \sim \mu} \ol D_u^{\otimes k} \|$ to be very different between the two priors. Specifically, carrying out a similar computation as in Claim \ref{claim:mshpc-k} for the prior where each vertex is included with probability $K/N$ yields that $\| \E_{u \sim \mu} \ol D_u^{\otimes k} \|$ is only $O_N(1)$ for much smaller values of $\gamma$.

The important properties of the prior $\mu$ used in this section, where $u$ is a random $K$-subset of $[N]$, are that: (1) $u$ is symmetric; (2) the size of $u$ concentrates around $K$; and (3) the distribution of $|u|$ has very small upper tails. In particular, replacing $\mu$ with any prior that chooses a clique size from the interval $[CK, K]$ for some constant $C > 0$ and then chooses a random clique of this size would not affect the bounds in either Claim \ref{claim:mshpc-ldlr} or Claim \ref{claim:mshpc-k}.

\paragraph{LDLR and $k$-sample LR bounds.} The following claims bound the LDLR and $k$-sample LR in multi-sample hypergraph PC in order to verify the conditions needed to apply our main results. Their proofs are standard computations and deferred to Appendix \ref{sec:pc-app}. Let $\mu$ denote the uniform distribution over $K$-subsets $u \subseteq [N]$.

\begin{claim}
\torestate{
\label{claim:mshpc-ldlr}
For any $s, K, N, k, d, m \in \mathbb{N}$, the $(d, k)$-$\textnormal{LDLR}_m$ for multi-sample hypergraph PC satisfies that $\| \E_{u \sim \mu} (\ol D_u^{\otimes m})^{\le d, k} - 1 \| = O_N(1)$ if the following conditions are satisfied:
$$\gamma \cdot \max\{ m, (ksd)^s \} = O_N(1) \quad \text{and} \quad \frac{2^{sk} e^2 K^2}{N} = 1 - \Omega_N(1)$$
where $\gamma = \tfrac{1 - q}{q}$.
}
\end{claim}

\begin{claim}
\torestate{
\label{claim:mshpc-k}
For any $K, N, k \in \mathbb{N}$, the $k$-sample LR is bounded by $\|\E_{u \sim \mu} \ol D_u^{\otimes k}\| = O_N(1)$ if the following condition are satisfied:
$$K^2 \le 3N \quad \text{and} \quad \gamma \le \frac{1}{2k} \cdot K^{1 - s} \log \left( \frac{N}{K^2} \right)$$
where $\gamma = \tfrac{1 - q}{q}$.
}
\end{claim}

\paragraph{Implications of our results and comparison to conjectured complexity barriers.} We now can deduce the implications of our main theorems.

\begin{corollary} \label{cor:mshpc}
Suppose that $s$ is a fixed constant, $K = \Theta(N^{1/2 - \delta})$ for some small constant $\delta > 0$ and $q \in (0, 1)$ satisfies $q \ge 1 - c_1 K^{1 - s}$ for a sufficiently small constant $c_1 > 0$. Then for multi-sample hypergraph PC with $N$ vertices, clique size $K$ and edge density $q$, it holds that $\SDA\left(\Theta\left( \tfrac{1}{t(1 - q)} \right) \right) \ge N^{\Omega(\log t)}$ for any $t \ge (\log N)^{1 + \Omega(1)}$.
\end{corollary}

\begin{proof}
In multi-sample hypergraph PC, each $D_u$ is a product measure on the hypercube and Theorem \ref{thm:prod} applies. Consider setting the parameters $d = 1$, $k = c_2 \log N$ for a sufficiently small constant $c_2 > 0$, $K = \Theta(N^{1/2 - \delta})$ for a constant $\delta > 0$ and the number of samples $m$ to be $m = c_3/(1 - q)$ for some constant $c_3 > 0$. Note that $m$ is polynomially large in $N$. It now can be verified that, if $c_2$ is sufficiently small, then these parameters satisfy the conditions in Claim \ref{claim:mshpc-ldlr} and, if $c_1$ is sufficiently small, they also satisfy the condition in Claim \ref{claim:mshpc-k}. Now consider applying Theorem \ref{thm:prod} with SDA lower bound $q' = N^{\frac{c_2}{2}(\log t - \log \log N)}$. It can be verified that this implies $\SDA(\Theta(m/t)) \ge q'$, proving the corollary.
\end{proof}

Setting $t = (\log N)^{1 + \delta'}$ for some small $\delta' > 0$ recovers the predicted $K = \Theta(N^{1/2 - \delta})$ computational barrier in the SQ model for multi-sample hypergraph PC in the regime $\tfrac{1}{mN^{o(1)}} \le 1 - q \le O\left( \tfrac{1}{m} \right)$ of interest. It is worth noting that the loss of the $t = (\log N)^{1 + \Omega(1)}$ factor in $m$ on applying Theorem \ref{thm:prod} means that we cannot arrive at $m$ and $q$ satisfying that $1 - q = \Theta(1/m)$ exactly. Under the average-case equivalence from Section \ref{s:cloning}, this corresponds to single-sample hypergraph PC with exactly constant edge densities. However, this constraint does not affect the tightness of Corollary \ref{cor:mshpc}, as the resulting lower bound still corresponds to a single-sample instance of hypergraph PC with a nearly constant edge density in the range $N^{-o(1)} \le q \le 1 - N^{-o(1)}$ and thus $K^2 = N^{1 \pm o(1)}$ is still the conjectured computational barrier.

\begin{remark}
Our partial noise robustness results imply SQ lower bounds in multi-sample hypergraph PC, with a slightly different choice of the prior $\mu$. Let $\mu'$ be the prior formed by choosing a clique size according to $\text{Bin}(K, N^{-\delta})$ and then choosing a vertex set of this size uniformly at random from $[N]$ to be the planted clique, where $\delta > 0$ is a small constant. As in the discussion above, since $\text{Bin}(K, N^{-\delta})$ has zero probability mass above $K$, Claims \ref{claim:mshpc-ldlr} and \ref{claim:mshpc-k} can be adapted to accommodate this different prior. Furthermore, this prior concentrates will around $KN^{-\delta} = \Theta(N^{1/2 - 2\delta})$ if $K = \Theta(N^{1/2 - \delta})$.

If $T$ is the $(1,0)$ noise operator that resamples independently from $\Ber(q)$, then $m$-sample hypergraph PC with the prior $\mu'$ can be realized as a subtensor random restriction of the type in Theorem \ref{thm:random-rest-tensor} of $m$-sample hypergraph PC with the prior $\mu$. In particular, it can be realized with the noise operator $T$, restriction probability $N^{-\delta}$ and correlation parameter $\rho = 0$. Now consider setting the parameters $d = c_2^{-1} (\log N)^s$, $k = c_2 \log N$ for a sufficiently small constant $c_2 > 0$, $K = \Theta(N^{1/2 - \delta})$ for a constant $\delta > 0$ and the edge density $q$ and number of samples $m$ to again be $m = c_3/(1 - q)$. If $c_1$ and $c_2$ are sufficiently small, then the conditions in Claims \ref{claim:mshpc-ldlr} and \ref{claim:mshpc-k} are met. Adapting the arguments in these claims to accommodate $\mu'$ yields that the relevant LDLR and $k$-sample LR are both $O_N(1)$. Now consider applying Theorem \ref{thm:random-rest-tensor} together with Theorem \ref{thm:low-deg-sq}, similarly to as in Corollary \ref{cor:rand-rest}, again with the SDA lower bound $q' = N^{\frac{c_2}{2}(\log t - \log \log N)}$. If $c_2$ is sufficiently small, then $(N^{-\delta})^{2k^{-1} \sqrt[p]{(d + 1)/2}} m = o(1)$ and we recover the same lower bound as in Corollary \ref{cor:mshpc} for the prior $\mu'$.
\end{remark}

\subsection{Spiked Wishart PCA}
\label{sec:wishart}
The spiked Wishart model is a well-studied model for understanding sparse PCA.
We consider the following, standard version the problem.
As with the other problems considered here, many variations of this problem exist in the literature, see e.g.~\cite{perry2018optimality} for a more detailed discussion.
\begin{problem}[Sparse PCA with Wishart Noise]
For a positive integer $n$, $\rho \in [0, 1]$, and $\lambda \in [0, \infty)$, the {\em sparse PCA with Wishart noise} problem is the following many-vs-one hypothesis testing problem:
\begin{itemize}
\item Null: $m$ i.i.d. samples from the standard normal Gaussian, i.e. $D_{\emptyset} = \calN(0, \Id_n)$.
\item Alternate: $m$ i.i.d. samples from a Gaussian with randomly spiked covariance. 
Specifically, sample a vector $s$ via the following process.
First draw $s' \in \{-1, 0, 1\}^n$ so that each entry of $s'$ is independent and distributed as
\[
s_i' = \left\{ \begin{array}{ll}
         0 & \mbox{with probability $1 - \rho$};\\
        -1 & \mbox{with probability $\rho / 2$};\\
        +1 & \mbox{with probability $\rho / 2$}.\end{array} \right.
\]
\end{itemize}
Then, if $\|s' \|^2 > 2 \rho n$, let $s = 0$, otherwise let $s = \frac{1}{\sqrt{\rho n}} s'$.
Finally, draw $m$ samples from $D_s = \calN (0, \Id_n + \lambda ss^\top)$.
Denote the distribution over $s$ by $S_\rho$.
\end{problem}
\noindent
The choice of constant $2$ in this model is arbitrary and can be replaced by any constant larger than $1$.
By a Chernoff bound, for $\rho = \omega (1/n)$, $s \neq 0$ with high probability.
Note that this problem is naturally stated as a multi-sample problem.

Unfortunately, while the null hypothesis for this problem is the standard normal Gaussian, it does not cleanly fit into the framework of Theorem~\ref{thm:gauss}, as the alternate hypotheses are not additive shifts of $\calN(0, \Id_n)$.
However, the $(d, k)-\LDLR_m$ for this problem still has a nice form, which allows us apply our main theorem.

Recall the Hermite basis for $D_{\emptyset}^{\otimes t}$ is the set of polynomials over $(\R^n)^t$ given by $\{H_\alpha\}$, where $H_\alpha$ is parametrized by multi-indices $\alpha = (\alpha_1, \ldots, \alpha_t) \in (\mathbb{N}^n)^t$.
For any multi-index $\alpha \in \N^n$, and any $x \in \R^n$, let $x^\alpha = \prod_{i = 1}^n x_i^{\alpha_i}$.
Then, we have the following bound from~\cite{bandeira2019computational}:
\begin{lemma}[Lemma 5.8 in~\cite{bandeira2019computational}]
Let $(\alpha_1, \ldots, \alpha_t) \in (\mathbb{N}^n)^t$.
Then, we have:
\[
\left( \E_{u \sim S_\rho} \langle \ol D_u, H_\alpha \rangle \right)^2 = \left\{ \begin{array}{ll}
         \lambda^{\sum_{i = 1}^t |\alpha_i| } \cdot \prod_{i = 1}^t \frac{(|\alpha_i| - 1)!!}{\alpha_i!} \cdot \left( \E_{u \sim S_\rho} u^{\sum_{i = 1}^t \alpha_i} \right)^2 & \mbox{if $|\alpha_i|$ are even};\\
        0 & \mbox{otherwise}.\end{array} \right.
\]
\end{lemma}
\noindent
As a result, we have the following:
\begin{lemma}
\torestate{
\label{lem:wishart-low-deg}
Let $t, d \in \N$.
Suppose that $n \rho^2 \leq 1$, and that $d t \lambda \leq \rho n$. 
Then, we have:
\[
\left\| \E_{u \sim S_\rho} (\ol D_u^{\leq d} - 1)^{\otimes t} \right\|^2 \leq 2 \left( \frac{d^2 k \lambda}{\rho n} \right)^{2t}  \; .
\]
}
\end{lemma}
\noindent We prove Lemma~\ref{lem:wishart-low-deg} in Appendix~\ref{sec:wishart-app}.
Together with Claim~\ref{lem:many-to-one}, this immediately implies:
\begin{corollary}
\label{cor:wishart-low-deg}
Let $t, d$ be as in Lemma~\ref{lem:wishart-low-deg}.
Let $m$ be so that $m \leq \frac{\rho^2 n^2}{\lambda^2 d^4 k^2}$.
Then
\[
\left\| \E_{u \sim S_\rho} (\ol D^{\otimes m})^{\leq d, k} - 1 \right\|^2 \leq O(1) \; .
\]
\end{corollary}
\noindent
We now seek to bound the norm of the high degree part of the correlation.
To do so, we rely on the following lemma:
\begin{lemma}[\cite{bandeira2019computational}]
Let $\phi(x) = (1 - 4x)^{-1/2}$, and let $\phi^{\leq d} (x) = \sum_{\ell = 0}^d \binom{2\ell}{\ell} x^\ell$ and $\phi^{> d} (x) = \sum_{\ell = d + 1}^\infty \binom{2\ell}{\ell} x^\ell$ denote the low degree approximation and the approximation error of the degree $d$ Taylor approximation to $\phi(x)$ at zero, respectively.
Then 
\begin{align*}
 \left\| \E_{u \sim S_\rho} \ol D_u^{> d} \right\|^2 &= \E_{u, v \sim S_\rho} \left[ \phi^{> \lfloor d / 2 \rfloor} \left( \frac{\lambda^2 \langle u, v \rangle^2}{4 } \right) \right] \; .
\end{align*}
\end{lemma}
\noindent
As a result, we obtain the following bound:
\begin{lemma}
\torestate{
\label{cor:wishart-high-deg}
Assume that $2 n k (d + 1) \rho^2 \leq 1$.
For $\lambda < 1 / 2$ and $d$ even, we have:
\[
\left\| \E_{u \sim S_\rho} \left( \ol D_u^{> d} \right)^{\otimes k} \right\|^2 \leq \left( \frac{\lambda^2}{4 \rho n} \right)^{k(d + 1)} \; .
\]
}
\end{lemma}
\noindent The proof closely resembles the proof of Lemma~\ref{lem:gauss}, and we defer it to Appendix~\ref{sec:wishart-app}.
\noindent
Combining Corollary~\ref{cor:wishart-low-deg} and Lemma~\ref{cor:wishart-high-deg} with Theorem~\ref{thm:low-deg-sq}, we obtain:
\begin{corollary}
\label{cor:wishart}
Let $d, k \in N$.
Let $\lambda \leq 1/4$, let $\rho$ be so that $2 n k (d + 1) \rho^2 \leq 1$, let $m$ be so that $m \leq \frac{(\rho n)^2}{d^4 k^2 \lambda^2}$.
Then $\SDA (\calS, \widetilde{\Theta} (m / k)) \geq 2^k $.
\end{corollary}

\paragraph{Comparison to prior work and predictions.} The Wishart model for spiked PCA has two, well-studied regimes, the sparse PCA model, where the sparsity, governed by $\rho$, is sublinear in $n$, typically $ n \rho^2 \leq 1$, and the dense regime, when $\rho = \Theta (1)$.
In the dense regime, the celebrated BBP transition~\cite{baik2005phase} gives an exact prediction of when detection is computationally possible, and the computational limits in terms of the low degree likelihood ratio are known to exactly match these predictions~\cite{perry2018optimality,ding2019subexponential,bandeira2019computational}.
In particular, it is predicted that when $\rho$ is a fixed universal constant, recovery is possible if and only if $m \geq n / \lambda^2$.
While it is possible to plug in the machinery here with the LDLR bounds attained in~\cite{bandeira2019computational}, it appears to be an inherent limitation of the SDA framework for proving SQ lower bounds that it cannot predict exact (i.e. including constants) thresholds.
Thus, while we can attain SQ lower bounds matching the BBP transition up to constants, we cannot prove SQ lower bounds up to the transition.

For this reason, the calculations in the previous section primarily focus on the sparse regime.
The problem is well-studied in this setting, and the best known sample complexity for this problem is $m = \Omega \left( \frac{(\rho n)^2 \log n}{\lambda^2} \right)$~\cite{d2008optimal,berthet2013optimal}.
In contrast, information theoretically $m = \Omega \left( \frac{(\rho n) \log n}{\lambda^2} \right)$ samples suffice.
There is a slew of evidence~\cite{berthet2013complexity,hopkins2017power,brennan2019optimal} that suggests that this is the best possible.
Note that the SQ lower bounds and LDLR lower bounds we obtain witness this gap, up to logarithmic factors.
To the best of our knowledge, prior to our work there were no LDLR lower bounds for sparse PCA in the $\rho \leq 1/\sqrt{n}$ regime, and existing SQ lower bounds required $\lambda = o(1)$ and $\rho = n^{-7/8}$ \cite{wang2015sharp}.

\subsection{Testing Gaussian Mixture Models}
In this section, we prove LDLR bounds for robustly testing Gaussian Mixtures. 
We use the SDA bounds of \cite{diakonikolas2017statistical} in an almost black-box fashion (we must modify their proofs a little bit to account for the different notions of statistical dimension considered). 

\begin{problem}[Testing Gaussian Mixture Models]
For $n, s$ positive integers and $\eps \in (0,1)$, the {\em $(1-\eps)$-separated Gaussian $s$-mixture model testing problem} is the following hypothesis testing problem:
\begin{compactitem}
\item Null: $\calN(0,\Id_n)$
\item Alternate: uniform over $\calS = \{D_{U}\}_{U\in S}$ for some $S \subset \times_s \R^{n-1} $, where each $D_{U}$ for $U = u_1,\ldots,u_s$ is a mixture of $\calN(u_1, I),...,\calN(u_s,I)$ satisfying the conditions $d_{\mathrm{TV}}(D_{u,v},D_{\emptyset}) \ge 0.25$ and $d_{\mathrm{TV}}(\calN(u_i,I),\calN(u_j,I)) \ge 1-\eps$ for all $i\neq j \in [s]$.
\end{compactitem}
\end{problem}

In \cite{diakonikolas2017statistical}, the authors show lower bounds on the $\SDA_\times$ for this problem---however, because the lower bounds are for product-SDA, we must make some mild modifications to their proofs.
We use the following building blocks:

\begin{lemma}[Lemma 3.4 of \cite{diakonikolas2017statistical}]\label{lem:construction}
Suppose $A$ is a distribution over $\R$ which matches $m$ moments of $\calN(0,1)$. 
For each $u \in S^{n-1}$, define the distribution with probability density function $D_u(x) = A(\langle x, u \rangle) \cdot \gamma_{\perp u}(x)$, where $\gamma_{\perp u}$ is the projection of $D_{\emptyset} = \calN(0,\Id_n)$ orthogonal to $u$.
Letting $\overline D_u$ be the relative density of $D_u$ with respect to $D_{\emptyset}$, we have that for any $u,v \in S^{n-1}$,
\[
|\langle \ol D_u, \ol D_v \rangle -1| \le |\langle u,v \rangle|^{m+1} \cdot \|\ol A\|^2,
\]
for $\ol A$ the relative density of $A$ with respect to $N(0,1)$.
\end{lemma}

\begin{lemma}[Lemma 3.7 of \cite{diakonikolas2017statistical}]\label{lem:set}
For any $c \in (0,\frac{1}{2})$, there is a set $S$ of $2^{\Omega(n^c)}$ unit vectors in $\R^n$ so that for each $u,v \in S$ with $u \neq v$, $|\langle u,v \rangle| \le O(n^{c-1/2})$.
\end{lemma}

Now, we use the following propositions of \cite{diakonikolas2017statistical}, which selects a distribution $A$ for the GMM testing problem:

\begin{proposition}[Proposition 4.2 of \cite{diakonikolas2017statistical}]\label{prop:mix}
For any $\eps \in (0,1)$, $c \in (0,\frac{1}{2})$, and integer $s \ge 1$ there exists a distribution $A$ on $\R$ that is a mixture of $s$ Gaussians $A_1,\ldots,A_s$ with $d_{\mathrm{TV}}(A_i,A_j) \ge 1-\eps$ for all $i \neq j \in [s]$. 
Further, $\|\ol A\|^2 \le \exp(O(s))\log \frac{1}{\eps}$ and $A$ agrees with $N(0,1)$ on $2s-1$ moments, and if we construct $\{D_u\}_{u \in S}$ as described in Lemmas~\ref{lem:construction} and \ref{lem:set}, then each $D_u$ is a mixture of $s$ Gaussians and further for all $u,v \in S$, $d_{\mathrm{TV}}(D_u,D_v) \ge \frac{1}{2}$.
\end{proposition}

Putting these together, we have the following instance of the GGM testing problem:
\begin{problem}[ $(1-\eps)$-separated GGM testing instance from \cite{diakonikolas2017statistical}]\label{prob:gmm}
For $n,\ell$ positive integers and any $\eps \in (0,1)$, let $A$ be the mixture of $\ell$ Gaussians described in Proposition~\ref{prop:mix} and let $S$ be the subset of $S^{n-1}$ described in Lemma~\ref{lem:set} with $c = 0.26$. 
Consider the following instance of the $(1-\eps)$-separated Gaussian $\ell$-mixture model testing problem:
\begin{compactitem}
\item Null: $D_{\emptyset} = \calN(0,\Id_n)$
\item Alternate: Uniform over the set of distributions $\calS = \{D_u\}_{u \in S'}$, where $D_u(x) = A(\langle x,u \rangle)\cdot \gamma_{\perp u}(x)$ and $S'$ is the subset of $u \in S$ with $d_{\mathrm{TV}}(D_u, D_\emptyset) \ge \frac{1}{4}$ (note $|S'| \ge \frac{1}{2}|S|$).
\end{compactitem}
\end{problem}
We note that Problem~\ref{prob:gmm} is a valid instance of the $(1-\eps)$-separated Gaussian $\ell$-mixture testing problem:
since from Proposition~\ref{prop:mix} $A$ is a one-dimensional mixture of $\ell$ Gaussians with pairwise total variation distance $\ge 1-\eps$, each $D_u$ is also a mixture of $\ell$ Gaussians with pairwise total variation distance $\ge 1-\eps$.
Proposition~\ref{prop:mix} also guarantees that for each $u \neq v$, $d_{TV}(D_u,D_v) \ge \frac{1}{2}$.
By the triangle inequality, we have that $d_{\mathrm{TV}}(D_u, D_{\null}) + d_{\mathrm{TV}}(D_v,D_{\null}) \ge d_{\mathrm{TV}}(D_u,D_v) \ge \frac{1}{2}$, which implies that for at least half of $u \in S$, $d_{\mathrm{TV}}(D_u,D_v) \ge \frac{1}{4}$, and this half is exactly $S'$.

Putting these lemmas together, we have the following easy corollary:
\begin{corollary}\label{col:sda}
Let $\ell, n$ be integers with $n$ sufficiently large and $n^{\ell +1} \le 2^{n^{1/4}}$.
Let $\calS = \{D_u\}_{u \in S'}$ be as described in Problem~\ref{prob:gmm}. 
Then there exists a constant $c$ so that for all integers $n$ sufficiently large, for any $q\ge 1$,
\[
\SDA\left(\calS,\left(\frac{(n/c)^{(\ell+1)/5}}{\log \frac{1}{\eps}\left(1 + \frac{q^2}{2^{n^{1/4}}}\right)}\right)\right) \ge q.
\]
\end{corollary}
\begin{proof}
We have that $\Pr_{u,v \sim S} [u = v] = \frac{1}{|S'|}$.
Since Problem~\ref{prob:gmm} uses the construction from Lemma~\ref{lem:set} with $c = .26$, for $n$ sufficiently large $|S'| \ge 2^{n^{.255}}$ and $|\langle u, v \rangle| \le n^{-1/5}$ for all $u \neq v\in S'$.
Since Lemma~\ref{lem:construction} furnishes a bound on the correlation for $u \neq v$, for any event $\calE$,
\[
\E_{u,v \sim \mu}\left[\left|\langle \ol D_u, \ol D_v \rangle - 1\right| \mid \calE\right] \le \min\left(1,\frac{1}{|S'|\Pr[\calE]}\right)\cdot\|\ol A\|^2 + \max\left(0,1 - \frac{1}{|S'|\Pr[\calE]}\right)\cdot \frac{1}{n^{(\ell+1)/5}}\|\ol A\|^2,
\]
and substituting our bound on $|S'|$, using that $\|\ol A\|^2 \le \log \frac{1}{\eps} C^\ell$ for some constant $C$, and using the assumption that $n^{(\ell+1)/5}/2^{n^{0.255}}\le 2^{n^{1/4}}$, we have our conclusion.
\end{proof}

Applying Theorem~\ref{thm:converse}, we deduce the following bound:
\begin{corollary}\label{cor:gmm-ldlr}
There exists a real number $c \ge 0$ so that for any $\eps \in (0,1)$ and integer $\ell$, there exists $n$ sufficiently large that for any even integer $k\ll n^{1/8}$ and any $m \le \frac{(n/c)^{(\ell+1)/5}}{2\log \frac{1}{\eps}}$, the $(1-\eps)$-separated Gaussian $\ell$-mixture model testing problem $\calS = \{D_u\}_{u \in S}$ vs. $D_{\emptyset}$ described in Problem~\ref{prob:gmm} has $(\infty,k)$-$\LDLR_m$ bounded by
\[
\left\|\E_{u \sim S} (\ol D_{u}^{\otimes m})^{\le \infty,k} -1\right\|^2\le 1. 
\]
\end{corollary}
\begin{proof}
Let $m = \frac{(n/c)^{(\ell+1)/5}}{2\log \frac{1}{\eps}}$.
We notice that $|\langle \ol D_u, D_v \rangle -1 | \le \exp(O(\ell)) \log \frac{1}{\eps} \le m^{1/10}$ always, since $\eps,\ell$ are fixed constants.
Hence we meet the condition of Theorem~\ref{thm:converse} that $\|\E_{u}(\ol D_u -1)^{\otimes k}\|^2 \le m^{k/10}$.

Applying Corollary~\ref{col:sda} with $q = \sqrt{2^{n^{1/4}} \frac{m}{m'}}$, we have that for all $1 \le m' \le m$,
\[
\SDA(\calS,m') \ge \sqrt{2^{n^{1/4}} \frac{m}{m'}} \ge \left(\frac{100 m}{m'}\right)^k
\]
for any $k \le n^{.249}$.
This concludes the argument.
\end{proof}

\paragraph{Comparison with prior work and predictions}
The lower bound Corollary~\ref{cor:gmm-ldlr} is consistent with the SQ lower bounds of \cite{diakonikolas2017statistical}, suggesting efficient algorithms for learning a mixture of $\ell$ Gaussians in $n$ dimensions, each separated in total-variation distance, requires $d^{\Omega(\ell)}$ samples.
Information-theoretically, only $\poly(n,\ell)$ samples are required in this setting, although the information-theoretic sample complexity becomes exponential in $\ell$ if the Gaussians are not required to have total variation distance close to $1$ \cite{moitra2010settling}.
An algorithm using time and samples $d^{\poly(k)}$ is known \cite{moitra2010settling}.

\subsection{Gaussian Graphical Models}\label{sec:ex-ggm}
In this section, we prove an SDA lower bound for a hypothesis testing problem over Gaussian Graphical Models, and then show that this implies a LDLR lower bound for the same problem.
We will not succeed in establishing evidence for information computation gaps---the point of this example is to illustrate the utility of Theorem~\ref{thm:converse}, for a setting where LDLR lower bounds are highly intractable while SDA lower bounds are approachable.

In Gaussian Graphical models, we observe samples $x_1,\ldots,x_m \sim \calN(\mu, \Theta^{-1})$, where $\Theta$ is a sparse positive semidefinite matrix---since it is sparse, it is thought of as a graph.
The goal is to get algorithms for estimating $\Theta$ which do not depend on its condition number, and which take advantage of the graph sparsity.
The relevant parameters are the maximum degree $d$ and the non-degeneracy parameter $\kappa:= \min_{i,j \in [n]} \frac{|\Theta_{ij}|}{\sqrt{\Theta_{ii}\Theta_{jj}}}$.

\begin{problem}[Gaussian Graphical Models: planted $d$-regular subgraph]
For $n>s>d$ positive integers and $\kappa \in \R$ with $\kappa \sqrt{d} < \frac{1}{6}$, the {\em $\kappa$-nondegenerate $d$-sparse $s$-planted $n$-dimensional planted regular subgraph Gaussian Graphical Model ($(\kappa,d,s,n)$-prsGGM)} problem is the following many-vs-one hypothesis testing problem:
\begin{itemize}
\item Null: $D_{\emptyset} = \calN(0,\Id_{n})$.
\item Alternate: uniform mixture of $D_u = \calN(0, (\Id_n + \kappa \Delta_u)^{-1} )$, over $u \sim S$, where each $u$ is sampled by choosing $s$ of $n$ indices uniformly at random, and then planting a randomly signed random $d$-regular graph on those indices (conditioned on the graph having all eigenvalues bounded in magnitude by $2\sqrt{d}$), then taking $\Delta_u$ to be the adjacency matrix of that graph.
\end{itemize}
\end{problem}

We will prove the following Lemma, from which we obtain an LDLR lower bound as a corollary of Theorem~\ref{thm:converse}:

\begin{lemma}
\torestate{
\label{lem:ggm}
For any integer $d$ sufficiently large, any $s\gg d$ sufficiently large, any $n\gg s$ sufficiently large, and $\kappa \in (0,\frac{1}{6\sqrt{d}})$ such that the following holds:
If $\calS$ vs. $D_{\emptyset}$ is an instance of the $(\kappa,d,s,n)$-prsGGM problem, then for any even integer $k$ and $q \ge 1$,
\[
\SDA\left(\calS, \left(\frac{n}{q^2 s^2}\right)^{1/k}\frac{1}{\exp(\frac{1}{2}sd\kappa^2)-1}\right) \ge q,
\]
and further,
\[
\E_{u,v} \langle \ol D_u, \ol D_v \rangle^k 
\le 
\left(1 + \left(\frac{s^2}{n}\right)^{1/k}\left(\exp(\tfrac{1}{2}sd\kappa^2)-1\right)\right)^{k}.
\]
}
\end{lemma}
We give the proof of this Lemma in Appendix~\ref{sec:GGM-app}.
Combining Lemma~\ref{lem:ggm} with Theorem~\ref{thm:converse} gives us the following corollary:
\begin{corollary}\label{cor:ggm}
For any integer $d$ sufficiently large, any $s\gg d$ sufficiently large, any $n\gg s$ sufficiently large, and $\kappa \in (0,\frac{1}{6\sqrt{d}})$ such that the following holds:
If $\calS$ vs. $D_{\emptyset}$ is an instance of the $(\kappa,d,s,n)$-prsGGM problem, then for any even integers $k,t$ and $m \le \frac{1}{2}\left(\frac{n}{s^2}\right)^{1/k}\frac{1}{\exp(\frac{1}{2}sd\kappa^2)-1}$ with $sd\kappa^2 \le \frac{k}{10}\log m$, the $m$-sample $(t,\Omega(k))$-$\LDLR_m$ is bounded: 
\[
\left\|\E_{u \sim S} (\ol D_u^{\otimes m})^{\le t,k/2} - 1 \right\| \le 1.
\]
\end{corollary}

\paragraph{Comparison with prior work and predictions.}
For an arbitrary Gaussian Graphical Model with maximum degree $d$, $\kappa$-nondegeneracy, and dimension $n$, information-theoretically, $m \ge \frac{\log n}{\kappa^2}$ samples are required \cite{WWR10}, and the fastest known algorithms for $m = \Theta(\frac{\kappa^2}{\log n})$ run in time $n^{O(d)}$ \cite{KKMM19}, though faster algorithms are known for more structured cases \cite{KKMM19,RWRY11}.
Given the current state of the literature, it is not clear whether it is possible to achieve the information-theoretic limit with $n^{o(d)}$ time algorithms.

Our bounds are not strong enough to give evidence for an information-computation gap: for signal-to-noise ratios corresponding to $m =\Theta(\frac{\log n}{\kappa^2})$ samples, by choosing $s = \log n$ and $\kappa$ small enough we can rule out SQ algorithms with fewer than $\sqrt{n/(d\log^4 n)}$ queries, or degree-$O(\frac{\log n}{\log d})$ polynomial distinguishers (these bounds degrade as $d$ increases, instead of the other way around).
We do not expect that this bound is tight, and our bound from Lemma~\ref{lem:ggm} might easily be improved with a more careful analysis.
But, because the matrices that we use are well-conditioned, and because there are algorithms for well-conditioned matrices that require fewer samples, it is unlikely that the hypothesis testing problem we consider will give evidence for this information-computation tradeoff, even if analyzed optimally.

However, this example does illustrate that it is possible to obtain a bound depending on the sparsity and non-degeneracy; in this, it highlights the usefulness of Theorem~\ref{thm:converse}. 
In the GGM problem, any set of alternate hypotheses $\calS$ by definition involves Gaussian distributions whose {\em inverse} covariance matrices are easy to describe, but the covariance matrices themselves are not; this would make calculating the LDLR directly extremely arduous, even for our toy example of alternate distributions.
However, calculating some bound on the SDA is relatively tractable, and Theorem~\ref{thm:converse} lets us draw conclusions for the LDLR.

\subsection{Sparse Parity with Noise}\label{sec:sparse-par}

Theorem~\ref{thm:noisy} shows that if for the hypothesis testing problem $T_\rho \calS$ vs $D_{\emptyset}$, the $(s-1,k)$-$\LDLR_m$ is bounded by $\eps$, and $\|\E_{u} (\ol D_u)^{\otimes k}\|^2 \le O(1)$, and $\rho^{2s} = O(\frac{1}{m})$, then at least $2^k$ queries to $\VSTAT(O(m/k))$ are necessary.
The following example illustrates that this dependence on $\rho$ is tight.

\begin{problem}
The following is the {\em $2^k$-subset of $s$-sparse parities} problem:
\begin{itemize}
\item Null: $D_{\emptyset}$ is uniform over $\{\pm 1\}^n$.
\item Alternate: For $S$ an arbitrary subset of  $\binom{[n]}{s}$ with $|S| = 2^k$, define $\calS = \{D_u\}_{u \in D}$, where for each $u \in S$ we take $D_{u}$ uniform over $x\sim \{\pm 1\}^n$ conditioned on $x^u = 1$.
\end{itemize}
\end{problem}

\begin{claim}
For any $\rho \in [-1,1]$ and $T_\rho$ the standard Boolean noise operator, and any integer $m$, the many-vs-one $2^k$-subset of $s$-sparse parities problem $D_{\emptyset}$ vs $\calS = \{D_u\}$ has
\[
\|\E_{u \sim S} (\ol T_\rho D_u^{\otimes m})^{\le s-1,\infty} -1 \| = 0.
\]
\end{claim}
\begin{proof}
This is because each $\ol{D_u}$ has no Fourier mass on degrees $1$ through $s-1$.
\end{proof}

\begin{claim}
For the many-vs-one $2^k$-subset of $s$-sparse parities problem,
\[
\|\E_{u \sim S} (\ol D_u^{\otimes k}) \|^2 \le 2.
\]
\end{claim}
\begin{proof}
For each $u \neq v$, $\langle\ol D_u,\ol D_v \rangle = 1$, and $\langle \ol D_u, \ol D_u \rangle = 2$.
We then use the fact that $|S| \le 2^k$ to calculate,
\[
\|\E_u (\ol D_u)^{\otimes k}\|^2 
= \E_{u,v \sim S} \langle \ol D_u, \ol D_v \rangle^k
= \frac{1}{|S|} \cdot 2^k + (1-\frac{1}{|S|}) \cdot 1 \le 2.
\]
\end{proof}

Together, the above claims demonstrate that we meet the conditions of Theorem~\ref{thm:noisy}.
However, there is also a $2^k$-query $\VSTAT(\rho^{-2s})$ algorithm:
\begin{claim}
There is a $2^k$ query $\VSTAT(\rho^{-2s})$ algorithm for the $\rho$-noisy $2^k$-subset of $s$-sparse parities problem, $T_\rho \cal S$ vs. $D_{\emptyset}$.
\end{claim}
\begin{proof}
The algorithm is as follows: for each $u \in S$, take the query $\phi_u(x) = \frac{1}{2}(1 + x^u)$.
Under null, $\E_{D_{\emptyset}} \phi_u = \frac{1}{2}$. 
Under $T_\rho D_u$, $\E_{T_\rho D_u} \phi_u = \frac{1}{2}(1+\rho^{s})$.
Thus, a $\VSTAT(\rho^{-2s})$ algorithm can distinguish these cases.
\end{proof}
Hence, the requirement in Theorem~\ref{thm:noisy} that $\rho^{2s} =O(\frac{1}{m})$ is tight.

\section*{Acknowledgments}
T.S. thanks Ankur Moitra, Alex Wein, Fred Koehler, and Adam Klivans for helpful conversations regarding the nature of statistical query algorithms and the implications of this work.

\bibliographystyle{amsalpha}
\bibliography{ref,statcomplexityBIB}

\newcommand{\etalchar}[1]{$^{#1}$}
\providecommand{\bysame}{\leavevmode\hbox to3em{\hrulefill}\thinspace}
\providecommand{\MR}{\relax\ifhmode\unskip\space\fi MR }
\providecommand{\MRhref}[2]{%
  \href{http://www.ams.org/mathscinet-getitem?mr=#1}{#2}
}
\providecommand{\href}[2]{#2}
\begin{thebibliography}{ABDR{\etalchar{+}}18}

\bibitem[ABDR{\etalchar{+}}18]{atserias2018clique}
Albert Atserias, Ilario Bonacina, Susanna De~Rezende, Massimo Lauria, Jakob
  Nordstr\H{o}m, and Alexander Razborov, \emph{Clique is hard on average for
  regular resolution}, Symposium on the Theory of Computing (STOC), 2018.

\bibitem[ACBL12]{arias2012detection}
Ery Arias-Castro, S{\'e}bastien Bubeck, and G{\'a}bor Lugosi, \emph{Detection
  of correlations}, The Annals of Statistics \textbf{40} (2012), no.~1,
  412--435.

\bibitem[ACO08]{achlioptas2008algorithmic}
Dimitris Achlioptas and Amin Coja-Oghlan, \emph{Algorithmic barriers from phase
  transitions}, 2008 49th Annual IEEE Symposium on Foundations of Computer
  Science, IEEE, 2008, pp.~793--802.

\bibitem[ACV14]{arias2014community}
Ery Arias-Castro and Nicolas Verzelen, \emph{Community detection in dense
  random networks}, The Annals of Statistics \textbf{42} (2014), no.~3,
  940--969.

\bibitem[AGJ{\etalchar{+}}20]{arous2020algorithmic}
Gerard~Ben Arous, Reza Gheissari, Aukosh Jagannath, et~al., \emph{Algorithmic
  thresholds for tensor pca}, Annals of Probability \textbf{48} (2020), no.~4,
  2052--2087.

\bibitem[AWZ20]{arous2020free}
G{\'e}rard~Ben Arous, Alexander~S Wein, and Ilias Zadik, \emph{Free energy
  wells and overlap gap property in sparse pca}, Conference on Learning Theory,
  2020, pp.~479--482.

\bibitem[BAP{\etalchar{+}}05]{baik2005phase}
Jinho Baik, G{\'e}rard~Ben Arous, Sandrine P{\'e}ch{\'e}, et~al., \emph{Phase
  transition of the largest eigenvalue for nonnull complex sample covariance
  matrices}, The Annals of Probability \textbf{33} (2005), no.~5, 1643--1697.

\bibitem[BB19]{brennan2019optimal}
Matthew Brennan and Guy Bresler, \emph{Optimal average-case reductions to
  sparse pca: From weak assumptions to strong hardness}, Conference on Learning
  Theory, 2019, pp.~469--470.

\bibitem[BB20]{brennan2020reducibility}
\bysame, \emph{Reducibility and statistical-computational gaps from secret
  leakage}, Conference on Learning Theory (COLT), 2020.

\bibitem[BBH18]{brennan18}
Matthew Brennan, Guy Bresler, and Wasim Huleihel, \emph{Reducibility and
  computational lower bounds for problems with planted sparse structure},
  Conference on Learning Theory (COLT), 2018.

\bibitem[BBH19]{brennan2019universality}
\bysame, \emph{Universality of computational lower bounds for submatrix
  detection}, Conference on Learning Theory (COLT), 2019.

\bibitem[BBKW19]{BBKMW20}
Afonso~S Bandeira, Jess Banks, Dmitriy Kunisky, and Alexander~S Wein,
  \emph{Spectral planting and the hardness of refuting cuts, colorability,and
  communities in random graphs}, arXiv preprint arXiv:2008.12237 (2019).

\bibitem[Bei93]{beigel1993polynomial}
Richard Beigel, \emph{The polynomial method in circuit complexity}, [1993]
  Proceedings of the Eigth Annual Structure in Complexity Theory Conference,
  IEEE, 1993, pp.~82--95.

\bibitem[BFJ{\etalchar{+}}94]{blum1994weakly}
Avrim Blum, Merrick Furst, Jeffrey Jackson, Michael Kearns, Yishay Mansour, and
  Steven Rudich, \emph{Weakly learning dnf and characterizing statistical query
  learning using fourier analysis}, Proceedings of the twenty-sixth annual ACM
  symposium on Theory of computing, 1994, pp.~253--262.

\bibitem[BGL17]{BGL17}
Vijay Bhattiprolu, Venkatesan Guruswami, and Euiwoong Lee, \emph{Sum-of-squares
  certificates for maxima of random tensors on the sphere}, {APPROX/RANDOM}
  2017 (Klaus Jansen, Jos{\'{e}} D.~P. Rolim, David Williamson, and Santosh~S.
  Vempala, eds.), LIPIcs, vol.~81, Schloss Dagstuhl - Leibniz-Zentrum f{\"{u}}r
  Informatik, 2017, pp.~31:1--31:20.

\bibitem[BGS14]{BGS14b}
G.~Bresler, D.~Gamarnik, and D.~Shah, \emph{Hardness of parameter estimation in
  graphical models}, Neural Information Processing Systems, 2014.

\bibitem[BHK{\etalchar{+}}19]{barak2019nearly}
Boaz Barak, Samuel Hopkins, Jonathan Kelner, Pravesh~K Kothari, Ankur Moitra,
  and Aaron Potechin, \emph{A nearly tight sum-of-squares lower bound for the
  planted clique problem}, SIAM Journal on Computing \textbf{48} (2019), no.~2,
  687--735.

\bibitem[BKR{\etalchar{+}}11]{balakrishnan2011statistical}
Sivaraman Balakrishnan, Mladen Kolar, Alessandro Rinaldo, Aarti Singh, and
  Larry Wasserman, \emph{Statistical and computational tradeoffs in
  biclustering}, NeurIPS 2011 workshop on computational trade-offs in
  statistical learning, vol.~4, 2011.

\bibitem[BKW19]{bandeira2019computational}
Afonso~S Bandeira, Dmitriy Kunisky, and Alexander~S Wein, \emph{Computational
  hardness of certifying bounds on constrained pca problems}, arXiv preprint
  arXiv:1902.07324 (2019).

\bibitem[BR13a]{berthet2013complexity}
Quentin Berthet and Philippe Rigollet, \emph{Complexity theoretic lower bounds
  for sparse principal component detection}, Conference on Learning Theory,
  2013, pp.~1046--1066.

\bibitem[BR13b]{berthet2013optimal}
\bysame, \emph{Optimal detection of sparse principal components in high
  dimension}, The Annals of Statistics \textbf{41} (2013), no.~4, 1780--1815.

\bibitem[CJ13]{chandrasekaran2013computational}
Venkat Chandrasekaran and Michael~I Jordan, \emph{Computational and statistical
  tradeoffs via convex relaxation}, Proceedings of the National Academy of
  Sciences \textbf{110} (2013), no.~13, E1181--E1190.

\bibitem[CMP10]{chai2010array}
Anwei Chai, Miguel Moscoso, and George Papanicolaou, \emph{Array imaging using
  intensity-only measurements}, Inverse Problems \textbf{27} (2010), no.~1,
  015005.

\bibitem[CRT06]{candes2006stable}
Emmanuel~J Candes, Justin~K Romberg, and Terence Tao, \emph{Stable signal
  recovery from incomplete and inaccurate measurements}, Communications on Pure
  and Applied Mathematics: A Journal Issued by the Courant Institute of
  Mathematical Sciences \textbf{59} (2006), no.~8, 1207--1223.

\bibitem[CSV13]{candes2013phaselift}
Emmanuel~J Candes, Thomas Strohmer, and Vladislav Voroninski, \emph{Phaselift:
  Exact and stable signal recovery from magnitude measurements via convex
  programming}, Communications on Pure and Applied Mathematics \textbf{66}
  (2013), no.~8, 1241--1274.

\bibitem[CT07]{candes2007dantzig}
Emmanuel Candes and Terence Tao, \emph{The {D}antzig selector: Statistical
  estimation when p is much larger than n}, The Annals of Statistics
  \textbf{35} (2007), no.~6, 2313--2351.

\bibitem[CX16]{chen2016statistical}
Yudong Chen and Jiaming Xu, \emph{Statistical-computational tradeoffs in
  planted problems and submatrix localization with a growing number of clusters
  and submatrices}, Journal of Machine Learning Research \textbf{17} (2016),
  no.~27, 1--57.

\bibitem[dBG08]{d2008optimal}
Alexandre d’Aspremont, Francis Bach, and Laurent~El Ghaoui, \emph{Optimal
  solutions for sparse principal component analysis}, Journal of Machine
  Learning Research \textbf{9} (2008), no.~Jul, 1269--1294.

\bibitem[DGR00]{decatur2000computational}
Scott~E Decatur, Oded Goldreich, and Dana Ron, \emph{Computational sample
  complexity}, SIAM Journal on Computing \textbf{29} (2000), no.~3, 854--879.

\bibitem[DH20]{DH20}
Rishabh Dudeja and Daniel Hsu, \emph{Statistical query lower bounds for tensor
  {PCA}}, arXiv preprint arXiv:2008.04101 (2020).

\bibitem[DKS17]{diakonikolas2017statistical}
Ilias Diakonikolas, Daniel~M Kane, and Alistair Stewart, \emph{Statistical
  query lower bounds for robust estimation of high-dimensional gaussians and
  gaussian mixtures}, 2017 IEEE 58th Annual Symposium on Foundations of
  Computer Science (FOCS), IEEE, 2017, pp.~73--84.

\bibitem[DKS19]{diakonikolas2019efficient}
Ilias Diakonikolas, Weihao Kong, and Alistair Stewart, \emph{Efficient
  algorithms and lower bounds for robust linear regression}, Proceedings of the
  Thirtieth Annual ACM-SIAM Symposium on Discrete Algorithms, SIAM, 2019,
  pp.~2745--2754.

\bibitem[DKWB19]{ding2019subexponential}
Yunzi Ding, Dmitriy Kunisky, Alexander~S Wein, and Afonso~S Bandeira,
  \emph{Subexponential-time algorithms for sparse {PCA}}, arXiv preprint
  arXiv:1907.11635 (2019).

\bibitem[DM15]{deshpande2015improved}
Yash Deshpande and Andrea Montanari, \emph{Improved sum-of-squares lower bounds
  for hidden clique and hidden submatrix problems.}, Conference on Learning
  Theory (COLT), 2015, pp.~523--562.

\bibitem[Don06]{donoho2006compressed}
David~L Donoho, \emph{Compressed sensing}, IEEE Transactions on information
  theory \textbf{52} (2006), no.~4, 1289--1306.

\bibitem[FB96]{feng1996spectrum}
Ping Feng and Yoram Bresler, \emph{Spectrum-blind minimum-rate sampling and
  reconstruction of multiband signals}, Acoustics, Speech, and Signal
  Processing, 1996. ICASSP-96. Conference Proceedings., 1996 IEEE International
  Conference on, vol.~3, IEEE, 1996, pp.~1688--1691.

\bibitem[Fei02]{feige2002relations}
Uriel Feige, \emph{Relations between average case complexity and approximation
  complexity}, Proceedings of the thiry-fourth annual ACM symposium on Theory
  of computing, ACM, 2002, pp.~534--543.

\bibitem[Fel12]{feldman2012complete}
Vitaly Feldman, \emph{A complete characterization of statistical query learning
  with applications to evolvability}, Journal of Computer and System Sciences
  \textbf{78} (2012), no.~5, 1444--1459.

\bibitem[FGR{\etalchar{+}}17]{FGRVX}
Vitaly Feldman, Elena Grigorescu, Lev Reyzin, Santosh~S Vempala, and Ying Xiao,
  \emph{Statistical algorithms and a lower bound for detecting planted
  cliques}, Journal of the ACM (JACM) \textbf{64} (2017), no.~2, 1--37.

\bibitem[FGV17]{feldman2017statistical}
Vitaly Feldman, Cristobal Guzman, and Santosh Vempala, \emph{Statistical query
  algorithms for mean vector estimation and stochastic convex optimization},
  Proceedings of the Twenty-Eighth Annual ACM-SIAM Symposium on Discrete
  Algorithms, SIAM, 2017, pp.~1265--1277.

\bibitem[FHT08]{friedman2008sparse}
J.~Friedman, T.~Hastie, and R.~Tibshirani, \emph{Sparse inverse covariance
  estimation with the graphical lasso}, Biostatistics \textbf{9} (2008), no.~3,
  432--441.

\bibitem[FK03]{feige2003probable}
Uriel Feige and Robert Krauthgamer, \emph{The probable value of the
  lov{\'a}sz--schrijver relaxations for maximum independent set}, SIAM Journal
  on Computing \textbf{32} (2003), no.~2, 345--370.

\bibitem[FPV18]{feldman2018complexity}
Vitaly Feldman, Will Perkins, and Santosh Vempala, \emph{On the complexity of
  random satisfiability problems with planted solutions}, SIAM Journal on
  Computing \textbf{47} (2018), no.~4, 1294--1338.

\bibitem[GGJ{\etalchar{+}}20]{goel2020superpolynomial}
Surbhi Goel, Aravind Gollakota, Zhihan Jin, Sushrut Karmalkar, and Adam
  Klivans, \emph{Superpolynomial lower bounds for learning one-layer neural
  networks using gradient descent}, arXiv preprint arXiv:2006.12011 (2020).

\bibitem[GJS19]{gamarnik2019overlap}
David Gamarnik, Aukosh Jagannath, and Subhabrata Sen, \emph{The overlap gap
  property in principal submatrix recovery}, arXiv preprint arXiv:1908.09959
  (2019).

\bibitem[GJW20]{gamarnik2020low}
David Gamarnik, Aukosh Jagannath, and Alexander~S Wein, \emph{Low-degree
  hardness of random optimization problems}, arXiv preprint arXiv:2004.12063
  (2020).

\bibitem[Gri01]{grigoriev2001linear}
Dima Grigoriev, \emph{Linear lower bound on degrees of positivstellensatz
  calculus proofs for the parity}, Theoretical Computer Science \textbf{259}
  (2001), no.~1-2, 613--622.

\bibitem[GS14]{gamarnik2014limits}
David Gamarnik and Madhu Sudan, \emph{Limits of local algorithms over sparse
  random graphs}, Proceedings of the 5th conference on Innovations in
  theoretical computer science, 2014, pp.~369--376.

\bibitem[GZ19]{gamarnik2019landscape}
David Gamarnik and Ilias Zadik, \emph{The landscape of the planted clique
  problem: Dense subgraphs and the overlap gap property}, arXiv preprint
  arXiv:1904.07174 (2019).

\bibitem[HKP{\etalchar{+}}17]{hopkins2017power}
Samuel~B Hopkins, Pravesh~K Kothari, Aaron Potechin, Prasad Raghavendra, Tselil
  Schramm, and David Steurer, \emph{The power of sum-of-squares for detecting
  hidden structures}, 2017 IEEE 58th Annual Symposium on Foundations of
  Computer Science (FOCS), IEEE, 2017, pp.~720--731.

\bibitem[HKP{\etalchar{+}}18]{hopkins2018integrality}
Samuel~B Hopkins, Pravesh Kothari, Aaron~Henry Potechin, Prasad Raghavendra,
  and Tselil Schramm, \emph{On the integrality gap of degree-4 sum of squares
  for planted clique}, ACM Transactions on Algorithms (TALG) \textbf{14}
  (2018), no.~3, 28.

\bibitem[HL19]{HL19}
Samuel~B Hopkins and Jerry Li, \emph{How hard is robust mean estimation?},
  arXiv preprint arXiv:1903.07870 (2019).

\bibitem[Hop18]{hopkinsThesis}
Samuel~B Hopkins, \emph{Statistical inference and the sum of squares method},
  Ph.D. thesis, Cornell University, 2018.

\bibitem[HS17]{hopkins2017efficient}
Samuel~B Hopkins and David Steurer, \emph{Efficient bayesian estimation from
  few samples: community detection and related problems}, 2017 IEEE 58th Annual
  Symposium on Foundations of Computer Science (FOCS), IEEE, 2017,
  pp.~379--390.

\bibitem[HSS15]{hopkins2015tensor}
Samuel~B Hopkins, Jonathan Shi, and David Steurer, \emph{Tensor principal
  component analysis via sum-of-square proofs}, Conference on Learning Theory,
  2015, pp.~956--1006.

\bibitem[HW20]{holmgren2020counterexamples}
Justin Holmgren and Alexander~S Wein, \emph{Counterexamples to the low-degree
  conjecture}, arXiv preprint arXiv:2004.08454 (2020).

\bibitem[HWX15]{hajek2015computational}
Bruce~E Hajek, Yihong Wu, and Jiaming Xu, \emph{Computational lower bounds for
  community detection on random graphs.}, Conference on Learning Theory (COLT),
  2015, pp.~899--928.

\bibitem[IKKM12]{ibrahimi2012set}
Morteza Ibrahimi, Yashodhan Kanoria, Matt Kraning, and Andrea Montanari,
  \emph{The set of solutions of random xorsat formulae}, Proceedings of the
  twenty-third annual ACM-SIAM symposium on Discrete Algorithms, SIAM, 2012,
  pp.~760--779.

\bibitem[Jer92]{jerrum1992large}
Mark Jerrum, \emph{Large cliques elude the metropolis process}, Random
  Structures \& Algorithms \textbf{3} (1992), no.~4, 347--359.

\bibitem[JL09]{johnstone2009consistency}
Iain~M Johnstone and Arthur~Yu Lu, \emph{On consistency and sparsity for
  principal components analysis in high dimensions}, Journal of the American
  Statistical Association \textbf{104} (2009), no.~486, 682--693.

\bibitem[JMS04]{jia2004spin}
Haixia Jia, Cris Moore, and Bart Selman, \emph{From spin glasses to hard
  satisfiable formulas}, International Conference on Theory and Applications of
  Satisfiability Testing, Springer, 2004, pp.~199--210.

\bibitem[JNS13]{jain2013low}
Prateek Jain, Praneeth Netrapalli, and Sujay Sanghavi, \emph{Low-rank matrix
  completion using alternating minimization}, Proceedings of the forty-fifth
  annual ACM symposium on Theory of computing, ACM, 2013, pp.~665--674.

\bibitem[JOH]{jaganathan2013sparse}
Kishore Jaganathan, Samet Oymak, and Babak Hassibi, \emph{Sparse phase
  retrieval: Convex algorithms and limitations}, 2013 IEEE International
  Symposium on Information Theory.

\bibitem[JT18]{ji2018risk}
Ziwei Ji and Matus Telgarsky, \emph{Risk and parameter convergence of logistic
  regression}, arXiv preprint arXiv:1803.07300 (2018).

\bibitem[Kea98]{kearns}
Michael Kearns, \emph{Efficient noise-tolerant learning from statistical
  queries}, Journal of the ACM (JACM) \textbf{45} (1998), no.~6, 983--1006.

\bibitem[KKMM19]{KKMM19}
Jonathan Kelner, Frederic Koehler, Raghu Meka, and Ankur Moitra, \emph{Learning
  some popular gaussian graphical models without condition number bounds},
  arXiv preprint arXiv:1905.01282 (2019).

\bibitem[KMH{\etalchar{+}}20]{kaplan2020scaling}
Jared Kaplan, Sam McCandlish, Tom Henighan, Tom~B Brown, Benjamin Chess, Rewon
  Child, Scott Gray, Alec Radford, Jeffrey Wu, and Dario Amodei, \emph{Scaling
  laws for neural language models}, arXiv preprint arXiv:2001.08361 (2020).

\bibitem[KMOW17]{kothari2017sum}
Pravesh~K Kothari, Ryuhei Mori, Ryan O'Donnell, and David Witmer, \emph{Sum of
  squares lower bounds for refuting any csp}, Proceedings of the 49th Annual
  ACM SIGACT Symposium on Theory of Computing, 2017, pp.~132--145.

\bibitem[KS07]{klivans2007unconditional}
Adam~R Klivans and Alexander~A Sherstov, \emph{Unconditional lower bounds for
  learning intersections of halfspaces}, Machine Learning \textbf{69} (2007),
  no.~2-3, 97--114.

\bibitem[KWB19]{BKW19}
Dmitriy Kunisky, Alexander~S Wein, and Afonso~S Bandeira, \emph{Notes on
  computational hardness of hypothesis testing: Predictions using the
  low-degree likelihood ratio}, arXiv preprint arXiv:1907.11636 (2019).

\bibitem[LDP07]{lustig2007sparse}
Michael Lustig, David Donoho, and John~M Pauly, \emph{Sparse {MRI}: The
  application of compressed sensing for rapid {MR} imaging}, Magnetic Resonance
  in Medicine: An Official Journal of the International Society for Magnetic
  Resonance in Medicine \textbf{58} (2007), no.~6, 1182--1195.

\bibitem[LML{\etalchar{+}}17]{lesieur2017statistical}
Thibault Lesieur, L{\'e}o Miolane, Marc Lelarge, Florent Krzakala, and Lenka
  Zdeborov{\'a}, \emph{Statistical and computational phase transitions in
  spiked tensor estimation}, 2017 IEEE International Symposium on Information
  Theory (ISIT), IEEE, 2017, pp.~511--515.

\bibitem[LZ20]{luo2020tensor}
Yuetian Luo and Anru~R Zhang, \emph{Tensor clustering with planted structures:
  Statistical optimality and computational limits}, arXiv preprint
  arXiv:2005.10743 (2020).

\bibitem[MM09]{mezard2009information}
Marc Mezard and Andrea Montanari, \emph{Information, physics, and computation},
  Oxford University Press, 2009.

\bibitem[{Mon}14]{Montanari2014}
A.~{Montanari}, \emph{{Computational Implications of Reducing Data to
  Sufficient Statistics}}, ArXiv e-prints (2014).

\bibitem[Mon15]{montanari2015finding}
Andrea Montanari, \emph{Finding one community in a sparse graph}, Journal of
  Statistical Physics \textbf{161} (2015), no.~2, 273--299.

\bibitem[MPW15]{meka2015sum}
Raghu Meka, Aaron Potechin, and Avi Wigderson, \emph{Sum-of-squares lower
  bounds for planted clique}, Proceedings of the forty-seventh annual ACM
  symposium on Theory of computing, ACM, 2015, pp.~87--96.

\bibitem[MV10]{moitra2010settling}
Ankur Moitra and Gregory Valiant, \emph{Settling the polynomial learnability of
  mixtures of gaussians}, 2010 IEEE 51st Annual Symposium on Foundations of
  Computer Science, IEEE, 2010, pp.~93--102.

\bibitem[MW15]{ma2015computational}
Zongming Ma and Yihong Wu, \emph{Computational barriers in minimax submatrix
  detection}, The Annals of Statistics \textbf{43} (2015), no.~3, 1089--1116.

\bibitem[NKB{\etalchar{+}}19]{nakkiran2019deep}
Preetum Nakkiran, Gal Kaplun, Yamini Bansal, Tristan Yang, Boaz Barak, and Ilya
  Sutskever, \emph{Deep double descent: Where bigger models and more data
  hurt}, arXiv preprint arXiv:1912.02292 (2019).

\bibitem[PWB{\etalchar{+}}18]{perry2018optimality}
Amelia Perry, Alexander~S Wein, Afonso~S Bandeira, Ankur Moitra, et~al.,
  \emph{Optimality and sub-optimality of pca i: Spiked random matrix models},
  The Annals of Statistics \textbf{46} (2018), no.~5, 2416--2451.

\bibitem[RBE10]{rubinstein2010dictionaries}
Ron Rubinstein, Alfred~M Bruckstein, and Michael Elad, \emph{Dictionaries for
  sparse representation modeling}, Proceedings of the IEEE \textbf{98} (2010),
  no.~6, 1045--1057.

\bibitem[RCLV13]{ranieri2013phase}
Juri Ranieri, Amina Chebira, Yue~M Lu, and Martin Vetterli, \emph{Phase
  retrieval for sparse signals: Uniqueness conditions}, arXiv preprint
  arXiv:1308.3058 (2013).

\bibitem[RFP10]{recht2010guaranteed}
Benjamin Recht, Maryam Fazel, and Pablo~A Parrilo, \emph{Guaranteed
  minimum-rank solutions of linear matrix equations via nuclear norm
  minimization}, SIAM review \textbf{52} (2010), no.~3, 471--501.

\bibitem[RM14]{richard2014statistical}
Emile Richard and Andrea Montanari, \emph{A statistical model for tensor pca},
  Advances in Neural Information Processing Systems, 2014, pp.~2897--2905.

\bibitem[Ros08]{rossman2008constant}
Benjamin Rossman, \emph{On the constant-depth complexity of k-clique},
  Proceedings of the fortieth annual ACM symposium on Theory of computing, ACM,
  2008, pp.~721--730.

\bibitem[Ros14]{rossman2014monotone}
\bysame, \emph{The monotone complexity of k-clique on random graphs}, SIAM
  Journal on Computing \textbf{43} (2014), no.~1, 256--279.

\bibitem[RRS17]{RRS17}
Prasad Raghavendra, Satish Rao, and Tselil Schramm, \emph{Strongly refuting
  random {CSP}s below the spectral threshold}, Proceedings of the 49th Annual
  {ACM} {SIGACT} Symposium on Theory of Computing, 2017, pp.~121--131.

\bibitem[RSS18]{raghavendra2018high}
Prasad Raghavendra, Tselil Schramm, and David Steurer, \emph{High-dimensional
  estimation via sum-of-squares proofs}, arXiv preprint arXiv:1807.11419
  \textbf{6} (2018).

\bibitem[RWR{\etalchar{+}}11]{RWRY11}
Pradeep Ravikumar, Martin~J Wainwright, Garvesh Raskutti, Bin Yu, et~al.,
  \emph{High-dimensional covariance estimation by minimizing $\ell_1$-penalized
  log-determinant divergence}, Electronic Journal of Statistics \textbf{5}
  (2011), 935--980.

\bibitem[Ser99]{servedio1999computational}
Rocco~A Servedio, \emph{Computational sample complexity and attribute-efficient
  learning}, Proceedings of the thirty-first annual ACM symposium on Theory of
  computing, ACM, 1999, pp.~701--710.

\bibitem[SHN{\etalchar{+}}18]{soudry2018implicit}
Daniel Soudry, Elad Hoffer, Mor~Shpigel Nacson, Suriya Gunasekar, and Nathan
  Srebro, \emph{The implicit bias of gradient descent on separable data}, The
  Journal of Machine Learning Research \textbf{19} (2018), no.~1, 2822--2878.

\bibitem[SSS08]{shalev2008svm}
Shai Shalev-Shwartz and Nathan Srebro, \emph{{SVM} optimization: inverse
  dependence on training set size}, Proceedings of the 25th international
  conference on Machine learning, ACM, 2008, pp.~928--935.

\bibitem[SSST12]{shalev2012using}
Shai Shalev-Shwartz, Ohad Shamir, and Eran Tromer, \emph{Using more data to
  speed-up training time}, Artificial Intelligence and Statistics (AISTATS),
  2012, pp.~1019--1027.

\bibitem[SW20]{schramm2020computational}
Tselil Schramm and Alexander~S Wein, \emph{Computational barriers to estimation
  from low-degree polynomials}, arXiv preprint arXiv:2008.02269 (2020).

\bibitem[SWW12]{spielman2012exact}
Daniel~A Spielman, Huan Wang, and John Wright, \emph{Exact recovery of
  sparsely-used dictionaries}, Conference on Learning Theory (COLT), 2012,
  pp.~37--1.

\bibitem[WEAM19]{wein2019kikuchi}
Alexander~S Wein, Ahmed El~Alaoui, and Cristopher Moore, \emph{The {K}ikuchi
  hierarchy and tensor {PCA}}, 2019 {IEEE} 60th Annual Symposium on Foundations
  of Computer Science ({FOCS}), {IEEE}, 2019, pp.~1446--1468.

\bibitem[WGL15]{wang2015sharp}
Zhaoran Wang, Quanquan Gu, and Han Liu, \emph{Sharp computational-statistical
  phase transitions via oracle computational model}, arXiv preprint
  arXiv:1512.08861 (2015).

\bibitem[WWR10]{WWR10}
Wei Wang, Martin~J Wainwright, and Kannan Ramchandran,
  \emph{Information-theoretic bounds on model selection for gaussian markov
  random fields}, 2010 IEEE International Symposium on Information Theory,
  IEEE, 2010, pp.~1373--1377.

\bibitem[ZK16]{zdeborova_statistical_2016}
Lenka Zdeborová and Florent Krzakala, \emph{Statistical physics of inference:
  thresholds and algorithms}, Advances in Physics \textbf{65} (2016), no.~5,
  453--552.

\bibitem[ZX18]{zhang2018tensor}
Anru Zhang and Dong Xia, \emph{Tensor {SVD}: Statistical and computational
  limits}, IEEE Transactions on Information Theory (2018).

\end{thebibliography}
\appendix
\section{SDA, Product-SDA, and Simple-vs-Simple Hypothesis Testing}

We make several remarks here on technical differences between our hypothesis testing and statistical dimension setup and those of \cite{FGRVX}.
First, our version of statistical dimension bounds $\E\left [ \left|\iprod{\ol D_u ,\ol D_v} - 1\right| \, | \, A \right ]$ for all events $A$ in the joint distribution of $u,v \sim \mu$, while \cite{FGRVX} considers only $A$ of the form $A = B \tensor B$ for some event $B$ in $\mu$.\footnote{For this reason, we use $\Pr(A) \geq 1/q^2$ in our definition, rather than the more natural $\Pr(A) \geq 1/q$, to maintain consistency with \cite{FGRVX}.}
Our version corresponds to a \emph{stronger} computational model, in the sense that a lower bound on $\SDA(\calS, m)$ implies a lower bound on the statistical dimension of \cite{FGRVX}.
While we are not aware of any natural high-dimensional testing problems where these notions diverge, we give an artificial example where they differ in Appendix~\ref{sec:sda-counterex}.
Second, the problems considered in \cite{FGRVX} are {\em many vs. one} (simple vs. composite) hypothesis testing problems, but in Appendix~\ref{app:mv1} we show that statistical dimension implies lower bounds on SQ algorithms in our simple vs. simple hypothesis testing setting as well.\footnote{The difference between these two settings is the presence of the prior $\mu$.} 
Notationally, we write $\SDA(\calS,m)$ where \cite{FGRVX} writes $\SDA(\calS,D_{\emptyset},\tfrac{1}{m})$.

\subsection{Counterexample to Equivalence of Two Notions of Statistical Dimension}
\label{sec:sda-counterex}

In this appendix we construct a testing problem which shows that the definition of statistical dimension we use in this paper can differ from the statistical dimension of \cite{FGRVX}.
For reference, we repeat both definitions here.

Let $D_\emptyset$ vs. $\calS$ be a testing problem with prior $\mu$.
For $D_u, D_v \in \calS$, we write as usual the relative density $\ol D_u(x) = \frac{D_u(x)}{D_\emptyset(x)}$ (and $\ol D_v$ for $v$), and the inner product $\iprod{\ol D_u, \ol D_v} = \E_{x \sim D_\emptyset} \ol D_u(x) \ol D_v(x)$.
We have used the following notion of statistical dimension:
\begin{definition}[SDA]
  \[
  \SDA(\calS, m) = \max \left \{ q \in \N \, : \, \E_{u,v \sim \mu} \left [ \left|\iprod{\ol D_u ,\ol D_v} - 1\right| \, | \, A \right ] \leq \tfrac 1 m \text{ for all events $A$ s.t. } \Pr_{u,v \sim \mu}(A) \geq \tfrac 1 {q^2} \right \}\mper
  \]
\end{definition}

The work \cite{FGRVX} employs the a different, weaker notion, which we term \emph{product-SDA} or ${\SDA_{\times}}$ to distinguish it from the above:
\begin{definition}[Product SDA]
  \[
  \SDA_{\times}(\calS, m) = \max \left \{ q \in \N \, : \, \E_{u,v \sim \mu} \left [ \left|\iprod{\ol D_u ,\ol D_v} - 1\right| \, | \, A_u, A_v \right ] \leq \tfrac 1 m \text{ for all events $A_u$ s.t. } \Pr_{u \sim \mu}(A) \geq \tfrac 1 { q} \right \}\mper
  \]
  In the definition of product-SDA, the event $A_u \wedge A_v$ is a product of events occurring for a single samples $u,v \sim \mu$, rather than an event over the joint distribution of two samples $u,v \sim \mu$.
  In the definition of SDA, we use $1/q^2$ so that the event $A$ has probability equal to the probability of the event $\{u \in A_u, v \in A_v\}$, where $u \in A_u$ has probability $1/q$ according to $\mu$.
\end{definition}

Since the value of the product-SDA is the value of an optimization problem over a larger set than our notion of SDA, it is clear that $\SDA_\times(m) \ge \SDA(m)$.
We will sketch a proof of the following claim, which demonstrates an example for which this inequality is far from equality.
\begin{claim}
For every $n \in \N$ there is a number $t(n)$ and a family $\calS = \{D_i\}_{i \in [n]}$ of distributions over $[n]$ such that for the hypothesis testing problem $\calS, D_{\emptyset}$ for $D_{\emptyset}$ the uniform distribution over $[n]$, $\SDA(\calS,t(n)) \leq O(1)$ while $\SDA_\times(\calS,t(n)) \geq n^{\Omega(1)}$.
\end{claim}

We turn to our construction.
Regarding notation in what follows: for \emph{vectors} in $\R^n$, which we typically denote by lower-case letters, $\iprod{v,w}$ is the usual Euclidean inner product $\iprod{v,w} = \sum_{i \leq n} v_i w_i$.
For \emph{functions} $F \, : \, [n] \rightarrow \R$, which we denote by upper-case letters, $\iprod{F,G}$ is given by $\E_{i \sim [n]} F(i) G(i)$ (this is merely a difference in normalization).
We will use the following claim.

\begin{claim}
\label{clm:counterex}
  Let $v_1,\ldots,v_n \in \R^n$.
  Let $v_{\max} = \max_i \|v_i\|_\infty$ be the largest-magnitude entry in any $v_i$, and let $\alpha = \max_i |\iprod{v,{\bf 1}}|/\sqrt{n}$, where ${\bf 1}$ denotes the all-$1$'s vector.
  Then there exists a family of distributions $D_1,\ldots,D_n$ on $[n]$ such that, if $\ol D_i$ is the density of $D_i$ relative to the uniform distribution on $[n]$, then $\iprod{\ol D_i, \ol D_j} - 1 = \tfrac 1 {4nv_{\max}^2} (\iprod{v_i,v_j} \pm \alpha^2)$.
\end{claim}
\begin{proof}
  Let $w_i = v_i - \iprod{v_i,{\bf 1}} \cdot {\bf 1} /n $.
  By construction, $\iprod{w_i,{\bf 1}} = 0$.
  Let $\ol D_i \, : \, [n] \rightarrow \R$ be the function $\ol D_i(k) = \frac 1 {2v_{\max}} (w_{ik} + 2v_{\max})$.
  Then by construction $\E_{i \sim [n]} \ol D_i(j) = 1$ and $\ol D_i(j) \geq 0$ for all $i,j$, so $\ol D_i$ is a density relative to the uniform distribution on $[n]$.
  Furthermore,
  \[
  \E_{k \sim [n]} \ol D_i(k) \ol D_j(k) - 1 = \frac 1 n \cdot \frac 1 {4 v_{\max}^2} \iprod{w_i, w_j} = \frac 1 n \cdot \frac 1 {4 v_{\max}^2} ( \iprod{v_i,v_j} - \iprod{v_i,{\bf 1}}\iprod{v_j,{\bf 1}} /n) = \frac 1 n \cdot \frac 1 {4 v_{\max}^2} (\iprod{v_i,v_j} \pm \alpha^2)
  \]
  as desired.
\end{proof}

Now we will construct a \emph{random} testing problem and sketch its analysis.
Let $G$ be an $n \times n$ symmetric matrix with i.i.d. entries from $N(0,1)$.
Let $M = G + 3 \sqrt{n} I$.
With probability at least $0.99$ the following all hold (by standard concentration of measure):
\begin{itemize}
  \item $M \succeq 0$, since the least eigenvalue of $G$ is at most $2\sqrt{n}$ in magnitude, with high probability.
  \item If $v_1,\ldots,v_n \in \R^n$ are such that $\iprod{v_i,v_j} = M_{ij}$, then $|\iprod{v_i,{\bf 1}}|/\sqrt{n} \leq O(\sqrt{\log n}/n^{1/4})$ for all $i$, by rotation-invariance of $M$.
  \item $\max_i \|v_i\|_\infty \leq O(\sqrt{\log n}/n^{1/4})$, again by rotation invariance.
\end{itemize}

Let $\beta = \max_i \|v_i\|_\infty$.
By Claim~\ref{clm:counterex}, there is a family of distributions $D_1,\ldots,D_n$ on $[n]$ such that
for all $i,j$,
\[
|\iprod{\ol D_i,\ol D_j} - 1| = \Abs{\frac 1 n \cdot \frac 1 {4\beta^2} (\iprod{v_i,v_j} \pm O(\log n / \sqrt{n}))} \, .
\]

Now, for all constant $q$, we can find a subset of $n^2/q^2$ entries of $M_{ij}$ such that $M_{ij} = \iprod{v_i,v_j} \geq \Omega(\sqrt{\log q})$.
So there is some constant $C$ such that for all constant $q$,
\[
\SDA \Paren{\{D_i\},\frac{ C n \beta^2 }{\sqrt{\log q}}} \leq q^2\mper
\]

On the other hand, we consider product-$\SDA$ -- we aim to show that product-$\SDA(\{D_i\}, \frac{C n \beta^2}{\sqrt{\log q}}) \gg q^2$.
Take any subset $S \subseteq [n]$ of size $s$.
Then
\[
\frac 1 {n 4 \beta^2} \E_{i,j \sim S} |\iprod{v_i,v_j} \pm O(\log n / \sqrt n)| \leq \frac 1 {4n\beta^2} \left [ (1 \pm o(1))\E_{g \sim \calN(0,1)} |g| + \frac 1 s \cdot O(\sqrt n) + O(\log n /\sqrt{n}) \right ] \mper
\]
We can take $s$ a small as $n^{1-\Omega(1)}$ and still have $\E_{i,j \sim S} |\iprod{\ol D_i, \ol D_j} - 1| \ll \frac{\sqrt{\log q}}{n \beta^2}$, so $\SDA_\times(\{D_i\}, \frac{C n \beta^2}{\sqrt{\log q}}) \geq n^{\Omega(1)}$.

\subsection{Statistical Dimension as a Lower Bound for Hypothesis Testing}\label{app:mv1}

Here, we extend the argument of \cite{FGRVX} which relates the product-statistical dimension to the SQ complexity of many-to-one hypothesis testing to simple hypothesis tests and our more powerful notion of statistical dimension.

\begin{theorem}\label{thm:testing}
Let $\calS = \{D_u\}$ vs. $D_{\emptyset}$ be a hypothesis testing problem with prior $\mu$ on $\calS$.
Let $q,k \in \N$ with $k$ even.
If $\SDA(\frac{3}{t}) > q$, then no $q$-query $\VSTAT(\frac{1}{t})$ algorithm solves the hypothesis testing problem $\calS$ vs. $D_{\emptyset}$.
\end{theorem}
\begin{proof}
We prove the contrapositive.
Let the distributions be supported on $\calX$.
Suppose there is a $q$-query $\VSTAT(1/t)$ algorithm for the testing problem.
Then there must be some $h:\calX \to [0,1]$ which distinguishes between $D_{\emptyset}$ and $D_u \sim \calS$ with probability at least $\frac{1}{q}$ over the choice of $D_u$ given oracle access to $\VSTAT(1/t)$.
Without loss of generality with $\E_{D_\emptyset} h < \frac{1}{2}$, as this affects $p$ by a factor of at most $2$. 
Let $a := \E_{D_\emptyset} h$, and let $a_u = \E_{D_u} h$.

Whenever $h$ succeeds in distinguishing $D_u$ from $D_{\emptyset}$, by definition of $\VSTAT(1/t)$ we have that for every $u$ for which $h$ is successful,
\[
\min\left(\sqrt{ta(1-a)},\sqrt{ta_u(1-a_u)}\right) \le |\langle \ol D_u-1,h \rangle|.
\]
By Lemma 3.5 of \cite{FGRVX} (a simple calculation), using the fact that $a \le \frac{1}{2}$, this further implies that
\[
\sqrt{\frac{ta}{3}} \le \left|\langle \ol D_u -1, h \rangle\right|.
\]
Now for any even $k\in \N$ we have that 
\begin{align*}
\Pr_{u \sim \mu}[h \text{ succeeds on } D_u] \cdot \sqrt{\frac{ta}{3}} 
&\le \E_{u \sim \mu} \left|\langle \ol D_u-1,h\rangle\right|\cdot \Ind[h \text{ succeeds on }D_u]\\
&=\left\langle \E_{u \sim \mu} (\ol D_u - 1) \cdot \sgn(\langle \ol D_u-1,h\rangle) \cdot \Ind[h \text{ succeeds on }D_u], h \right\rangle\\
&\le \|h\| \cdot \sqrt{\E_{u,v \sim \mu} |\left\langle (\ol D_u - 1), (\ol D_v - 1)\right\rangle| \cdot \Ind[h \text{ succeeds on } D_u, D_v]}\\
&= \sqrt{a} \cdot \sqrt{\E_{u,v \sim \mu} |\left\langle \ol D_u , \ol D_v \right\rangle-1 | \cdot \Ind[h \text{ succeeds on } D_u, D_v]},
\end{align*}
where in the penultimate line we have chosen the worst-case signs, and in the final line we have used that $\|h\| = \sqrt{a}$.
Now, we square the above expression and divide by $\Pr_{u \sim \mu}[ h \text{ succeeds on }D_u]^2$:
\[
\frac{t}{3} \le \E_{u,v \sim \mu}\left[|\langle \ol D_u, \ol D_v \rangle -1 | \mid ~ h \text{ succeeds on } D_u, D_v\right],
\]
where we have used that $u,v \sim \mu$ independently.
Furthermore, again by the independence of $u,v\sim\mu$, $\Pr_{u,v \sim \mu}[h \text{ succeeds on }D_u,D_v] \ge \frac{1}{q^2}$.
So by definition of $\SDA$, if $\VSTAT(1/t)$ succeeds then $\SDA(3/t) \le q$.
\end{proof}

\newcommand{\Db}{{\bar D}}

\section{VSTAT Algorithms Imply Low-Degree Distinguishers}
\label{sec:conversion}

In this section, we will give a direct argument that the existence of a VSTAT algorithm implies the existence of a good low-degree algorithm.
We will prove the following theorem, which recovers a nearly identical parameter dependence to Theorem \ref{thm:low-deg-sq} and successfully transfers lower bounds against low-degree algorithms to statistical query algorithms.
However, since $\SDA$ is {\em not} a characterization for $\VSTAT$, and $q$-query $\VSTAT(m)$ algorithms may fail even when  $\SDA(m)< q$, Theorem~\ref{thm:low-deg-sq} is stronger.

\begin{theorem}[VSTAT Algorithms to LDLR] \label{thm:vstat-ldlr}
Let $d,k,m,q \in \mathbb{N}$ with $k$ even, and $\tau,\eta \in (0,1]$. 
Let $D_{\emptyset}$ be a null distribution over $\R^n$, and let $\calS = \{D_v\}_{v \in S}$ be a collection of alternative probability distributions, with $\ol D_u$ the relative density of $D_u$ with respect to $D_{\emptyset}$.
Suppose that the $k$-sample high-degree part of the likelihood ratio of $\calS$ is bounded by $\| \E_{u \sim S}( \ol D_u^{> d})^{\otimes k} \| \le \delta$. 

If there is a (randomized) $q$-query $\VSTAT(1/\tau)$ algorithm which solves the many-vs-one hypothesis testing problem of $D_{\emptyset}$ vs. $\calS = \{D_u\}_{u \in S}$ with probability at least $1-\eta$, then it must follow that
$$\tau \le \frac{4q^{2/k}}{m(1 - \eta)^{2/k}} \left( k\cdot \left\| \E_{u \sim S} (\overline D_u^{\otimes m})^{\le d,k}-1\right\|^{2/k} + \delta^{2/k} m \right)\,.$$
\end{theorem}

The proof of this theorem will consist of two lemmas. The first uses a VSTAT algorithm to construct a good polynomial test of sample-wise degree $(\infty, k)$.

\begin{lemma}\label{lem:q-poly}
Let $m,q$ be non-negative integers, let $k$ be a non-negative even integer, and let $\tau > 0$ and $\eta \in [0,1]$.
If there is a (randomized) $q$-query $\VSTAT(1/\tau)$ algorithm which solves the many-vs-one hypothesis testing problem of $D_{\emptyset}$ vs. $\calS = \{D_u\}_{u \in S}$ with probability at least $1-\eta$, then there is a polynomial $f:(\R^{n})^{\otimes m} \to \R$ of sample-wise degree $(\infty, k)$ such that
\[
\E_{u\sim S}\E_{D_u^{\otimes m}} f \ge (1-\eta) \sqrt{\binom{m}{k}\cdot \left(\frac{\tau}{2}\right)^{k}}, \quad \E_{D_\emptyset^{\otimes m}} f = 0, \quad \text{ and } \quad \sqrt{\E_{D_\emptyset^{\otimes m}} f^2}\le q\,.
\]
Furthermore, $f = \E_{g \sim \Psi} \sum_{\substack{i_1,\ldots,i_k \in [m]\\i_1 < i_2 < \cdots < i_k}} \prod_{\ell=1}^k g(x_{i_\ell}), $
for $\Psi$ a distribution over functions $g:\R^n \to \R$ with $\E_{D_\emptyset}g = 0$.
\end{lemma}

\begin{proof}
Let $\Psi = \psi_1,\ldots,\psi_q:\R^n \to [0,1]$ be any sequence of $q$ statistical queries, and without loss of generality assume that $0 < \E_{D_\emptyset} \psi_t \le \frac{1}{2}$ for all $t \in [q]$.
Call $p_t = \E_{D_\emptyset} \psi_t$, and define $\overline \psi_t(x) := \frac{1}{\sqrt{p_t}}(\psi_t (x) - p_t)$, the re-centered and re-normalized version of $\psi_t$ so that $\E_{D_\emptyset} \overline \psi_t(x) = 0$, and $\E_{D_\emptyset} \overline \psi_t(x)^2 \le 1$.
Define $f_{\Psi}:(\R^n)^{\otimes m}\to \R$ by
\[
f_{\Psi}(x_1,\ldots,x_m) = \sum_{t = 1}^q \left(\sqrt{\frac{1}{\binom{m}{k}}}\sum_{\substack{i_1,\ldots,i_k \in [m]\\ i_1 < i_2 < \cdots < i_k}} \prod_{\ell=1}^k \overline \psi_t(x_{i_\ell})\right)\,.
\]
Since the second summation is over products over $\overline \psi_t$ applied to independent samples, 
\[
\E_{D_\emptyset^{\otimes m}}[f_{\Psi}] = \sum_{t=1}^q\left( \sqrt{\frac{1}{\binom{m}{k}}} \sum_{\substack{i_1,\ldots,i_k \in [m]\\i_1<i_2<\cdots<i_k}} \prod_{\ell=1}^k \E_{D_{\emptyset}} \psi_t\right) = 0\,.
\]
Similarly, for any $\Psi,\Psi'$ we have
\begin{align*}
\E_{D_\emptyset^{\otimes m}} f_\Psi f_{\Psi'}
&= \sum_{s,t \in [q]} \E_{D_{\emptyset}^{\otimes m}}\left[\left(\sqrt{\frac{1}{\binom{m}{k}}}\sum_{\substack{i_1,\ldots,i_k \in [m]\\ i_1 < i_2 < \cdots < i_k}} \prod_{\ell=1}^k \overline \psi_s(x_{i_\ell})\right)\left(\sqrt{\frac{1}{\binom{m}{k}}}\sum_{\substack{i_1,\ldots,i_k \in [m]\\ i_1 < i_2 < \cdots < i_k}} \prod_{\ell=1}^k \overline \psi_t'(x_{i_\ell})\right)\right]\\
&\le q^2 \cdot \max_{\psi \in \Psi \cup \Psi'} \E_{D_\emptyset^{\otimes m}}\left[\left(\sqrt{\frac{1}{\binom{m}{k}}}\sum_{\substack{i_1,\ldots,i_k \in [m]\\ i_1 < i_2 < \cdots < i_k}} \prod_{\ell=1}^k \overline \psi(x_{i_\ell})\right)^2\right]
\le q^2,
\end{align*}
where the final inequality follows because for $i_1<\cdots < i_k$ and $j_1 < \cdots < j_k$,
\[
\E_{D_\emptyset} \left[\prod_{\ell=1}^k \overline \psi(x_{i_\ell})\prod_{\ell=1}^k \overline \psi(x_{j_\ell})\right] = \Ind[(i_1,\ldots, i_k)=(j_1, \ldots,j_k)] \cdot (\E_{D_\emptyset} \overline \psi^2)^k,
\]
And because $\E_{D_\emptyset} \overline \psi^2 \le 1$.
Therefore, for any distribution $Q$ over $\Psi$,
\[
\E_{D_\emptyset^{\otimes m}}\left[ \E_{\Psi \sim Q} f_\Psi\right] \le 0,\quad \text{ and }\quad \E_{D_\emptyset^{\otimes m}}\left[ \left(\E_{\Psi \sim Q} f_\Psi\right)^2\right] \le q^2.
\]

Now, supposing that $Q$ is a distribution over $\Psi$ so that with probability at least $1-\eta$ over $u \sim S$, the queries in $\Psi$ give a $\VSTAT(1/\tau)$ algorithm for distinguishing $D_u,D_\emptyset$; that is, with probability at least $1-\eta$ over $u \sim S, \Psi \sim Q$, we have the event 
\[
\calE 
:= \left\{\max_{t \in [q]} \left|\E_{D_u} \psi_{t} - \E_{D_\emptyset} \psi_t\right|\ge \max\left(\tau,\sqrt{\tau p_t (1-p_t)}\right)\right\} 
\implies \left\{\max_{t \in [q]} \left|\E_{D_u} \overline \psi_t\right| \ge \sqrt{\frac{\tau}{2}}\right\},
\]
where we have used the definition of $\overline \psi_t$ and the fact that $(1-p_t) > \frac{1}{2}$ by assumption.
This implies
\begin{align*}
\E_{u} \E_{D_u^{\otimes m}} \E_{\Psi \sim Q} f_\Psi
&= \E_{u} \E_{\Psi \sim Q}\left[\sum_{t=1}^q\E_{D_u^{\otimes m}} \left[ \left(\sqrt{\frac{1}{\binom{m}{k}}}\sum_{\substack{i_1,\ldots,i_k \in [m]\\i_1 < i_2 <\cdots < i_k}} \prod_{\ell=1}^k\overline \psi_t(x_{i_\ell})\right)\right]\right]\\
&= \E_{u} \E_{\Psi \sim Q}\left[\sum_{t=1}^q \sqrt{\binom{m}{k}} \left(\E_{D_u}\overline \psi_t\right)^{k}\right]\qquad\qquad\text{(independence of the $x_\ell$'s)}\\
&\ge (1-\eta) \E_{u} \E_{\Psi \sim Q}\left[\sum_{t=1}^q \sqrt{\binom{m}{k}} \left(\E_{D_u}\overline \psi_t\right)^{k} \mid \calE\right]\\
&\ge (1-\eta) \sqrt{\binom{m}{k}} \cdot \left(\sqrt{\frac{\tau}{2}}\right)^{k},
\end{align*}
where in the third line we use the law of conditional expectation and the fact that $k$ is even to drop the expectation in the event $\overline \calE$, and in the final line we use the implication of $\calE$ and the fact that $k$ is even.
Letting $f := \E_{\Psi\sim Q} f_\Psi$, our conclusion now follows by linearity of expectation.
\end{proof}

We now will show that if the $k$-sample high-degree part of the likelihood ratio of $\calS$ is bounded, then a good polynomial test of sample-wise degree $(\infty, k)$ also implies one of samplewise degree $(d, k)$. We remark that the resulting test is not necessarily the degree $(d, k)$-projection $f^{\le d, k}$ of the degree $(\infty, k)$ test $f$.
 We instead bound the distance between $f$ and $f^{\le d, k}$ directly by $(d, k)$-$\textnormal{LDLR}_m$.
 This amounts to showing that if $f$ and $f^{\le d, k}$ are far, then there must be a different good polynomial test of sample-wise degree $(d, k)$.
 This argument is carried out below.

\begin{lemma}\label{lem:new-trunc}
Let $D_{\emptyset}$ vs. $\calS$ be a hypothesis testing problem over $\R^n$, and suppose that the $k$-sample high-degree part of the likelihood ratio of $\calS$ is bounded, $\| \E_{u \sim S}( \ol D_u^{> d})^{\otimes k} \| \le \delta$. 
Let $\Psi$ be a distribution over functions from $\R^n \to \R$.
If $f:(\R^n)^{\otimes m} \to \R$ is a sample-wise degree-$(\infty,k)$ polynomial of the form 
\[
f(x_1,\ldots,x_m) = \E_{g \sim \Psi} \sum_{\substack{i_1,\ldots,i_k \in [m]\\i_1 < i_2 < \cdots < i_k}} \prod_{\ell=1}^k g(x_{i_\ell})\,,
\] 
and $\E_{D_\emptyset} g = 0$ for all $g \sim \Psi$, then we have that
$$\left( \left\| \E_{u \sim S} (\overline D_u^{\otimes m})^{\le d,k}-1\right\|^{2/k} + \delta^{2/k} \cdot \binom{m}{k}^{1/k} \right)^{k/2} \ge \frac{1}{2} \cdot \frac{\E_{u}\E_{D_u^{\otimes m}} f}{\sqrt{\E_{D_\emptyset^{\otimes m}} f^2}}\,.$$
\end{lemma}

\begin{proof}
Since the samples $x_1,\ldots,x_m \sim D_u^{\otimes m}$ are independent and identically distributed, the moments of $f$ under the $m$-sample distribution $D^{\otimes m}$ are within a multiplicative factor of the moments of one of the summands under the $k$-sample distribution $D^{\otimes k}$,
\begin{align}
\E_u \E_{D_u^{\otimes m}} f 
= \E_{g \sim \Psi}\sum_{\substack{i_1,\ldots,i_k \in [m]\\i_1 < \cdots < i_k}} \E_u \E_{D_u^{\otimes m}}\left[\prod_{\ell=1}^k g(x_{i_\ell})\right]
&= \binom{m}{k} \cdot \E_{g \sim \Psi} \E_u \E_{D_u^{\otimes k}}\left[\prod_{\ell=1}^k g(x_{\ell})\right].\label{eq:sum}
\end{align}
For any $g \sim \Psi$, let $g^{\le d}$ be its sample-wise degree $(d,\infty)$ projection, and let $g^{\otimes k}(x_1,\ldots,x_k) = \prod_{i=1}^k g(x_i)$.
We have that
\begin{align*}
\E_{g \sim \Psi} \E_u\E_{D_u^{\otimes k}} (g^{\le d})^{\otimes k} &= \left \langle  \E_u (\overline D_u^{\le d})^{\otimes k}, \E_{g \sim \Psi} g^{\otimes k}\right\rangle \\
&= \E_u \E_{D_u^{\otimes k}} g^{\otimes k} - \left\langle \E_u \left(\overline D_u^{\otimes k} - (\overline D_u^{\le d})^{\otimes k}\right), \E_{g \sim \Psi} g^{\otimes k}\right\rangle\\
&\ge \E_u \E_{D_u^{\otimes k}} g^{\otimes k} - \left\| \E_{g \sim \Psi} g^{\otimes k} \right\| \cdot \left\|\E_u \left(\overline D_u^{\otimes k} - (\overline D_u^{\le d})^{\otimes k}\right) \right\| 
\end{align*}
by Cauchy-Schwarz.
Now note that $\E_u(\overline D_u^{\le d})^{\otimes k}$ is the orthogonal projection of $\E_u \overline D_u^{\otimes k}$ onto the set of degree-$(d, k)$ polynomials.
This set contains all constant polynomials and the projection of $\E_u(\overline D_u^{\le d})^{\otimes k}$ onto the set of constant polynomials is $1$. 
Combining this with Lemmas~\ref{lem:holder}, we have
\begin{align*}
\left\|\E_u \left(\overline D_u^{\otimes k} - (\overline D_u^{\le d})^{\otimes k}\right) \right\|^2 
\le \left\|\E_u \overline D_u^{\otimes k} - 1 \right\|^2 
&= \E_{u, v} \left( \langle \ol D_u, \ol D_v \rangle - 1 \right)^k \\
&\le \left(\E_{u,v}\left[ \left(\langle \ol D_u^{\le d}, \ol D_v^{\le d} \rangle-1\right)^k \right]^{1/k} + \delta^{2/k} \right)^k \\
&\le \left( \frac{1}{\binom{m}{k}^{1/k}} \cdot \left\| \E_{u \sim S} (\overline D_u^{\otimes m})^{\le d,k}-1\right\|^{2/k} + \delta^{2/k} \right)^k,
\end{align*}
where the last line is from Lemma \ref{lem:boosting}. 
Returning to (\ref{eq:sum}), by linearity of projection to sample-wise degree $(d,k)$ and since $f$ is already sample-wise degree-$(\infty,k)$, we have that
\begin{align}
\E_u \E_{D_u^{\otimes m}} f^{\le d,k} 
= \E_{u}\E_{D_u^{\otimes m}} f - \binom{m}{k} \cdot\left\| \E_{g \sim \Psi} g^{\otimes k} \right\| \left( \frac{1}{\binom{m}{k}^{1/k}} \cdot \left\| \E_{u \sim S} (\overline D_u^{\otimes m})^{\le d,k}-1\right\|^{2/k} + \delta^{2/k} \right)^{k/2}\,, \label{eq:cauch}
\end{align}
where we used the independence of the samples to equate $\binom{m}{k}\E_{g\sim\Psi}\E_u\E_{D_u^{\otimes k}} g^{\otimes k}$ and $\E_u \E_{D_u^{\otimes m}} f$. 

By independence of samples, the terms $\prod_{\ell=1}^k g(x_{i_\ell})$ and $\prod_{\ell = 1}^k h(x_{j_\ell})$ are uncorrelated when $x \sim D_{\emptyset}^{\otimes m}$, unless $i_1,\ldots,i_k = j_1,\ldots,j_k$.
Using the fact that for every $g \sim \Psi$, $\E_{D_\emptyset} g = 0$, and the independence of the samples, this implies that
\begin{align}
\E_{D_\emptyset^{\otimes m}} f^2 
&=\E_{g,h \sim \Psi} \sum_{\substack{i_1,\ldots,i_k \in [m]\\i_1 < \cdots < i_k}} \E_{D_\emptyset^{\otimes m}}\left[\prod_{\ell=1}^k g(x_{i_\ell})h(x_{i_\ell}) \right]\nonumber \\
&= \E_{g,h\sim \Psi}\binom{m}{k} \cdot \E_{D_\emptyset^{\otimes k}}\left[\prod_{\ell=1}^k g(x_{\ell})h(x_{\ell}) \right] = \binom{m}{k} \cdot \left\|\E_{g\sim \Psi} g^{\otimes k}\right\|^2. \label{eq:var}
\end{align}
Therefore we have that
\begin{align}
\left\| \E_{u \sim S} (\overline D_u^{\otimes m})^{\le d,k}-1\right\| 
&\ge \frac{\E_u \E_{D_u^{\otimes m}} f^{\le d,k}}{\sqrt{\E_{D_\emptyset^{\otimes m}} (f^{\le d, k})^2}} \nonumber\\
&\ge \frac{\E_u \E_{D_u^{\otimes m}} f^{\le d,k}}{\sqrt{\E_{D_\emptyset^{\otimes m}} f^2}} \nonumber \\
&\ge \frac{\E_{u}\E_{D_u^{\otimes m}} f}{\sqrt{\E_{D_\emptyset^{\otimes m}} f^2}} - \left( \left\| \E_{u \sim S} (\overline D_u^{\otimes m})^{\le d,k}-1\right\|^{2/k} + \delta^{2/k} \cdot \binom{m}{k}^{1/k} \right)^{k/2}\,.\label{eq:compare}
\end{align}
The first inequality follows from the fact that the left-hand side gives the optimal signal to noises ratio among all sample-wise degree-$(d, k)$ polynomials for the distinguishing problem of $D_\emptyset^{\otimes m}$ versus $\E_u D_u^{\otimes m}$ (see Section~\ref{sec:prelims}).
The second inequality follows since $f^{\le d, k}$ is a projection of $f$ onto a convex set, and the final inequality follows by combining (\ref{eq:cauch}) and (\ref{eq:var}).
Finally, note that
$$\left\| \E_{u \sim S} (\overline D_u^{\otimes m})^{\le d,k}-1\right\| \le \left( \left\| \E_{u \sim S} (\overline D_u^{\otimes m})^{\le d,k}-1\right\|^{2/k} + \delta^{2/k} \cdot \binom{m}{k}^{1/k} \right)^{k/2},$$
Applying this after rearranging (\ref{eq:compare}) now completes the proof of the lemma.
\end{proof}

Theorem \ref{thm:vstat-ldlr} now follows immediately on applying these two lemmas.

\begin{proof}[Proof of Theorem \ref{thm:vstat-ldlr}]
Let $f$ be as in Lemma \ref{lem:q-poly}. Combining Lemmas \ref{lem:q-poly} and \ref{lem:new-trunc} now yields that
\begin{align*}
q^{-1}(1-\eta) \sqrt{\binom{m}{k}\cdot \left(\frac{\tau}{2}\right)^{k}} &\le \frac{\E_{u}\E_{D_u^{\otimes m}} f}{\sqrt{\E_{D_\emptyset^{\otimes m}} f^2}} \\
&\le 2\left( \left\| \E_{u \sim S} (\overline D_u^{\otimes m})^{\le d,k}-1\right\|^{2/k} + \delta^{2/k} \cdot \binom{m}{k}^{1/k} \right)^{k/2}\,.
\end{align*}
Rearranging and upper bounding $2^{1 + 2/k} \le 4$ yields that
$$\tau \le \frac{4q^{2/k}}{(1 - \eta)^{2/k}} \left( \frac{1}{\binom{m}{k}^{1/k}} \cdot \left\| \E_{u \sim S} (\overline D_u^{\otimes m})^{\le d,k}-1\right\|^{2/k} + \delta^{2/k} \right)\,.$$
The fact that $(m/k)^k \le \binom{m}{k}$ now completes the proof of the theorem.
\end{proof}

\section{Proofs of Cloning Facts}\label{app:clone}

\begin{lemma*}[Restatement of Lemma~\ref{lem:g-clone}]
There is a randomized algorithm taking as input a real number $x$ and outputting $m$ independent random variables $Y_1,\dots, Y_m$ such that for any $\mu\in \R$
if $x\sim \calN(\mu,1)$ , then $Y_i\sim \calN(\mu/\sqrt{m}, 1)$. 
\end{lemma*}

\begin{proof}
Let $U\in \R^{m\times m}$ be a matrix with all entries in the first column equal to $1/\sqrt{m}$ and with remaining columns chosen so that $U$ is orthogonal, i.e., $U^\top U=I_m$. Generate independent variables $Z_2,\dots, Z_m\sim \calN(0,1)$ and let $Z = (X,Z_2,\dots, Z_m)^\top$. Now put $Y =  UZ$. Note that $Z \stackrel{d}= \mu\cdot e_1 +  W$, where $W\sim \calN(0, I_m)$ and $e_1$ is the first standard basis vector, and the result follows since $UW \stackrel{d}= W$.
\end{proof}

\begin{lemma*}[Restatement of Lemma~\ref{lem:p-clone}]
There is an algorithm that when given $m$ independent samples from $\calG(n,U,\gamma)$ for any $U \subseteq [n]$, efficiently produces a single instance distributed according to $\calG(n,U,\gamma^m)$. 
	Conversely, there is an efficient algorithm taking a graph as input and producing $m$ random graphs, such that given an instance of planted clique $\calG(n,U,\gamma)$ with unknown clique position $U$, produces $m$ independent samples from $\calG(n,U,\gamma^{1/m})$. 
\end{lemma*}
\begin{proof}
The first direction is immediate: given $Y_1,\dots, Y_m\sim \calG(n,U,\gamma)$, form the graph $X$ by letting $X_e=\prod_{i\in [m]} Y_{i,e}$. 
For the other direction, we will show how to produce $m$ independent Bernoulli variables with appropriate bias from a single Bernoulli. The claim for planted clique will then follow immediately by applying the procedure to the edge indicators of the input graph.

Suppose that $p\in \{\gamma, 1\}$ for some $\gamma\in [0,1]$. We describe how to map a single $x\sim \mathrm{Bern}(p)$ to $(y_1,\dots, y_m)\sim \mathrm{Bern}(p^{1/m})^{\otimes m}$ without knowing which is the true value of $p$. Given input $x=1$, output $y_1=\cdots = y_m=1$. Now suppose $x=0$. Let $y=v$ for each $v\in \{0,1\}^m\setminus\{\mathbf{1}\}$ with probability $(\gamma^{|v|_1/m}(1-\gamma^{1/m})^{m-|v|_1})/(1-\gamma)$, where $|v|_1 = \sum v_i$ is the number of ones in $v$. Note that this probability mass function is exchangeable and thus can be sampled in $\text{poly}(m)$ time as follows. First sample the support size $|y|_1 \in \{0, 1, \dots, m - 1\}$, which has distribution explicitly given by $\Pr(|y|_1 = x) = \binom{m}{x} \gamma^{x/m} (1 - \gamma^{1/m})^{m - x}/(1 - \gamma)$ since the distribution of $y$ is exchangeable. Then produce $y$ by sampling a random binary string in $\{0, 1\}^m$ with support size exactly $|y|_1$, uniformly at random.

To check that the output distribution of $(y_1,\dots, y_m)$ is indeed $\mathrm{Bern}(p^{1/m})^{\otimes m}$ for $p\in \{\gamma, 1\}$, first observe that if $p=1$ then $x=1$ deterministically and so too are $y_1,\dots, y_m$. If $p=\gamma$, then 
$$
\Pr(y=v) = \gamma\cdot  \Id_{v = \mathbf{1}} + (1-\gamma)\cdot \Id_{v \neq \mathbf{1}}\cdot  \frac {\gamma^{|v|_1/m}(1-\gamma^{1/m})^{m-|v|_1}}{1-\gamma} = (\gamma^{1/m})^{|v|_1}(1-\gamma^{1/m})^{m-|v|_1}\,,
$$
which is precisely the probability mass function of $\mathrm{Bern}(\gamma^{1/m})^{\otimes m}$. 
\end{proof}

\section{Omitted Calculations from Applications}\label{app:ex-cal}
In this section, we include the calculations omitted from Section~\ref{sec:examples}.

\subsection{Tensor PCA}\label{sec:tpca-app}

\restateclaim{claim:tpca-k}
\begin{proof}
To obtain the first conclusion, we expand
\[
\left\|\E_{u \sim S} \ol D_u^{\otimes k}\right\|^2  
= \E_{u,v} \langle\ol D_u, \ol D_v \rangle^k
= \E_{u,v} \exp(k\lambda \langle u, v \rangle^r),
\]
Where for the final equality we have used a simple calculation analogous to that in the proof of Proposition 2.5 of \cite{BKW19}.
Since $\langle u,v \rangle$ for $u,v$ sampled uniformly independently from $S$ is distributed as the mean of $n$ Rademacher random variables, we have that $\Pr[|\langle u, v \rangle| \ge \frac{C}{\sqrt n}] \le 2\exp(- \frac{C^2}{2})$, and $|\langle u,v \rangle|\le 1$.
So we have
\begin{align*}
\E_{u,v} \exp(k \lambda^2 \langle u, v \rangle^r)
\le \E_{u,v} \exp(k \lambda^2 |\langle u, v \rangle|^r)
&\le 2\int_{0}^{\sqrt{n}} \exp\left(k \lambda^2 \left(\frac{C}{\sqrt{n}}\right)^r - \frac{C^2}{2}\right) dC\\
&\le 2\int_{0}^{\sqrt{n}} \exp\left(-\frac{1}{2}\left(1-\frac{2k \lambda^2}{n}\right) C^2 \right) dC
\le \sqrt{\frac{2\pi}{1 - \frac{2k\lambda^2}{n}}},
\end{align*}
where to obtain the second line we have substituted $C = \sqrt{n}$ for $r-2$ copies of $C$, and to obtain the final conclusion we have used that $2k\lambda^2 < n$ and the expression for the Gaussian probability density function.
\end{proof}

\restateclaim{claim:tpca-ldlr}
\begin{proof}
For a given $D_u = \calN(\lambda u^{\otimes r}, \Id_{n^r})$, from $D_u^{\otimes m}$ we have $m$ samples samples be $\{T_i\}_{i = 1}^m$ with each $T_i = \lambda u^{\otimes r} + G_i$, where $G_i \sim \calN(0,\Id_{n^r})$ are independent across samples.
We will use the Fourier basis for $(D_{\emptyset}^{\otimes m})^{\le 1,k}-1$ , which is given by
\[
\left\{\chi_{S}\, \mid \, S \in \bigcup_{t = 1}^{k}\binom{[n]^r}{1}^{\otimes t} \times \binom{[m]}{t}\right\},
\]
that is, for each $S = \{(A_\ell , j_\ell)\}_{\ell = 1}^t$, which specifies a collection $(A_1,\ldots,A_t)$ of $t$ indices in $(\R^n)^{\otimes r}$ and $t$ sample indices $(j_1,\ldots,j_t)$ in $[m]$, we take $\chi_S(T_1,\ldots,T_m) = \prod_{\ell = 1}^{t} (T_{j_\ell})_{S_\ell}$.
For any such $S$ with $|S| = t$, we may compute
\[
\E_{u} \E_{T_1,\ldots,T_m \sim D_u}[\chi_S(T_1,\ldots,T_m)]
= \E_{u} \prod_{\ell = 1}^t \left(\lambda u^{A_\ell} + G^{(\ell)}_{A_\ell}\right)
= \left(\frac{\lambda}{(\sqrt n)^r}\right)^{|S|} \cdot \Ind[S \text{ is even}],
\]
where by ``$S$ is even'' we mean that the multiset $\cup_{\ell = 1}^t A_\ell$ contains every $i \in [n]$ with even multiplicity.
This is because the indices $j_1,\ldots,j_t \in [m]$ are all distinct, so any term in the expansion of the product with nonzero degree in the $G_{A_\ell}^{(\ell)}$ variables has expectation $0$, and for any multiset of indices $B \subset [n]^r$, $\E_{u} u^B = 0$ if any index appears in $B$ with odd multiplicity, and $\E_u u^{B} = n^{-r|B|/2}$ otherwise.

The even $S$ of size $t$ for a fixed set of samples $j_1,\ldots,j_t \in \binom{[m]}{t}$ are in bijection with $t$-edge hypergraph with hyperedges from $[n]^r$ in which every vertex has even degree.
Since there can be at most $rt/2$ vertices in such a hypergraph, and once the vertex set is fixed there are at most $\frac{(\frac{rt}{2}!)}{2^{rt/2} (r!)^t}$ ways of choosing an even hypergraph on them according to the configuration model (assign every vertex 2 half-edges, assign every hyperedge $r$ half-edges, and then count the number of distinct matchings),
\[
|\{S\, \mid\, |S| = t, \, S \text{ even}\}| 
\le\binom{m}{t}\cdot n^{rt/2}\cdot \left(\frac{\frac{rt}{2}!}{2^{rt/2} (r!)^t}\right)
\le \left(\frac{em}{t}\right)^t\cdot n^{rt/2}\cdot \left(t\right)^{rt/2},
\]
where we have applied Stirling's approximation and used that $r \ge 2$.
Thus, we can bound the $\LDLR$,
\begin{align*}
\left\|\E_u (\ol D_u^{\otimes m})^{\le 1,k} -1\right\|^2 
&= \sum_{t = 1}^k |\{S \mid |S| = t,\, S\text{ even}\}| \cdot \E_{u}\E_{D_u^{\otimes m}} [\chi_S]^2 \\
&\le \sum_{t=1}^k \left(e m t^{(r-2)/2} n^{r/2}\right)^t \cdot \left(\frac{\lambda}{n^{r/2}}\right)^{2t}\\
&= \sum_{t=1}^k \left(\frac{e m \lambda^2 t^{(r-2)/2}}{n^{r/2}}\right)^{t}
\,\, \le \sum_{t=1}^k \left(\frac{e m \lambda^2 k^{(r-2)/2}}{n^{r/2}}\right)^{t}
\,\, \le 2\frac{e m \lambda^2 k^{(r-2)/2}}{n^{r/2}},
\end{align*}
where in the final line we have used that $2 e m \lambda^2 k^{(r-2)/r} \le n^{r/2}$ and the fact that the sum is geometric.
\end{proof}

\subsection{Planted Clique} \label{sec:pc-app}

\restateclaim{claim:bpds-ldlr}

\begin{proof}
We will compute the Fourier coefficients of $\ol D = \E_{u \sim \mu} \ol D_u^{\otimes m}$ as a function on $\{0, 1\}^{m \times N}$. For each $m$-tuple of subsets $\alpha = (\alpha_1, \alpha_2, \dots, \alpha_m)$ where $\alpha_i \subseteq [N]$, define the Fourier character
$$\chi_\alpha(x) = \prod_{i = 1}^m \prod_{j \in \alpha_i} \frac{x_{ij} - q}{\sqrt{q(1 - q)}}$$
for each $x \in \{0, 1\}^{m \times N}$. Note that the $\chi_{\alpha}$ form an orthogonal basis with respect to $D_{\emptyset}^{\otimes m}$. For each $\alpha$, let $L(\alpha) = \alpha_1 \cup \alpha_2 \cup \cdots \cup \alpha_m$ and $R(\alpha) = \{ i \in m : \alpha_i \neq \emptyset \}$. A direct computation yields that the Fourier coefficients of $\ol D$ are given by
$$\widehat{D}(\alpha) = \E_{u \sim \mu} \E_{x \sim D_u^{\otimes m}} \chi_\alpha(x) = \left( \frac{K}{N} \right)^{|L(\alpha)| + |R(\alpha)|} \gamma^{\frac{1}{2} \sum_{i = 1}^m |\alpha_i|}$$
By Parseval's identity, we now have that
\begin{align*}
\left\|\E_{u \sim \mu} (\ol D_u^{\otimes m})^{\le d, k} - 1 \right\|^2 &= \sum_{t = 1}^k \binom{m}{t} \sum_{1 \le |\alpha_1|, \dots, |\alpha_t| \le d} \widehat{D}(\alpha_1, \dots, \alpha_t, \emptyset, \dots, \emptyset)^2 \\
&= \sum_{t = 1}^k \binom{m}{t} \sum_{1 \le |\alpha_1|, \dots, |\alpha_t| \le d} \left( \frac{K}{N} \right)^{2|L(\alpha)| + 2t} \gamma^{\sum_{i = 1}^t |\alpha_i|} \numberthis \label{eqn:ldlrbpc}
\end{align*}
Here, we have used the fact that $\widehat{D}(\alpha) = \widehat{D}(\alpha_\sigma)$ where $\alpha_\sigma = (\alpha_{\sigma(1)}, \alpha_{\sigma(2)}, \dots, \alpha_{\sigma(m)})$ for all $\sigma \in S_m$, by symmetry. Now note that for any fixed $A \subseteq [N]$, we have that
\begin{align*}
\sum_{1 \le |\alpha_1|, \dots, |\alpha_t| \le d \, : \, L(\alpha) = A} \left( \frac{K}{N} \right)^{2|L(\alpha)| + 2t} \gamma^{\frac{1}{2} \sum_{i = 1}^t |\alpha_i|} &\le \left( \frac{K}{N} \right)^{2|A| + 2t} \sum_{1 \le |\alpha_1|, \dots, |\alpha_t| \le d \, : \, L(\alpha) \subseteq A} \gamma^{\sum_{i = 1}^t |\alpha_i|} \\
&= \left( \frac{K}{N} \right)^{2|A| + 2t} \left( \sum_{\ell = 1}^{\min(d, |A|)} \binom{|A|}{\ell} \gamma^{\ell} \right)^t \\
&\le \left( \frac{K}{N} \right)^{2|A| + 2t} (1 + \gamma)^{|A|t}
\end{align*}
where the last inequality follows from the observation
$$\sum_{\ell = 1}^{\min(d, |A|)} \binom{|A|}{\ell} \gamma^{\ell} \le \sum_{\ell = 0}^{|A|} \binom{|A|}{\ell} \gamma^{\ell} = (1 + \gamma)^{|A|}$$
Note that $|L(\alpha)|$ can vary between $1$ and $kd$. The fact that there are $\binom{N}{s} \le N^{s}$ possible $A$ with a given fixed size $|A| = s$ combined with Equation (\ref{eqn:ldlrbpc}) now yields that
\begin{align*}
\left\|\E_{u \sim \mu} (\ol D_u^{\otimes m})^{\le d, k} - 1 \right\|^2 &\le \sum_{t = 1}^k \sum_{s = 1}^{kd} m^t N^s \left( \frac{K}{N} \right)^{2s + 2t} (1 + \gamma)^{ts} \\
&\le \sum_{t = 1}^k \sum_{s = 1}^{kd} \left( \frac{K^2 m}{N^2} \right)^{t} \left( \frac{K^2 (1 + \gamma)^{k}}{N} \right)^{s} 
\end{align*}
where the second inequality follows from the fact that $(1 + \gamma)^{ts} \le (1 + \gamma)^{ks}$ and rearranging. Under the given condition, this upper bound is the product of two geometric series with ratios $1 - \Omega_N(1)$, completing the proof of the claim.
\end{proof}

\restateclaim{claim:bpds-k}

\begin{proof}
The follows from Claim \ref{claim:bpds-ldlr} applied with $d = N$ and $m = k$, and the observation
$$\left\|\E_{u \sim \mu} \ol D_u^{\otimes k} \right\|^2 = \left\|\E_{u \sim \mu} (\ol D_u^{\otimes k})^{\le N, k} - 1 \right\|^2 + 1$$
since $(\ol D_u^{\otimes k})^{\le N, k} = \ol D_u^{\otimes k}$ and $\langle \E_{u \sim \mu} \ol D_u^{\otimes k}, 1 \rangle = 1$.
\end{proof}

\restateclaim{claim:mshpc-ldlr}

\begin{proof}
Similar to as in Claim \ref{claim:bpds-ldlr}, we will compute the Fourier coefficients of $\ol D = \E_{u \sim \mu} \ol D_u^{\otimes m}$ as a function on $\{0, 1\}^{m \times H}$ where $H = \binom{[N]}{s}$. The relevant orthogonal basis of Fourier characters is indexed by $m$-tuples of families of subsets $\alpha = (\alpha_1, \alpha_2, \dots, \alpha_m)$ where $\alpha_i \subseteq H$ and given by
$$\chi_\alpha(x) = \prod_{i = 1}^m \prod_{e \in \alpha_i} \frac{x_{ie} - q}{\sqrt{q(1 - q)}}$$
for each $x \in \{0, 1\}^{m \times H}$. Given some $\alpha_i \subseteq H$, let $V(\alpha_i) = \bigcup_{\{v_1, v_2, \dots, v_s\} \in \alpha_i} \{v_1, v_2, \dots, v_s\}$ be the vertex set of the hyperedges in $\alpha$. Furthermore, let $V(\alpha) = V(\alpha_1) \cup V(\alpha_2) \cup \cdots \cup V(\alpha_m)$ where $\alpha = (\alpha_1, \alpha_2, \dots, \alpha_m)$. Note that $\E_{x \sim D_u^{\otimes m}} \chi_\alpha(x) = 0$ unless $V(\alpha) \subseteq u$, which occurs with probability $\binom{K}{|V(\alpha)|}/\binom{N}{|V(\alpha)|}$ if $u \sim \mu$. Therefore the Fourier coefficients of $\ol D$ are then given by
$$\widehat{D}(\alpha) = \E_{u \sim \mu} \E_{x \sim D_u^{\otimes m}} \chi_\alpha(x) = \frac{\binom{K}{|V(\alpha)|}}{\binom{N}{|V(\alpha)|}} \cdot \gamma^{\frac{1}{2} \sum_{i = 1}^m |\alpha_i|} \le \left( \frac{eK}{N} \right)^{|V(\alpha)|} \gamma^{\frac{1}{2} \sum_{i = 1}^m |\alpha_i|}$$
where the inequality follows from $(a/b)^b \le \binom{a}{b} \le (ea/b)^b$. The same application of Parseval's as in Claim \ref{claim:bpds-ldlr} now yields that
$$\left\|\E_{u \sim \mu} (\ol D_u^{\otimes m})^{\le d, k} - 1 \right\|^2 \le \sum_{t = 1}^k \binom{m}{t} \sum_{1 \le |\alpha_1|, \dots, |\alpha_t| \le d} \left( \frac{eK}{N} \right)^{2|V(\alpha)|} \gamma^{\sum_{i = 1}^t |\alpha_i|}$$
We now have that for any $A \subseteq [N]$,
\begin{align*}
\sum_{1 \le |\alpha_1|, \dots, |\alpha_t| \le d \, : \, V(\alpha) = A} \left( \frac{eK}{N} \right)^{2|V(\alpha)|} \gamma^{\sum_{i = 1}^t |\alpha_i|} &\le \left( \frac{eK}{N} \right)^{2|A|} \sum_{1 \le |\alpha_1|, \dots, |\alpha_t| \le d \, : \, \alpha_i \subseteq \binom{A}{s}} \gamma^{\sum_{i = 1}^t |\alpha_i|} \\
&= \left( \frac{eK}{N} \right)^{2|A|} \left( \sum_{\ell = 1}^{\min\left(d, \binom{|A|}{s}\right)} \binom{\binom{|A|}{s}}{\ell} \gamma^{\ell} \right)^t \\
&\le \left( \frac{eK}{N} \right)^{2|A|}  \gamma^t \binom{|A|}{s}^t (1 + \gamma)^{\binom{|A|}{s}t}
\end{align*}
The last inequality holds because of the following observation
$$\sum_{\ell = 1}^{\min(d, y)} \binom{y}{\ell} \gamma^{\ell} \le y \gamma \cdot \sum_{\ell = 1}^{\min(d, y)} \binom{y - 1}{\ell - 1} \gamma^{\ell - 1} \le y \gamma (1 + \gamma)^y$$
for any $y \in \mathbb{N}$. Note that if $\alpha = (\alpha_1, \alpha_2, \dots, \alpha_t)$ satisfies that that $1 \le |\alpha_i| \le d$, then $s \le |V(\alpha)| \le ksd$. Give that there are $\binom{N}{a} \le N^a$ sets $A \subseteq [N]$ of a fixed size $|A| = a$, we have
\begin{align*}
\left\|\E_{u \sim \mu} (\ol D_u^{\otimes m})^{\le d, k} - 1 \right\|^2 &\le \sum_{t = 1}^k \sum_{a = s}^{ksd} \frac{m^t}{t!} \cdot N^a \left( \frac{eK}{N} \right)^{2a}  \gamma^t a^{st} (1 + \gamma)^{a^s t} \\
&= \sum_{a = s}^{ksd} \left( \frac{e^2K^2}{N} \right)^{a} \sum_{t = 1}^k \frac{\left( m \gamma a^{s} (1 + \gamma)^{a^s} \right)^t}{t!} \\
&\le \sum_{a = s}^{ksd} a^{sk} \left( \frac{e^2K^2}{N} \right)^{a} \sum_{t = 1}^k \frac{\left( m \gamma \cdot \exp(\gamma k^s s^s d^s) \right)^t}{t!} \\
&\le \left( \sum_{a = s}^{ksd} \left( \frac{2^{sk} e^2 K^2}{N} \right)^{a} \right) \cdot \exp \left( m \gamma \cdot \exp(\gamma k^s s^s d^s) \right)
\end{align*}
The second last line follows from the inequalities $a^{st} \le a^{sk}$, $a \le ksd$ and $1 + \gamma \le \exp(\gamma)$. The last line follows from the fact that if $x > 0$, $\sum_{t = 1}^k x^t/t! \le \exp(x)$ and $a^{sk} \le 2^{ask}$ since $a \ge 1$. The given conditions now imply that the exponential factor is $O_N(1)$ and that the geometric series has ratio $1 - \Omega_N(1)$ and thus is also $O_N(1)$, completing the proof of the claim.
\end{proof}

\restateclaim{claim:mshpc-k}

\begin{proof}
Note that $\ol D_u(x) = \prod_{e \in \binom{u}{s}} q^{-1} x_e$ for each $x \in \{0, 1\}^{\binom{[N]}{s}}$. Therefore we have that
\begin{align*}
\langle \ol D_u, \ol D_v \rangle &= \E_{x \sim D_\emptyset} \left[ \prod_{e \in \binom{u \cap v}{s}} q^{-2} x_e \prod_{e \in \binom{u}{s} \Delta \binom{v}{s}} q^{-1} x_e \right] \\
&= \prod_{e \in \binom{u \cap v}{s}} q^{-2} \E_{x_e \sim \Ber(q)} [x_e] \prod_{e \in \binom{u}{s} \Delta \binom{v}{s}} q^{-1} \E_{x_e \sim \Ber(q)} [x_e] \\
&= q^{-\binom{|u \cap v|}{s}}
\end{align*}
where $A \Delta B$ denotes the symmetric difference of the sets $A$ and $B$. Now since $X = |u \cap v|$ is distributed as $\text{Hypergeometric}(N, K, K)$, we have that
$$\left\| \E_{u \sim \mu} \ol D_u^{\otimes k} \right\|^2 = \E_{u, v \sim \mu} \langle \ol D_u, \ol D_v \rangle^k = \E q^{-k\binom{X}{s}} = \sum_{x = 0}^K \frac{\binom{K}{x} \binom{N - K}{K - x}}{\binom{N}{K}} \cdot q^{-k\binom{x}{s}}$$
Now note that for each $0 \le x \le K$,
\begin{align*}
\frac{\binom{K}{x} \binom{N - K}{K - x}}{\binom{N}{K}} &= \frac{\binom{K}{x} K(K-1) \cdots (K - x + 1)}{N^x \prod_{i = 0}^{x - 1} \left( 1 - \frac{i}{N} \right) \prod_{i = 0}^{K - x - 1} \left( 1 - \frac{K - x}{N - k - i} \right)} \\
&\le \frac{K^{2x}}{N^x \left( 1 - \sum_{i = 0}^{x - 1} \frac{i}{N} - \sum_{i = 0}^{K - x - 1} \frac{K - x}{N - k - i} \right)} \\
&\le \frac{K^{2x}}{N^x \left( 1 - \frac{2K^2}{N - 2K + 1}\right)} \le \frac{1}{2} \left( \frac{K^2}{N} \right)^x
\end{align*}
where the last inequality follows from the fact that $K^2 \le 3N$. Now since $q^{-1} \le \exp(\gamma)$ and $\binom{x}{s} \le x K^{s - 1}$ for all $x \le K$, we have that
$$\left\| \E_{u \sim \mu} \ol D_u^{\otimes k} \right\|^2 \le \frac{1}{2} \sum_{x = 0}^K \exp\left( k\gamma x K^{s - 1} - x \log \left( \frac{N}{K^2} \right) \right) \le \frac{1}{2} \sum_{x = 0}^K \left( \frac{K^2}{N} \right)^{x/2} = O_N(1)$$
by the given condition on $\gamma$. This completes the proof of the claim.
\end{proof}

\subsection{Spiked Wishart PCA}\label{sec:wishart-app}
\restatelemma{lem:wishart-low-deg}
\begin{proof}
Fix any multi-index $\alpha = (\alpha_1, \ldots, \alpha_t)$ so that $|\alpha_i|$ is even and so that $2 \leq |\alpha_i| \leq d$, for all $i = 1, \ldots, t$. 
Suppose moreover that $\left| \{j: \exists i: \alpha_{ij} \neq 0 \} \right| = \ell$, and let $s = |\alpha|$.
Then the proceeding lemma implies that
\begin{align*}
\left( \E_{u \sim S_\rho} \langle \ol D_u, H_\alpha \rangle \right)^2 &\leq \left( \frac{d \lambda}{\rho n} \right)^s \rho^{2 \ell} \; .
\end{align*}
The total number of such monomials can be naively upper bounded by $\binom{n}{\ell} \ell^s$.
Hence the contribution to the LDLR of all such monomials, for a fixed $\ell$ and $s$, can be upper bounded by
\[
\binom{n}{\ell} \ell^s \left( \frac{d \lambda}{\rho n} \right)^s \rho^{2 \ell} \leq \left( \frac{d \ell \lambda}{\rho n} \right)^s (n \rho^2)^{\ell} \leq \left( \frac{d \ell \lambda}{\rho n} \right)^s \; , 
\]
by assumption.
Summing over all $2t \leq s \leq dt$, and $1 \leq \ell \leq dt$, we obtain that 
\[
\left\| \E_{u \sim S_\rho} (\ol D_u^{\leq d} - 1)^{\otimes t} \right\|^2 \leq \sum_{2t \leq s \leq dt, 1 \leq \ell \leq dt} \left( \frac{d \ell \lambda}{\rho n} \right)^s \leq 2 \left( \frac{d^2 k \lambda}{\rho n} \right)^{2t} \; ,
\]
since from our assumptions, the sum is convergent.
\end{proof}

\restatelemma{cor:wishart-high-deg}
\begin{proof}
This proof closely resembles the proof of Lemma~\ref{lem:gauss}.
Let $Z$ be the random variable given by $Z = \frac{\lambda^2 \langle u, v \rangle^2}{4}$ when $u, v \sim S_\rho$.
From the proceeding lemma, we have that
\begin{align*}
\left\| \E_{u \sim S_\rho} \left( \ol D_u^{> d} \right)^{\otimes k}  \right\|^2 &\leq \E_{Z} \left[ \phi^{> d / 2} \left( Z \right)^k \right] 
\end{align*}
By Taylor's theorem, since the function $\phi(x)$ is analytic for all $|x| \leq 1/4$, we know that 
\[
\left| \phi^{> d / 2 }(x) \right| \leq \binom{d + 2}{d/2 + 1} x^{d + 1} \left( 1 - 4 \eta(x) \right)^{-(d + 3)/2} \leq \binom{d + 2}{d/2 + 1} x^{d + 1} \phi(x)^{d + 3} \; .
\]
where $0 \leq \eta(x) \leq x$, and the last inequality follows since $\phi$ is monotone.
Hence
\[
\left\| \E_{u \sim S_\rho} \left( \ol D_u^{> d} \right)^{\otimes k}  \right\|^2 \leq d^{k(d + 2)} \left( 1 - 4 \lambda^2 \right)^{k(d + 3)} \E_{Z} Z^{k(d+1)} \; .
\]
The moment can only be increased by considering the inner product between the two untruncated vectors.
Let $Z'$ be distributed as the untruncated version of $Z$.
Then $Z' = \frac{\lambda^2}{4 \rho n} \left( \sum_{i = 1}^n Y_i \right)^2$ where each $Y_i$ is independent, $Y_i = 0$ with probability $1 - \rho^2 / 2$, $Y_i = 1$ with probability $\rho^2  / 4$, and $Y_i = -1$ with probability $\rho^2 / 4$.
Hence
\begin{align*}
\E_Z Z^{k(d + 1)} \leq \E_{Z'} (Z')^{k(d + 1)} &= \left( \frac{\lambda^2}{4 \rho n} \right)^{k(d + 1)} \E_{Y_1, \ldots, Y_n} \left( \sum_{i = 1}^n Y_i \right)^{2k(d + 1)} \\
&= \left( \frac{\lambda^2}{4 \rho n} \right)^{k(d + 1)} \sum_{|\alpha| = 2k(d + 1)} \E Y^\alpha \\
&\leq \left( \frac{\lambda^2}{4 \rho n} \right)^{k(d + 1)} \left( \sum_{\ell = 1}^{k(d + 1)} \binom{n}{\ell} \cdot \binom{k(d + 1) + \ell}{\ell} \rho^{2\ell} \right) \\
&\leq \left( \frac{\lambda^2}{4 \rho n} \right)^{k(d + 1)} \sum_{\ell = 1}^{k(d + 1)} \left( 2 n k (d + 1) \rho^2 \right)^\ell \leq \left( \frac{\lambda^2}{4 \rho n} \right)^{k(d + 1)} \; ,
\end{align*}
where the final summand is convergent by assumption.
\end{proof}

\subsection{Gaussian Graphical Models}\label{sec:GGM-app}
\restatelemma{lem:ggm}

To prove this lemma, we will make use of the following claim:
\begin{claim}\label{claim:corform}
Let $A,B$ be symmetric $n \times n$ real matrices, let $D_{\emptyset} = \calN(0,\Id)$.
Suppose $\Id_n + A + B \succ 0$, $\Id_n + A \succ 0$, and $\Id_n + B \succeq 0$.
Let $D_a = \calN(0, (\Id + A)^{-1})$ and $D_b = \calN(0,(\Id + B)^{-1})$, and let $\ol D_a, \ol D_b$ be the respective relative densities.
Then
\[
\langle \ol D_a, \ol D_b \rangle_{D_\emptyset} = 
\frac{1}{\sqrt{\det\left(\Id - (\Id + A)^{-1}AB(\Id+B)^{-1}\right)}}.
\]
\end{claim}
\begin{proof}
We have that
\begin{align*}
\langle \ol D_a, \ol D_b \rangle
&= \frac{1}{\sqrt{(2\pi)^{n}\det((\Id + A)^{-1})\det((\Id + B)^{-1})}}\int_{\R^n} \exp\left(-\frac{1}{2}x^\top \left(\Id_n + A + B\right)x\right)dx\\
&= \sqrt{\frac{\det((\Id + A + B)^{-1})}{\det((\Id + A)^{-1})\det((\Id + B)^{-1})}}\\
&= \frac{1}{\sqrt{\det\left(\Id - (\Id + A)^{-1}AB(\Id+B)^{-1}\right)}},
\end{align*}
where the second line follows by integrating the Gaussian pdf with covariance $(\Id + A + B)^{-1}$, and the third line follows by noting that $\det(X^{-1}) = \det(X)^{-1}$, that $\det(X) \det(Y) = \det(XY)$, and that $\Id + A + B = (\Id+A)(\Id+B) - AB$.
 This completes the proof.
\end{proof}

\begin{proof}[Proof of Lemma~\ref{lem:ggm}]
First, since a random signed $d$-regular graph on $s$ vertices has its spectrum within $[-2\sqrt{d-1}(1+\eps), 2\sqrt{d-1}(1+\eps)]$ with high probability, for sufficiently large $d$ the condition on the spectrum is met with very high probability, and $S$ has size at least $\binom{n}{s} \cdot \binom{s}{d}^{s/100}$ (a vast underestimate of the number of $d$-regular random graphs on $s$ vertices planted within $n$-vertex empty graphs).

Since $\kappa 2\sqrt{d} < \frac{1}{3}$, the matrices $\Id + \kappa \Delta_u$ and $\Id + \kappa \Delta_u + \kappa \Delta_v$ meet the conditions of Claim~\ref{claim:corform}.
Using Claim~\ref{claim:corform}, it suffices to bound
\begin{align}
\E_{u,v \sim S} \left(\langle \ol D_u, \ol D_v -1 \rangle\right)^k
= \E_{u,v \sim S} \left(\sqrt{\frac{1}{\det(\Id - \kappa^2 (\Id + \kappa \Delta_u)^{-1}\Delta_u\Delta_v(\Id+\kappa \Delta_v)^{-1})}} - 1\right)^k,\label{eq:det}
\end{align}
since to obtain the $\SDA$ bound we may apply Equation~(\ref{eq:moment-to-cor}), and to get the second conclusion we use H\"{o}lder's inequality and the triangle inequality,
\[
\E_{u,v\sim S} \langle \ol D_u, \ol D_v \rangle^k
\le \sum_{\ell = 0}^k \binom{k}{\ell} \E_{u,v} \left[|\langle \ol D_u,\ol D_v \rangle -1|^{\ell}\right]
\le \left(1 + \E_{u,v} \left[(\langle \ol D_u, \ol D_v \rangle - 1)^k\right]^{1/k}\right)^{k},
\]

Now, when $u,v \sim S$, with probability at least $1-\frac{s^2}{n}$, $\Delta_u$ and $\Delta_v$ correspond to graphs with disjoint support, so $\Delta_u\Delta_v = 0$.
For such $u,v$, the right-hand side of (\ref{eq:det}) is zero.

Otherwise, if $\Delta_u, \Delta_v$ overlap, the $(\Id + \kappa \Delta_u)^{-1}\Delta_u \Delta_v(\Id + \kappa \Delta_v)^{-1}$ has at most $s$ eigenvalues which are not $1$ (since $\Delta_u, \Delta_v$ are rank-$s$).
Further, since all eigenvalues $\Delta_u, \Delta_v$ are in the interval $[-2\sqrt{d}, 2\sqrt{d}]$, and since $\Delta_u$ and $(\Id+\kappa\Delta_u)^{-1}$ commute, the eigenvalues of $(\Id + \kappa \Delta_u)^{-1}\Delta_u, \Delta_v(\Id + \kappa \Delta_v)^{-1}$ are in the interval $[-\frac{2\sqrt{d}}{1-\kappa 2\sqrt{d}},\frac{2\sqrt{d}}{1-\kappa 2\sqrt{d}}]$.
This implies that all eigenvalues of $(\Id + \kappa \Delta_u)^{-1}\Delta_u \Delta_v(\Id + \kappa \Delta_v)^{-1}$ are in the interval $[- \frac{4d}{(1-\kappa 2\sqrt{d})^2}, \frac{4d}{(1-\kappa 2\sqrt{d})^2}]$.
Thus, for such $u,v$,
\[
\sqrt{\frac{1}{\det(\Id - \kappa^2 (\Id + \kappa \Delta_u)^{-1}\Delta_u\Delta_v(\Id+\kappa \Delta_v)^{-1})}} \le \left(\frac{1}{1 - \frac{\kappa^2 d}{(1-\kappa 2\sqrt{d})^2}}\right)^{s/2}.
\]
Putting these observations together with (\ref{eq:det}),
\[
\E_{u,v} \left(\langle \ol D_u, \ol D_v \rangle -1\right)^k
\le\frac{s^2}{n}\left(\left(\frac{1}{1-d\left(\frac{\kappa}{1-\kappa 2\sqrt{d}}\right)^2 }\right)^{s/2} - 1\right)^{k}
\le\frac{s^2}{n}\left(\left(1+\kappa^2 d\right)^{s/2} - 1\right)^{k},
\]
where we have used that $\kappa\sqrt{d} < \frac{1}{6}$.
We can further simplify the above by noting that $1+x \le \exp(x)$.

Thus, applying Equation~(\ref{eq:moment-to-cor}), we have that for any $q \ge 1$,
\[
\SDA\left(\calS, \left(\frac{n}{q^2 s^2}\right)^{1/k}\frac{1}{\exp( sd\kappa^2/2)-1}\right) \ge q,
\]
and we obtain the bound on $\|\E_u \ol D_u^{\otimes k}\|$ using H\"{o}lder's as described above.
\end{proof}

\end{document}